\documentclass[11pt,DIV=9]{scrartcl}
\usepackage[utf8]{inputenc}

\usepackage{amsfonts,amsthm,amsmath,amssymb,mathtools}
\usepackage{dsfont,newtxtext,newtxmath,microtype}
\usepackage{booktabs,multirow}
\usepackage[section]{placeins}
\usepackage{fullpage}
\input{braket.tex}

\usepackage{epsfig}
\usepackage{graphicx}
\graphicspath{{Images/}}
\usepackage[dvipsnames]{xcolor}
\usepackage{tikz}
\usetikzlibrary{quantikz}

\definecolor{color3}{RGB}{140,29,49}
\definecolor{color2}{RGB}{31,107,76}
\definecolor{color1}{RGB}{29,66,91}
\usepackage[colorlinks=true,linkcolor=color1,urlcolor=color2,citecolor=color3,anchorcolor=green,pdfusetitle]{hyperref}

\usepackage{authblk}

\usepackage[capitalize]{cleveref}
\crefname{lemma}{Lemma}{Lemmas}
\crefname{proposition}{Proposition}{Propositions}
\crefname{definition}{Definition}{Definitions}
\crefname{theorem}{Theorem}{Theorems}
\crefname{conjecture}{Conjecture}{Conjectures}
\crefname{corollary}{Corollary}{Corollaries}
\crefname{example}{Example}{Examples}
\crefname{section}{Section}{Sections}
\crefname{appendix}{Appendix}{Appendices}
\crefname{figure}{Fig.}{Figs.}
\crefname{equation}{eq.}{eqs.}
\crefname{table}{Table}{Tables}
\crefname{item}{Property}{Properties}
\crefname{remark}{Remark}{Remarks}
\crefname{problem}{}{}

\newtheorem{theorem}{Theorem}
\newtheorem{definition}[theorem]{Definition}
\newtheorem{corollary}[theorem]{Corollary}

\newtheorem{lemma}[theorem]{Lemma}

\usepackage{ltablex}

\makeatletter
\newcommand{\probleminput}[1]{\gdef\@probleminput{#1}}
\newcommand{\problemquestion}[1]{\gdef\@problemquestion{#1}}
\newcommand{\problempromise}[1]{\gdef\@problempromise{#1}}
\DeclareDocumentEnvironment{problem}{}{
\probleminput{}\problempromise{}\problemquestion{}
\leavevmode
}{
  \par\addvspace{.5\baselineskip}
  \noindent
  \begin{itemize}
      \item[\textbf{Input:}] \@probleminput
\ifx\@problempromise\empty%
\else%
      \item[\textbf{Promise:}] \@problempromise
\fi%
      \item[\textbf{Output:}] \@problemquestion
  \end{itemize}
}
\makeatother

\usepackage{pifont}
\newcommand{\cmark}{\text{\ding{51}}}%
\newcommand{\xmark}{\text{\ding{55}}}%

\def\R{{\mathds{R}}}
\def\N{{\mathds{N}}}
\def\C{\mathds{C}}
\def\1{{\mathds{1}}}

\newcommand{\mysymbol}[1]{{\mbox{\raisebox{-0.3em}{\epsfysize=1.2em\epsfbox{#1}}}}}

\newcommand{\leftend}{\mysymbol{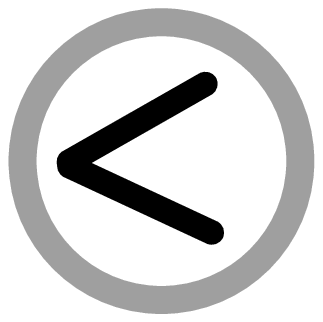}}
\newcommand{\rightend}{\mysymbol{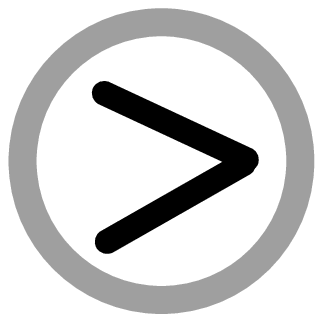}}

\newcommand{\class}[1]{\textup{#1}\xspace}
\newcommand{\QMA}{\class{QMA}}

\newcommand{\coQMA}{\class{co-QMA}}

\newcommand{\PSPACE}{\class{PSPACE}}

\newcommand{\BQP}{\class{BQP}}
\newcommand{\QCMA}{\class{QCMA}}
\newcommand\QCMAEXP{\class{QCMA\textsubscript{EXP}}}
\newcommand\PreciseQMA{\class{PreciseQMA}}

\newcommand{\EXP}{\class{EXP}}
\newcommand{\NEXP}{\class{NEXP}}

\newcommand\sparallel{\ensuremath{\parallel}}

\newcommand\PP{\class{P}}
\newcommand\NP{\class{NP}}

\newcommand\PQMA{\PP\unskip\textsuperscript{\QMA}\xspace}
\newcommand\PQMApar{\PP\unskip\textsuperscript{\sparallel\QMA}\xspace}
\newcommand\PQMAexp{\PP\unskip\textsuperscript{\QMAEXP}\xspace}
\newcommand\PQMAlog{\PP\unskip\textsuperscript{\QMA\unskip[\ensuremath{\log}]}\xspace}
\newcommand\QMAEXP{\QMA\unskip\textsubscript{\EXP}\xspace}
\newcommand{\EXPpar}{\EXP\unskip\textsuperscript{\sparallel\QMA}\xspace}

\newcommand{\DQ}{\class{D}\textsuperscript{\sparallel\class{Q}}}
\newcommand\StoqMA{\class{StoqMA}}

\newcommand{\GSCON}{\class{GSCON}}
\newcommand{\TIGSCON}{\class{TI-\GSCON}}

\newcommand{\set}[1]{{\left\{#1\right\}}}    
\newcommand{\spa}[1]{\mathcal{#1}}
\newcommand{\sA}{\spa{A}}
\newcommand{\sB}{\spa{B}}

\newcommand{\hout}{H_{\textup{out}}}
\newcommand{\trace}{\tr}

\newcommand{\lowstate}{\ket{\psilow}}
\newcommand{\lowstatebra}{\bra{\psilow}}

\newcommand{\base}{(\C^2)}
\newcommand{\aaa}{\eta_1}
\newcommand{\bbb}{\eta_2}
\newcommand{\ccc}{\eta_3}
\newcommand{\ddd}{\eta_4}

\newcommand{\abs}[1]{\left\lvert #1\right\rvert}
\newcommand{\norm}[1]{\left\lVert #1 \right\rVert}
\newcommand{\tr}{\operatorname{Tr}}

\newcommand{\setft}[1]{\mathrm{#1}}
\newcommand{\lin}[1]{\setft{L}\left(#1\right)}

\newcommand{\unitary}[1]{\setft{U}\left(#1\right)}
\newcommand{\herm}[1]{\setft{Herm}\left(#1\right)}
\newcommand{\pos}[1]{\setft{Pos}\left(#1\right)}

\newcommand{\ayes}{A_{\textup{yes}}} 
\newcommand{\ano}{A_{\textup{no}}} 
\newcommand{\ainv}{A_{\textup{inv}}} 
\newcommand{\Null}[1]{\operatorname{Null}\left(#1\right)}

\newcommand{\Span}{\mathrm{Span}}

\newcommand{\enorm}[1]{\norm{#1}_{\mathrm{2}}}      
\newcommand{\trnorm}[1]{\norm{#1}_{\mathrm {tr}}}  
\newcommand{\snorm}[1]{\norm{#1}_{\mathrm {\infty}}}    

\newcommand{\p}{b}

\newtheorem{observation}{Observation}[section]

\newenvironment{mylist}[1]{\begin{list}{}{
	\setlength{\leftmargin}{#1}
	\setlength{\rightmargin}{0mm}
	\setlength{\labelsep}{2mm}
	\setlength{\labelwidth}{8mm}
	\setlength{\itemsep}{0mm}}}
	{\end{list}}

\usepackage{xparse,xspace}

\newsavebox\TBox

\DeclareMathOperator{\poly}{poly}

\DeclareMathOperator{\Tr}{Tr}
\DeclareMathOperator{\BigO}{O}

\newcommand\field\mathds
\newcommand\YES{{\normalfont{\textsc{Yes}}}\xspace}
\newcommand\NO{{\normalfont{\textsc{No}}}\xspace}

\newcommand\Hprop{H_\mathrm{prop}}

\DeclareMathOperator{\enc}{enc}

\newcommand{\Pin}{P_{\textup{in}}}
\newcommand{\nats}{{\mathbb N}}
\newcommand{\reals}{{\mathbb R}}

\newcommand{\qout}{q_{\textup{out}}}
\newcommand{\qflag}{q_{\textup{flag}}}
\newcommand{\ang}{\sqrt{3}/(2m)}
\newcommand{\angf}{\frac{\sqrt{3}}{2m}}
\newcommand{\angft}{\frac{\sqrt{3}}{m}}
\newcommand{\angftt}{\frac{3}{2m^2}}
\newcommand{\HW}{\textup{HW}}
\newcommand{\pyw}{p_{y,w}}

\newcommand{\qvc}{QVClass}

\newcommand{\spre}{S_{\textup{pre}}}
\newcommand{\spost}{S_{\textup{post}}}
\newcommand{\f}{f} 
\newcommand{\g}{g} 
\newcommand{\psiin}{\psi_{\textup{in}}}
\newcommand{\psiout}{\psi_{\textup{out}}}
\newcommand{\phiin}{\phi_{\textup{in}}}

\newcommand{\wm}{w_{1\cdots m}}

\newcommand{\prs}{\textstyle\Pr^*}
\newcommand{\sno}{S_{\textup{no}}}
\newcommand{\syes}{S_{\textup{yes}}}
\newcommand{\sinv}{S_{\textup{inv}}}
\newcommand{\pgood}{p_{\textup{good}}}
\newcommand{\hin}{H_{\textup{in}}}
\newcommand{\hprop}{H_{\textup{prop}}}
\newcommand{\hstab}{H_{\textup{stab}}}
\newcommand{\polylog}{\textup{polylog}}
\newcommand{\PNPpar}{\PP^{\parallel\NP}}
\newcommand{\PNPlog}{\PP^{\NP[\log]}}
\newcommand\Hw{H_\mathrm{w}}
\newcommand{\psilow}{\psi_{\textup{low}}}

\makeatletter
\newtheorem*{rep@theorem}{\rep@title}
\newcommand{\newreptheorem}[2]{%
	\newenvironment{rep#1}[1]{%
		\def\rep@title{#2 \ref{##1}}%
		\begin{rep@theorem}}%
		{\end{rep@theorem}}}
\makeatother

\newreptheorem{theorem}{Theorem}
\newreptheorem{lemma}{Lemma}

\DeclareMathOperator{\lmin}{\lambda_\mathrm{min}}

\def\R{{\mathds{R}}}
\def\N{{\mathds{N}}}
\def\C{\mathds{C}}
\def\1{{\mathds{1}}}

\usepackage[backend=biber,style=alphabetic,doi=false,isbn=false,url=false,maxbibnames=20,maxcitenames=2]{biblatex}

\bibliography{references}
\newbibmacro{string+doi}[1]{\iffieldundef{doi}{#1}{\href{https://dx.doi.org/\thefield{doi}}{#1}}}
\DeclareFieldFormat{title}{\usebibmacro{string+doi}{\mkbibemph{#1}}}
\DeclareFieldFormat[article]{title}{\usebibmacro{string+doi}{\mkbibquote{#1}}}
\DeclareFieldFormat[incollection]{title}{\usebibmacro{string+doi}{\mkbibquote{#1}}}

\begin{document}
	\title{The Complexity of Translationally Invariant Problems beyond Ground State Energies}
	\author[1]{James D. Watson}
	\author[2]{Johannes Bausch}
    \author[3]{Sevag Gharibian}
	\affil[1]{Department of Computer Science, University College London, UK}
	\affil[2]{CQIF, DAMTP, University of Cambridge, UK}
    \affil[3]{Department of Computer Science, Paderborn University, Germany}
    \date{December 2020}
	\maketitle

\begin{abstract}
  It is known that three fundamental questions regarding local Hamiltonians---approximating the ground state energy (the Local Hamiltonian problem), simulating local measurements on the ground space (APX-SIM), and deciding if the low energy space has an energy barrier (\GSCON)---are \QMA-hard, $\PQMAlog$-hard and \QCMA-hard, respectively, meaning they are likely intractable even on a quantum computer.
  Yet while hardness for the Local Hamiltonian problem is known to hold even for translationally-invariant systems, it is not yet known whether APX-SIM and \GSCON remain hard in such ``simple'' systems. In this work, we show that the translationally invariant versions of both APX-SIM and \GSCON remain intractable, namely are $\PP^{\QMAEXP}$- and $\QCMAEXP$-complete, respectively.

  Each of these results is attained by giving a respective generic ``lifting theorem'' for producing hardness results. For APX-SIM, for example, we give a framework for ``lifting'' any abstract local circuit-to-Hamiltonian mapping $H$ (satisfying mild assumptions) to hardness of APX-SIM on the family of Hamiltonians produced by $H$, while preserving the structural and geometric properties of $H$ (e.g.\ translation invariance, geometry, locality, etc). Each result also leverages counterintuitive properties of our constructions: for APX-SIM, we ``compress'' the answers to polynomially many parallel queries to a QMA oracle into a single qubit. For GSCON, we give a hardness construction robust against highly non-local unitaries, i.e.\ even if the adversary acts on all but one qudit in the system in each step.
\end{abstract}
\thispagestyle{empty}

\clearpage

\tableofcontents

\section{Introduction}
The Hamiltonian operator of a quantum system describes its static and dynamic properties, such as its ground state energy, spectral gap, and time-evolution. Yet, Hamiltonians for many systems of interest are often too complicated to be solvable analytically; obtaining all of their energy levels, for instance, is simply intractable.

The study of ``how complex'' a Hamiltonian is---i.e.\ how easy it is to extract information about the system it describes---has led to the field of Hamiltonian complexity theory (see, e.g.,~\cite{O11,Boo14,GHLS14}). One of the central questions one can ask about a many-body system is that of approximating its ground state energy, i.e.\ the energy level the system settles into when cooled to low temperature. For local Hamiltonians, this is dubbed the \emph{Local Hamiltonian Problem (LH)}. In 2002, Kitaev proved LH is \QMA-hard \cite{Kitaev2002}, meaning---under widely-accepted complexity theoretic assumptions---that LH remains difficult, even with a universal and error-corrected quantum computer in hand. Through a series of subsequent results it was shown that LH remains hard on ever simpler systems, such as if one restricts to qubits on a lattice with 2-local interactions \cite{Oliveira2008}, Hamiltonians on a 1D spin chain with local dimension 8 \cite{Aharonov2009,N08,HNN13}, and even for 1D translationally invariant, nearest neighbour Hamiltonians (albeit with a large local Hilbert space dimension) \cite{Gottesman2009}. A complexity classification for all 2-local qubit Hamiltonians was given by \cite{Cubitt_Montanaro_2013}.

\paragraph{Beyond ground state energies.} Since Kitaev's seminal work on the ground state energy problem, further studies have explored tractability of other important Hamiltonian properties. These include estimating excitation energies of local Hamiltonians \cite{Jordan_Gosset_Love_2010}; determining a system's density of states \cite{Brown_Flammia_Schuch_2011,SZ}; minimizing interaction terms yielding frustrated ground spaces~\cite{GK12}; deciding if a ground space has an energy barrier~\cite{GS14,GMV17}; simulating local measurements on ground spaces~\cite{Ambainis2013,GY16_2,GPY20}; Hamiltonian sparsification~\cite{AZ18}; estimating spectral gaps of local Hamiltonians (in both finite-size and infinite-size settings)~\cite{Ambainis2013,Cubitt_Perez-Garcia_Wolf2015,GY16_2}; estimating the partition function or free energy~\cite{K17,HMS20,KKB20}; and ``universal'' Hamiltonian models simulating the physics of all other quantum many-body systems~\cite{BH17,CMP,Piddock2020,Kohler2020}.

Of these, we focus on two problems here: \emph{Simulating local measurements on ground spaces}~\cite{Ambainis2013} and \emph{deciding if a ground space has an energy barrier}~\cite{GS14}. The first problem, denoted Approximate Simulation (APX-SIM), asks how difficult it is to estimate the expectation of a local measurement against the ground state of a local Hamiltonian. Given that much of condensed matter physics is devoted to determining the low-energy properties of materials, and that local measurements are the only tools available to experimentalists to examine these systems, this is an extremely important problem. As discussed shortly, APX-SIM turns out to be even harder than QMA, being $\PQMA$-complete.

The second problem is the Ground State Connectivity problem (GSCON)~\cite{GS14}, which roughly asks: given two ground states $\ket{\psi}$ and $\ket{\phi}$ of a local Hamiltonian, is there a low energy path connecting $\ket{\psi}$ to $\ket{\phi}$? Informally, this captures the question of determining if a ground space has an energy barrier, and has applications to quantum stabilizer codes and quantum memories. Similar to APX-SIM, GSCON is intractable, being \QCMA-complete (QCMA is QMA but with a classical proof).\\

\vspace{-2mm}
\noindent \emph{When are these questions relevant to, say, material sciences?} The more realistic a system the Hamiltonian describes while retaining hard-to-compute properties, the more insight about the intrinsic complexity of the system one gains. For instance, a local spin dimension of $2$ is often seen in nature (e.g.\ an electron spin up/down); a spin dimension of $2^{12}$ less so.
Furthermore, condensed matter systems in real life often feature symmetries, such as a regular lattice structure with nearest-neighbour couplings that are both isotropic and translationally invariant. A goal of Hamiltonian complexity theory is thus to find increasingly ``simple'' systems that retain hard-to-compute properties.
This renders claims more generic, and the resulting implications stronger.

Beyond their relevance for real-world systems, translationally-invariant Hamiltonians in particular are widely believed to be simpler than general Hamiltonians. Intuitively, due to the spatial invariance of the system, the degrees of freedom available to encode complex behaviour are limited; and less information can be encoded into the couplings throughout the system (assuming they are specified to the same precision).

As for spatial structure, the most basic lattice model is the one-dimensional spin chain, for which any hardness results often immediately imply respective hardness results for two- or higher-dimensional systems (by simply repeating the system along the extra dimensions).
And just like translational symmetry, having only one dimension often renders systems more tractable, which is supported by empirical and theoretical evidence \cite{lieb1961a,affleck1987,franchini2017}.
For instance, algorithms like DMRG \cite{white1992} to approximate ground state energies have long been known to work well in practice in 1D. Indeed, this eventually led to rigorous polynomial-time algorithms to solve ground state energy problems for gapped one-dimensional spin chains \cite{Landau2013}. Beyond that, there often exist closed form solutions in 1D, such as for fermionic 1D systems described by the Fermi-Hubbard model or 1D Heisenberg model~\cite{B31}; similar systems are notoriously difficult to simulate in higher dimensions.

So are one-dimensional, nearest neighbour, translationally-invariant (TI) systems tractable, or at least ``more tractable'' than their higher-dimensional counterparts?
Alas, results such as \cite{Gottesman2009} show that finding the ground state of a 1D spin chain is \QMAEXP-complete ($\QMAEXP$ is a quantum analogue of NEXP), and the question of existence of a spectral gap in one spatial dimension---even for couplings with translational invariance---remains undecidable \cite{Bausch2018b}.
Yet for other natural questions, such as APX-SIM or \GSCON, the verdict is still open.

\paragraph{Our results.}  \emph{In this work, we show that both APX-SIM and \GSCON remain hard even for one dimensional, nearest neighbour, translationally invariant systems.} The exact notion of ``hard'' depends on the specific problem and context considered. Specifically, recall that the usual Local Hamiltonian problem (LH), which is QMA-complete with an inverse polynomial promise gap~\cite{KSV02}, becomes \QMAEXP-complete in the 1D translationally invariant (1D TI) setting~\cite{Gottesman2009}. Intuitively, this is because a full 1D TI system on $N$ qubits can typically be specified using just $\polylog(N)$ bits---the description of the single $2$-local interaction term $H_{i,i+1}$ to be repeated along the chain, and the length $N$ of the chain (in binary).
With exponential precision, the LH problem for 1D TI problem becomes \PSPACE-complete \cite{Kohler2020}.

Here we see a similar progression: While APX-SIM is known to be \PQMAlog-complete in the general case and for inverse polynomial measurement precision (in the length of the chain)~\cite{Ambainis2013,GY16_2,GPY20}, we show it becomes \PQMAexp-complete for the same precision (\cref{th-intro:apx-sim-poly}), and \PSPACE-complete for exponential precision (\cref{th-intro:apx-sim-precise}) in the 1D translationally-invariant setting. Likewise, while \GSCON is \QCMA-complete in the general case \cite{GS14} (even for commuting Hamiltonians~\cite{GMV17}), we show it becomes \QCMAEXP-complete in the 1D, nearest neighbour, translationally-invariant case.

As discussed below, a theme of this paper is to give generic ``lifting theorems'', which take broad classes of circuit-to-Hamiltonian mappings and ``lift them'' in a black-box fashion to obtain hardness results for APX-SIM or GSCON. For APX-SIM, in particular, the framework we give works for almost arbitrary local circuit-to-Hamiltonian construction $H$ (formally defined via \cref{def:local-mapping}, which has minimal assumptions), so that the hardness result we obtain via ``lifting'' for APX-SIM automatically satisfies the structural and geometric properties of $H$. We now state our results somewhat more formally, followed by a discussion of techniques.

\subsection{Hardness of Measuring Local Observables}
The problem of estimating expectation values of local observables in the low-energy subspace of a Hamiltonian is formalized via the Approximate Simulation (APX-SIM) problem of Ambainis~\cite{Ambainis2013}.
\begin{definition}[APX-SIM$(H,A,k,l,a,b,\delta)$~\cite{Ambainis2013}]\label{def:apx}
		Given a $k$-local Hamiltonian $H=\sum_i H_i $ acting on $N$ qubits, an $l$-local observable $A$, and real numbers $a$, $b$, and $\delta$ such that $b-a\geq N^{-c}$ and $\delta\geq N^{-c'}$, for $c,c'>0$ constant, decide:
	\begin{itemize}
		\item[YES.] If $H$ has a ground state $\ket{\psi}$ satisfying $\bra{\psi}A\ket{\psi}\leq a$.
		\item[NO.] If for all $\ket{\psi}$ satisfying $\bra{\psi}H\ket{\psi}\leq \lmin(H)+\delta$, it holds that $\bra{\psi}A\ket{\psi}\geq b$.
	\end{itemize}
\end{definition}
\noindent Introduced in \cite{Ambainis2013}, APX-SIM was shown to be $\PQMAlog$-complete for $\BigO(1)$-local Hamiltonians with $1$-local measurements (with inverse polynomial precision) in~\cite{Ambainis2013,GY16_2}. Here, $\PQMAlog$ is the class of decision problems decidable by a polynomial time (deterministic) Turing machine with access to logarithmically many adaptive queries to a QMA oracle. Note that by definition, $\QMA\subseteq \PQMAlog$, and this containment is believed to be strict, since $\PQMAlog$ also contains\footnote{To decide a co-QMA instance, the $\PQMAlog$ machine queries its QMA oracle once, and then negates the oracle's answer. It is believed that $\QMA\neq \textup{co-QMA}$.} co-QMA. Thus, $\PQMAlog$ is believed \emph{harder} than QMA. It is known that $\PQMAlog\subseteq\class{PP}$~\cite{GY16_2}, and that $\PQMAlog=\PQMApar$~\cite{GPY20}, where $\PQMApar$ is defined as $\PQMAlog$ except where the P machine makes \emph{polynomially many} queries \emph{in parallel} (i.e.\ the queries must be non-adaptive; this setup will be used in this work as well). Finally, \cite{GPY20} showed $\PQMAlog$-completeness for APX-SIM on physically motivated 2D models, such as for the Heisenberg interaction, and on (non-translationally invariant) 1D chains.

In this work, motivated by the goal of showing hardness of APX-SIM in the even simpler setting of 1D TI systems, we give a much more general framework for ``lifting'' hardness results about ground state energies (i.e.\ for LH) to hardness results for APX-SIM. Formally, these are given via the Lifting Lemma (\cref{lem:lift}) and applications in \cref{sscn:other}; here, we informally state the general premise as follows.
\begin{theorem}[LH to APX-SIM (informal)]
	If the family of Hamiltonians $\mathcal{F}$ admits a circuit-to-Hamiltonian mapping such that approximating the ground state energy is $\textup{C}$-hard, then the APX-SIM problem for $\mathcal{F}$ is either $\PP^{\textup{C[log]}}$ or $\PP^{\textup{C}}$-complete, depending on how the input is encoded.
\end{theorem}
\noindent In contrast to previous approaches for showing hardness for APX-SIM, which were custom-designed based on the circuit-to-Hamiltonian constructions in mind, here we obtain a black-box mapping (\cref{lem:lift}) which requires minimal assumptions, and which automatically \emph{preserves structural properties of $\mathcal{F}$}, such as locality, geometry, translational invariance, etc. This demonstrates that the Local Hamiltonian problem fundamentally characterises the complexity of computing properties (such as simulating measurements) of the low energy states of a family of Hamiltonians.

Importantly, this gives the following set of results, among others.
\begin{theorem}\label{th-intro:apx-sim-poly}
    Let $n$ denote input size, and set $N\in O(\exp (n))$. Then, APX-SIM is \PQMAexp-complete for 1D, nearest neighbour, translationally invariant Hamiltonians on $N$ qudits of local dimension $44$, for $\delta=\Omega(1/\poly(N))$, $b-a =\Omega(1/\poly(N))$, or 3D, 4-local, translationally invariant Hamiltonians on an fcc lattice of local dimension 4, for the same parameters.
\end{theorem}
\begin{theorem}\label{th-intro:apx-sim-precise}
	Let $n$ denote input size, and set $N\in O(\poly (n))$. Then, APX-SIM is \PSPACE-complete for 1D, nearest neighbour, translationally invariant Hamiltonians on $N$ qudits of local dimension $44$ and for a precision of $\delta=\Omega(1/\exp(N))$, $b-a =\Omega(1/\exp(N))$.
\end{theorem}
\noindent Above, \QMAEXP is to \QMA as \NEXP is to \NP, i.e.\ \QMAEXP has an exponentially long proof and verification time relative to the input size (see \cref{Def:QMA_EXP}). Such results are particularly notable given that 1D TI systems are among the simplest physical systems imaginable, and that their classical counter-parts are computationally tractable.

We note that for \cref{th-intro:apx-sim-precise}, the 3D fcc lattice construction as well as \cite{Aharonov2009}'s best non-TI 1D spin chain with local dimension 12 (cf.\ \cref{tab:uwHs}) employ a history state which shuffles the workspace qubits one block per time step.
As this implies that the system size $\sim$ computational time, one cannot encode an $\exp$-time computation in these constructions easily; the \PSPACE-hardness reduction thus does not work without first modifying the circuit-to-Hamiltonian mapping underlying these results.

\subsection{Hardness of Ground State Connectivity}\label{sscn:introGSCON}

The second problem we study is the Ground State Connectivity problem (\GSCON), which roughly asks whether a low energy space has an energy barrier. This problem was introduced in \cite{GS14}, where it was shown to be \QCMA-complete for 5-local Hamiltonians. Subsequently, it was discovered~\cite{GMV17} that \GSCON remains QCMA-hard even for commuting Hamiltonians~\cite{GMV17}; this is in contrast to LH, which is not known to be QMA-hard in the commuting case~\cite{BV05,AE11,S11,AE13_2}. Informally, \GSCON is defined as follows (formal definitions in  \cref{def:GSCON} in \cref{sec:TI_GSCON}):
\begin{definition}[Ground State Connectivity (\GSCON) (informal) \cite{GS14}]\label{Def:GSCON_informal}
	Given as input a local Hamiltonian $H$ and two ground states $\ket{\psi}$ and $\ket{\phi}$ (represented succinctly via quantum circuits) of $H$, as well as parameters $m$ and $\p$, does there exist a sequence of $\p$-qubit unitaries $\left(U_i\right)_{i=1}^m$ such that:	
	\begin{enumerate}
		\item ($\ket{\psi}$ mapped to $\ket{\phi}$) $U_m\cdots U_1 \ket{\psi}\approx \ket{\phi}$, and
		\item (intermediate states have low energy) $\forall i\in [m]$, $U_i\cdots U_1\ket{\psi}$
		{has low energy with respect to $H$.}
	\end{enumerate}
\end{definition}
In terms of motivation, GSCON is connected to quantum memories and stabilizer codes. For example, a Hamiltonian $H$ for a YES instance of \GSCON has a short sequence of {local} unitaries mapping between low energy states $\ket{\psi}$ and $\ket{\phi}$ {through} the low energy space of $H$. In a quantum memory, $\ket{\psi}$ and $\ket{\phi}$ may encode distinct logical states. Since errors in physical systems are often modeled to be local, this suggests such an $H$ might not be a good candidate for a quantum memory---not only might $\ket{\psi}$ easily be inadvertently mapped to $\ket{\phi}$, since there is no energy barrier ``separating'' $\ket{\psi}$ from $\ket{\phi}$  (corrupting the content of the quantum memory), but this corrupting process takes place \emph{completely in the low energy space}, meaning such errors are not easily detectable.

In this work, following our theme of generic lifting theorems, we give a lifting construction (Lifting \cref{lem:liftGSCON}) for GSCON, which allows us to obtain the first hardness result for GSCON in a physically motivated setting:
\begin{theorem}\label{thm:TIGSCON}
    \TIGSCON is \QCMAEXP-complete for 1D, nearest neighbour, translationally invariant Hamiltonians
on N qudits, for $m\in\poly(N)$, $\delta \in \Theta(1/\poly(N))$, and any $\p\in\set{2,\ldots, N-1}$.
\end{theorem}
\noindent Here, $\QCMAEXP$ (\cref{def:QCMAEXP}) is to $\QCMA$ as $\QMAEXP$ is to $\QMA$, i.e.\ one has an exponentially long proof and verification time. We remark that the applicability of our lifting theorem for GSCON is not as wide as that for APX-SIM; details are given next in \cref{sscn:techniques}.

\subsection{Techniques}\label{sscn:techniques}

\paragraph{For APX-SIM.} Instead of the traditional route of picking a concrete complexity class $\DQ$ (such as $\PQMApar$, which recall equals $\PQMAlog$~\cite{GPY20}) and circuit-to-Hamiltonian mapping $\Hw$, and tailor-designing a reduction to show $\DQ$-hardness for Hamiltonians produced by $\Hw$, here we take a much more general approach. We consider general notions of ``deterministic classes'' D (e.g.\ P, EXP; see \cref{def:dTM}) and ``existentially quantified quantum verification classes (QVClasses)'' Q (e.g.\ NP, QCMA, QMA; see \cref{def:qV}). We also define an abstract notion of ``local circuit-to-Hamiltonian mappings'' $\Hw$ (\cref{def:local-mapping}), which requires minimal assumptions: given any unitary $U$ (think of $U$ as, say, a QMA verification circuit), $\Hw(U)$ has a non-empty null space spanned by ``correctly initialized computation history states'', where we assume very little about the structure of such ``history states'', other than a single primary property: in the Hamiltonian picture, there exist measurement operators $M_1$ and $M_2$, such a measurement in the standard basis on the first and second output qubits of $U$ can be simulated by measuring $M_1$ and $M_2$ in the Hamiltonian picture on $\Hw(U)$. (A formal statement is given in \cref{def:local-mapping}, and involves two further minor assumptions regarding how the final timestep of $U$ is encoded; see \cref{eqn:sim3} and \cref{eqn:sim2}.)

With just this basic property of measuring two qubits of the ground state in hand, we give a black box Lifting Lemma (\cref{lem:augmented-circuit}) which takes any $\DQ$ computation and embeds it in an APX-SIM instance, while automatically preserving structural properties of the circuit-to-Hamiltonian mapping $\Hw$ used. As a bonus, this construction only requires a single free parameter, $\alpha$, which needs to be set according to the desired hardness result and the properties of classes D and Q (e.g.\ what is the desired promise gap; see \cref{sscn:other} for applications).  At a high level, the Lifting Lemma works by starting with an idea similar to the 1D non-translationally $\PQMAlog$-hardness result of~\cite{GPY20}. Namely, we replace all parallel oracle calls to Q with explicit \emph{verification circuits} for Q. In contrast to~\cite{GPY20}, we then ``count'' the number of YES Q-queries via a single qubit---each time a Q-verifier outputs YES, we rotate a designated ``flag qubit'' by a fixed amount. We then push this entire ``bootstrapped'' computation through the given abstract circuit-to-Hamiltonian mapping, followed by use of the measurement term $M_2$ to simulate a penalty on the flag qubit. (This is reminiscent of how the concept of ``sifters'' worked in~\cite{GPY20}.) Remarkably, by carefully adjusting the weight on this one flag qubit, it turns out that with high probability, we can force \emph{all} Q-queries to simultaneously be answered correctly. \emph{A priori}, this is perhaps surprising; for example, Holevo's theorem~\cite{H73} roughly says that $n$ qubits cannot transmit more than $n$ bits of information, and yet here we are ``cramming'' (for lack of a better word) many query answers into a single flag qubit, and are still able to ``make some sense'' of the information therein.

This flag qubit construction now allows us to circumvent the nearest neighbour and translationally invariant restrictions (since there is only a \emph{single} flag qubit to keep track of, we are not worried about how it is arranged geometrically within the final system). Additionally, we remark that a key part of the soundness analysis is an exchange argument (\cref{lem:lowhelp}), which may be of independent interest: given a joint entangled proof $\ket{\wm}$ to $m$ Q-verifiers $V_i$, if the $i$\textsuperscript{th} local component of $\ket{\wm}$ is $\epsilon$-suboptimal for verifier $V_i$, we give a rigorous lower bound on the deviation from the optimal ``counted sum'' on the flag qubit.
We note it is well-known how to argue that a tensor product proof is optimal in such settings (e.g.~\cite{Aharanov_Naveh_2002}); the key here is we give a robust deviation bound from optimality.

We conclude this discussion on APX-SIM by highlighting three distinct advantages of this approach over traditional constructions used to date:
\begin{enumerate}
	\item No modifications to existing Hamiltonian constructions is necessary; the reduction happens at a circuit level.
	This means that \emph{almost all} past local Hamiltonian hardness results---of which there are many, see \cref{tab:uwHs}---are shown to have variants of a hard APX-SIM problem associated to them.
	\item We do not employ additional perturbation gadgets for our reduction. As our result is not based on simulating a more complex system, we do not require constraints on how computation is embedded that e.g.\ preserves locality (a \emph{local simulation}); our demands on the host Hamiltonian family is thus milder than required for universal Hamiltonian families \cite{CMP,Kohler2020,Piddock2020}.
	\item The overall construction is significantly more streamlined than the existing query Hamiltonian gadget approach~\cite{Ambainis2013,GY16_2,GPY20}. A tradeoff, however, is that our construction appears to require error reduction for verification class Q, ruling out a black-box application of \cref{lem:lift} to the case $\textup{Q}=\StoqMA$. (Hardness for APX-SIM for stoquastic Hamiltonians cam be shown, however, via the query gadget approach~\cite{GPY20}.)
\end{enumerate}

\paragraph{For GSCON.} For \QCMAEXP-hardness of GSCON, we again follow a generic lifting approach (Lifting Lemma for GSCON, \cref{lem:liftGSCON}; for clarity, this lifting lemma is proven using completely different techniques than our Lifting Lemma for APX-SIM), although less generic than what we are able to achieve for APX-SIM. Namely, we restrict attention to quantum verification classes such as \QCMA or $\QCMAEXP$ (since we require the ability to prepare low energy/history states efficiently in the YES case), and to general 1D TI circuit-to-Hamiltonian mappings (again, with mild restrictions; see the definition of TI-standard, \cref{def:TIstd}. The reason for restricting to 1D TI mappings will be stated after proof techniques below). The main \QCMAEXP-hardness result of \cref{thm:TIGSCON} is achieved by instantiating our lifting framework with the 1D TI Gottesman-Irani construction~\cite{Gottesman2009}.

In terms of techniques, at a high level, our starting setup is similar to~\cite{GS14}, which we briefly review here for context (see \cref{sec:TI_GSCON} for details): given a QCMA verification circuit $V$,  apply Kitaev's circuit-to-Hamiltonian construction to obtain a local Hamiltonian $H=\sum_i H_i$, such that if $V$ is a YES (NO) instance, $\lmin(H)$ is small (large), for $\lmin(H)$ the smallest eigenvalue of $H$. Then, attach a particular ``switch gadget'' to $H$ to obtain a new Hamiltonian $H'$, so that any polynomial-length traversal of the low energy space of $H'$, from start state $\ket{\psi}$ to target state $\ket{\phi}$, forcibly ``switches on'' all terms of $H$. In the NO case, switching on $H$ incurs a large energy penalty, i.e.\ we hit the claimed energy barrier.

Extending this to the 1D TI setting presents various challenges. First, the construction of ~\cite{GS14} is highly non-local geometrically, as each switch qubit is coupled to \emph{all} local terms $H_i$ of $H$. In order to maintain this level of coupling in 1D, we first use an idea reminiscent of the space-time circuit-to-Hamiltonian construction of \citeauthor{Breuckmann_Terhal_2014}~\cite{Breuckmann_Terhal_2014}, and instead endow each qudit on the 1D chain with its \emph{own} ``local switch qudit'' (\cref{sscn:o1}). We then wish to add ``string constraints'' on these ``local switch qubits'' to force a prover to ``switch on'' each term $H_{i,i+1}$ of the 1D Hamiltonian one at a time on the chain. But here we face an additional pair of challenges: first, any na\"ive implementation of string constraints allows a cheating prover to switch on only $N-\p$ of the chain's local terms $H_{i,i+1}$ (we call this ``partial activation''; \cref{sscn:o2}); this is intuitively because the prover is allowed to apply arbitrary $\p$-local unitaries in each step, allowing it to ``shortcut'' the last $\p$ switch qudits in one step. Second, we cannot satisfy the desired completeness properties for GSCON by simply switching on local terms $H_{i,i+1}$ iteratively from left to right (\cref{sscn:o3}). Rather, we must allow a \emph{non-linear} order of activation.

It turns out that, not only can both the second and third challenges above be addressed in a unified black-box fashion\footnote{There is also a white-box approach to resolving the partial activation obstacle, which exploits a remarkable property of~\cite{Gottesman2009} which we call ``self-organization''. Since our aim in this work is black-box lifting theorems, we do not use this white-box fix; it is, however, detailed in \cref{app:altfix2} for the interested reader.}, but the unified fix will also make the construction remarkably ``robust''. Specifically, we first increase the local switch space dimension to $7$, which roughly will allow non-linear activation orders. We then carefully construct our string constraints so that, very briefly, any ground space evolution satisfying said constraints becomes ``trapped'' in a low-dimensional joint switch subspace on all qudits (\cref{fig:neighbour}). This low-dimensional space is precisely set up to achieve two things: (1) force \emph{all} local terms of $H$ to be simultaneously switched on, and (2) be ``logically protected'' from any switch subspace deviating from property (1) by a ``string'' of local unitaries of length $\Theta(N)$. The formal proof of correctness uses, among other tools, the Traversal Lemma (\cref{l:traversal}) of~\cite{GS14,GMV17}. Perhaps counterintuitively, soundness holds even if a cheating prover can apply $(N-1)$-local unitaries in each step, i.e.\ can act on \emph{all but one} qudit of the chain per step. Thus, our construction is significantly more robust (with respect to $\p$, the locality of the unitaries applied), than\footnote{The construction of \cite{GS14} can be ``blown up'' to have soundness parameter $\p$ scaling as $cN$ for any desired constant $c$; however, this requires increasing the switch gadget register size by $\Theta(N)$, and cannot achieve soundness against, say, $N-O(1)$, as we do here. Moreover, such a blow-up trick would be rather artificial, whereas here our desired low-dimensional switch subspace is ``logically protected'' from the remaining Hilbert space, making our construction ``inherently'' robust.} \cite{GS14}.

Finally, we can now explain our restriction to general 1D TI circuit-to-Hamiltonian mappings for GSCON, as alluded to earlier. Namely, the string constraints above are currently more-or-less tailored to a 1D setup (\cref{fig:neighbour}). In principle, these constraints \emph{can} be extended to higher dimensions; doing so in an abtract fashion to suit any desired geometry, however, appears to become cumbersome.

\subsection{Paper Structure}
We begin in \cref{sec:prelims} with notation and basic definitions (APX-SIM, GSCON, oracle complexity classes, etc). Section~\ref{sec:apx-sim} gives our generic lifting results for APX-SIM, including \PQMAexp-completeness for 1D TI Hamiltonians. Section~\ref{sec:TI_GSCON} gives our generic lifting framework for 1D TI GSCON, including \QCMAEXP-completeness of 1D TI GSCON via an instantiation using~\cite{Gottesman2009}. We conclude in \cref{sec:discussion} with a discussion and outlook.

\section{Preliminaries}\label{sec:prelims}

\subsection{Notation}
The sets $\lin{\spa{X}}$, $\herm{\spa{X}}$, $\unitary{\spa{X}}$ and $\pos{\spa{X}}$ denote the sets of linear, Hermitian, unitary, and positive semi-definite operators acting on space $\spa{X}$, respectively. For $H\in\herm{\spa{X}}$, $\lmin(H)$ is its smallest eigenvalue, and for $A\in\lin{\spa{X}}$, $\Null{A}$ is the null-space/kernel of $A$.

\subsection{APX-SIM and GSCON}\label{sec:apx-sim-gscon}

\paragraph{Approximate simulation.} APX-SIM is concerned with the properties of ground states. However, it is usually more natural to consider a low-energy subspace. With this in mind we follow \cite{GPY20} and define a more symmetric variant of the above problem APX-SIM which concerns itself with the low-energy subspace rather than the ground state:

\begin{definition}[$\forall$-APX-SIM$(H,A,k,l,a,b,\delta)$~\cite{GPY20}]\label{def:forall-apx-sim}
		Given a $k$-local Hamiltonian $H=\sum_i H_i $ acting on $N$ qubits, an $l$-local observable $A$, and real numbers $a$, $b$, and $\delta$ such that $b-a\geq N^{-c}$ and $\delta\geq N^{-c'}$, for $c,c'>0$ constant, decide:
	\begin{itemize}
		\item[YES.] If for all $\ket{\psi}$ satisfying $\bra{\psi}H\ket{\psi}\le \lmin(H)+\delta$, it holds that $\bra{\psi}A\ket{\psi} \le a$.
		\item[NO.] If for all $\ket{\psi}$ satisfying $\bra{\psi}H\ket{\psi}\leq \lmin(H)+\delta$, it holds that $\bra{\psi}A\ket{\psi}\geq b$.
	\end{itemize}
\end{definition}

\noindent Note that in contrast to APX-SIM, for completeness in $\forall$-APX-SIM one requires \emph{all} states below a threshold energy $\delta$ above the ground state energy to have expectation value upper-bounded by $a$. Throughout this work we consider the translationally invariant versions of APX-SIM and $\forall$-APX-SIM which simply for translationally invariant Hamiltonians:
\begin{definition}[TI-APX-SIM]\label{def:TIAPX}
	Defined analogously to APX-SIM, except the input Hamiltonian is specified via local term $h$ of a translationally invariant Hamiltonian $H=\sum h $ acting on $N$ qubits, where N is specified in binary, and each local term is describable in $\BigO(\log(N))$ bits.
\end{definition}
\begin{definition}[$\forall$-TI-APX-SIM]
Defined analogously to $\forall$-APX-SIM, except the input Hamiltonian is specified via local term $h$ of a translationally invariant Hamiltonian $H=\sum h $ acting on $N$ qubits, where N is specified in binary, and each local term is describable in $\BigO(\log(N))$ bits.
\end{definition}

\noindent \emph{Remarks.} (1) TI-APX-SIM is related to the local Hamiltonian problem, i.e.\ the question of estimating the ground state energy of a Hamiltonian. Indeed, if we take the local terms $h$ of $H$ as observables $A$, TI-APX-SIM yields YES if there exists a ground state such that all local terms cumulatively contribute energy at most $a$; and NO if each low energy state has cumulative energy at least $b$. (2) $\forall$-TI-APX-SIM immediately many-one reduces to TI-APX-SIM, in that if $\forall$-TI-APX-SIM yields NO, then also TI-APX-SIM yields NO; similarly if $\forall$-TI-APX-SIM yields YES, then so does TI-APX-SIM.
Thus, hardness results for $\forall$-TI-APX-SIM imply hardness results for TI-APX-SIM.

\paragraph{Ground state connectivity.}

\begin{definition}[Ground State Connectivity (\GSCON $({H},{k},{\aaa},{\bbb},{\ccc},{\ddd},{\delta},{\p},{m},{U_\psi},{U_\phi})$)\cite{GS14}]\label{def:GSCON}
\leavevmode
\begin{mylist}{\parindent}
\item Input parameters:
    \begin{enumerate}
        \item {$k$}-local Hamiltonian ${H}=\sum_i H_i$ acting on $N$ qubits with $H_i \in \herm{(\C^2)^{\otimes {k}}}$ satisfying $\snorm{H_i}\leq 1$.

        \item ${\aaa},{\bbb},{\ccc},{\ddd}, {\delta}\in\R$, and integer ${m}\geq0$, such that ${\bbb}-{\aaa}\geq {\delta}$ and ${\ddd}-{\ccc}\geq {\delta}$.

        \item Polynomial size quantum circuits {$U_\psi$} and {$U_\phi$} generating ``starting'' and ``target'' states $\ket{\psi}$ and $\ket{\phi}$ (starting from $\ket{0}^{\otimes N}$), respectively, satisfying $\bra{\psi}{H}\ket{\psi}\leq {\aaa}$ and $\bra{\phi}{H}\ket{\phi}\leq {\aaa}$.
    \end{enumerate}
\item Output:
\begin{enumerate}
    \item If there exists a sequence of {$\p$}-local unitaries $(U_{i})_{i=1}^m \in\unitary{\C^2}^{\times {m}}$ such that:
    \begin{enumerate}
        \item (Intermediate states remain in low energy space) For all $i\in [{m}]$ and intermediate states ${\ket{\psi_i}\coloneqq U_i\cdots U_2U_1\ket{\psi}}$, one has $\bra{\psi_i}{H}\ket{\psi_i}\leq {\aaa}$, and
        \item (Final state close to target state) $\enorm{ U_{{m}} \cdots U_1 \ket{\psi}-\ket{\phi}} \leq {\ccc}$,
    \end{enumerate}
    then output YES.
    \item If for all ${\p}$-local sequences of unitaries $(U_{i})_{i=1}^{{m}}\in\unitary{\C^2}^{\times {m}}$, either:
    \begin{enumerate}
        \item (Intermediate state obtains high energy) There exists $i\in [{m}]$ and an intermediate state ${\ket{\psi_i}\coloneqq U_i\cdots U_2U_1\ket{\psi}}$, such that $\bra{\psi_i}{H}\ket{\psi_i}\geq {\bbb}$, or
        \item (Final state far from target state) $\enorm{ U_{{m}} \cdots U_1 \ket{\psi}-\ket{\phi}} \geq {\ddd}$,
    \end{enumerate}
    then output NO.
\end{enumerate}
\end{mylist}
\end{definition}

\begin{definition}[\TIGSCON $({H},{k},{\aaa},{\bbb},{\ccc},{\ddd},{\delta},{\p},{m},{U_\psi},{U_\phi})$)]~\label{def:TIGSCON}
	Defined analogously to \GSCON, except the input Hamiltonian is specified via local term $h$ of a translationally invariant Hamiltonian $H=\sum h $ acting on $N$ qubits, where N is specified in binary representation, and each local term is describable in $\BigO(\log(N))$ bits.
\end{definition}

We now state the Extended Projection Lemma, which consists of three claims, the first of which was given in~\cite{Kempe2006}. The lemma was later extended to include the second and third claims~\cite{GY16_2}.

\begin{lemma}[Extended Projection Lemma (\cite{Kempe2006,GY16_2})]\label{l:kkr}
	Let $H=H_1+H_2$ be the sum of two Hamiltonians operating on some Hilbert space $\spa{H}=\spa{S}+\spa{S}^\perp$. The Hamiltonian $H_1$ is such that $\spa{S}$ is a zero eigenspace and the eigenvectors in $\spa{S}^\perp$ have eigenvalue at least $J>2\snorm{H_2}$. Let $K\coloneqq \snorm{H_2}$. Then, for any $\delta\geq0$ and $\ket{\psi}$ satisfying $\bra{\psi}H\ket{\psi}\leq \lmin(H)+\delta$, there exists a $\ket{\psi'}\in \spa{S}$ such that:
	\begin{itemize}
        \item (Ground state energy bound)
        \[
		\lmin(H_2|_{\spa{S}})-\frac{K^2}{J-2K}\leq \lmin(H)\leq \lmin(H_2|_{\spa{S}}),
	\]
where $\lmin(H_2|_{\spa{S}})$ denotes the smallest eigenvalue of $H_2$ restricted to space $\spa{S}$.
        \item (Ground state deviation bound)
        \[
            \abs{\braket{\psi}{\psi'}}^2\geq {1-\left(\frac{K+\sqrt{K^2+\delta(J-2K)}}{J-2K}\right)^2}.
            \]
        \item  (Energy obtained by perturbed state against $H$)
        \[
            \bra{\psi'}H\ket{\psi'}\leq\lmin(H)+\delta+2K\frac{K+\sqrt{K^2+\delta(J-2K)}}{J-2K}.
        \]
    \end{itemize}
\end{lemma}
Next, we state a quantum analogue of the union bound for commuting operators (see, e.g.~\cite{OMW19}).
\begin{lemma}[Commutative Quantum Union Bound]\label{lem:unionbound}
    Let $\set{P_i}_{i=1}^m$ be a set of pairwise commuting projectors, each satisfying $0\preceq P_i\preceq I$. Then for any quantum state $\rho$,
    \[
        1-\trace(P_m\cdots P_1\rho P_1\cdots P_m)\leq\sum_{i=1}^m\trace((I-P_i)\rho).
    \]
\end{lemma}
The following is a standard fact (see, e.g., Equation 1.33~\cite{G13} for a proof):
\begin{equation}\label{eqn:enorm}
    \trnorm{\ketbra{v}{v}-\ketbra{w}{w}}=2\sqrt{1-\abs{\braket{v}{w}}^2}\leq 2\enorm{\ket{v}-\ket{w}}.
\end{equation}

\subsection{Relevant (Oracle) Complexity Classes}

As discussed extensively in \cite[Sec.~4]{Kohler2020}, the natural complexity class for the local Hamiltonian problem for a translationally-invariant system is \QMAEXP.
Inuitively, this is because for  translationally invariant system, the only parameter is the system size $N$.
If $N$ is the input, then it can be encoded in a string which we label $\enc(N)$ (e.g.\ a binary string) which has length $O(\log(N))$.
Together with a promise gap which closes $\propto 1/T^2$, where $T$ is the run-time of the embedded computation, and a $1/\poly(N)$ promise gap allowed in the definition of the local Hamiltonian problem, this means that we can only saturate this bound if we allow $T=\poly(N)$---i.e., the encoded computation runs in time exponential in the input size ($\poly(\log N)$), which naturally gives \QMAEXP.\footnote{We take care to distinguish this from the class \normalfont\textsf{QMA\textsubscript{exp}} of \cite{Fefferman2016} which is for an exponentially small promise gap in the input size, but polynomial length run time, also called \normalfont\textsf{PreciseQMA}.}

\begin{definition}[\QMAEXP~\cite{Gottesman2009}] \label{Def:QMA_EXP}
	A promise problem $\Pi=(\ayes,\ano)$ is in $\QMAEXP$ if and only if there exists a $k\in \BigO(1)$ and a Quantum Turing Machine $M$ such that for each input $x\in\set{0,1}^*$ with $n=n$, and any proof $\ket{\psi}\in\base^{\otimes 2^{n^k}}$, on input $(x,\ket{\psi})$, $M$ halts in $2^{n^k}$ steps. Furthermore,
	\begin{itemize}
		\item (Completeness) If $x \in \ayes$, $\exists\, \ket{\psi}\in\base^{\otimes 2^{n^k}}$ such that $M$ accepts $(x,\ket{\psi})$ with probability $\ge2/3$.	
		\item (Soundness) If $x \in \ano$, then $\forall\, \ket{\psi}\in \base^{\otimes 2^{n^k}}$, $M$ accepts $(x,\ket{\psi})$ with probability $\le1/3$.
	\end{itemize}
\end{definition}

In \cite{GPY20}, the authors prove that APX-SIM is $\PP^{\QMA[\log]}$-complete by showing that APX-SIM is $\PP^{||\QMA}$-complete, and in a second step that $\PP^{||\QMA}=\PP^{\QMA[\log]}$.
For the translationally invariant version of APX-SIM, we are essentially dealing with the same encoded computation, yet with exponentially input less information (i.e.\ the problem is succinctly encoded), just like in the case for the local Hamiltonian problem.
Here, the base class now has a runtime $\poly(N)$, but an input of length $O(\log(N))$ and hence is an \emph{exponential} time computation. 
The right class for which TI-APX-SIM is complete is thus $\PQMAexp$, where in the following we will implicitly assume that \QMAEXP is defined to take a $\log(N)$-sized input.
Technically, this class comes out of the $\EXPpar$-completeness result, but it is also motivated in another way. Since \QMAEXP has verification time $\poly(N)$, its circuit-to-Hamiltonian construction will have norm $\poly(N)$. Thus, the number of adaptive queries to the \QMAEXP oracle to estimate the ground state energy within $1/\poly(N)$ error is $\log(N)=O(\poly(n))$, which is polynomial in the input size.
So the \PP machine, which runs in time $\poly(\log(N))=O(\poly(n))$, makes $\poly(\log(N))=O(\poly(n))$ queries to the oracle, which yields \QMAEXP.

\QMA with a classical proof yields the complexity class \QCMA.
\QCMAEXP is then to \QCMA what \QMAEXP is to \QMA, as given in the following definition.
\begin{definition}[\QCMAEXP]\label{def:QCMAEXP}
	A promise problem $\Pi=(\ayes,\ano)$ is in \QCMAEXP if and only if there exists $k\in \BigO(1)$ and an exponential-time uniform family of quantum circuits $\set{Q_n}$, where $Q_n$ takes as input a string $x\in\Sigma^*$ with $n=n$, a classical proof ${y}\in \set{0,1}^{2^{n^k}}$, and $2^{n^k}$ ancilla qubits in state $\ket{0}^{\otimes 2^{n^k}}$, such that:
	\begin{itemize}
		\item (Completeness) If $x\in\ayes$, $\exists y\in\set{0,1}^{2^{n^k}}$ such that $Q_n$ accepts $(x,{y})$ with probability $\ge2/3$.
		\item (Soundness) If $x\in\ano$,$\forall{y}\in \set{0,1}^{2^{n^k}}$, $Q_n$ accepts $(x,{y})$ with probability $\le1/3$.
	\end{itemize}
\end{definition}

\paragraph{Oracle complexity classes.} The classes $\PQMAlog$, $\PQMApar$, and $\PQMA$ denote the set of languages decidable by a polynomial-time deterministic Turing machine with access to, respectively, logarithmically many adaptive queries to a \QMA oracle, polynomially many parallel queries to a \QMA oracle, and polynomially many adaptive queries to a \QMA oracle, respectively. It is known that $\PQMAlog=\PQMApar$~\cite{GPY20}. The classes $\EXPpar$ and $\PQMAexp$ are defined analogously, except the former has an exponential-time deterministic Turing machine as its base (and hence can make exponentially many parallel queries), and where the oracle is for $\QMAEXP$, respectively. Note that by ``queries to an oracle'', we mean feeding the oracle an input to a complete problem for the oracle's class (e.g.\ an instance of the Local Hamiltonian problem for a \QMA oracle).

\section{Encoding Computation into Measurement Problems on Low Energy Spaces}\label{sec:apx-sim}

\subsection{Overview and Circuit-to-Hamiltonian Mappings}\label{sec:circuit-to-ham}
Given a quantum circuit $U = U_TU_{T-1}\ldots U_1$ acting on some Hilbert space $\mathcal H_q$, one can write down a so-called \emph{history state} Hamiltonian
\[
    \Hprop = \sum_{t=0}^{T-1} h_t
    \quad\text{where}\quad
    h_t \coloneqq  \sum_{\ket e} \big( \ket t \ket e - \ket{t+1} U_t \ket e \big)\big(\bra t \bra e - \bra{t+1} \bra e U_i^\dagger\big),
\]
where we sum over a basis $\{ \ket e \}$ for $\mathcal H_q$.
Then $\Hprop$ has a degenerate ground space spanned by states of the form
\begin{equation}\label{eq:history-state}
    \ket{\psi} = \frac{1}{\sqrt T} \sum_{t=0}^T \ket t \ket{\psi_t}
    \quad\text{where}\quad
    \ket{\psi_t} \coloneqq  U_tU_{t-1}\cdots U_1\ket{\psi_0}
\end{equation}
for any $\ket{\psi_0} \in \mathcal H_q$.
These history states---based on an idea by \cite{Feynman1986}---are a fundamental ingredient in deriving Hamiltonian complexity results, such as Kitaev's seminal proof of \QMA-hardness of the local Hamiltonian problem \cite{Kitaev2002}.
Since $\Hprop$ and variants thereof translate quantum circuits and their evolution into the ground space of a Hamiltonian, they are often collectively called \emph{circuit-to-Hamiltonian} mappings; for an in-depth historic overview see \cite{Bausch2016}.

\newcommand\tbhead[2]{\vtop{\hbox{\strut #1}\hbox{\strut #2}}}
\begin{table}[t]
    \centering
    \hspace*{-8mm}
    \begin{tabular}{rlllllll}
        \toprule
         & \tbhead{universal}{for} & \tbhead{scaling}{of $T(N)$} & symmetry & \tbhead{interaction}{graph} & locality & \tbhead{local}{dimension} & \tbhead{instance $\ell$}{encoded in} \\
         \midrule
        \cite{Kitaev2002} & \multirow{5}{*}{\QMA} & \multirow{5}{*}{$\displaystyle\frac{1}{\poly}$} & \multirow{5}{*}{none} & arbitrary & 5 & 2 & \multirow{5}{*}{local terms} \\
        \cite{Kempe2003a} & & & & arbitrary & 3 & 2 &  \\
        \cite{Kempe2006} & & & & arbitrary & 2 & 2 & \\
        \cite{Oliveira2008} & & & & 2D planar & 2 & 2 & \\
        \cite{Aharonov2009} & & & & 1D line & 12 & 2 & \\
        \midrule
        \cite{Bravyi_Bessen_Terhal_2006}  & \StoqMA & $\displaystyle\frac{1}{\poly}$ & none & arbitrary  & 2 & 2 & local terms \\
        \midrule
        \cite{Gottesman2009} & \multirow{3}{*}{\QMAEXP} & \multirow{3}{*}{$\displaystyle\frac{1}{\poly}$} & \multirow{3}{*}{translational} & 1D line & 2 & huge & \multirow{3}{*}{system size} \\
        \cite{Bausch2016} & & & & 1D line & 2 & 42 & \\
        \cite{Bausch2017} & & & & 3D fcc lattice & 4 & 4 & \\
        \midrule
        \cite{Fefferman2016}$^\dagger$ & \multirow{2}{*}{\PreciseQMA} & \multirow{2}{*}{$\displaystyle\frac{1}{\exp}$} & none & arbitrary & 3 & 2 & local terms \\
        \cite{Kohler2020}$^\ddagger$ & & & translational & 1D line & 2 & 42 & system size \\
         \bottomrule
    \end{tabular}
    \caption{History state Hamiltonians from existing literature that satisfy \cref{def:local-mapping} with varying characteristics. Shown the complexity class for which they encode a witness in the ground state energy, size $T$ of history state given in \cref{eq:history-state} in terms of the system size $N$, further properties of the Hamiltonian, as well as their dependence on the problem instances $\ell$.\\$^\dagger$ references \cite{Kempe2003a} as the underlying construction; $^\ddagger$ references \cite{Bausch2016}.}
    \label{tab:uwHs}
\end{table}

As it is currently written, there is no way for $\Hprop$ alone to ensure the ground state is spanned by states of the form \cref{eq:history-state}, but such that $\ket{\psi_0}$ is initialised to a specific state in any way.
Yet in order to raise a history state Hamiltonian to a construction which can encode a useful computation within the ground state, one requires the ability to constrain the input of the computation, e.g.\ to prepare ancillas in a predefined state.
This is readily achieved by adding an input projector conditioned on the clock register being in state $\ket0$---e.g.\ $\ketbra0 \otimes \Pin$ would restrict $\ket{\psi_0} \in \ker(\Pin)$.
In this fashion, one can e.g.\ encode a \BQP-computation into the unique ground state of a history state Hamiltonian.
The combination of such an input penalty and history state Hamiltonian we denote with $\Hw \coloneqq  \Hprop + \Pin$.

If only the input is constrained, the ground state energy of $\Hw$ is still exactly zero, assuming $\ker(\Pin) \neq \emptyset$.
In this case, for all correctly-initialised $\ket{\psi_0} \in \ker(\Pin)$, $\ker(\Hw)$ contains all correctly initialised history states of the form given in \cref{eq:history-state}.
For this null space spanned by correctly-initialised history states, we write $\Null{\Hw} \coloneqq  \ker(\Hw)$.
The spectral gap of $\Hw$---i.e.\ the difference in energy between ground space and first excited state---is then $\sim 1/T^2$, where $T$ is the size of the history state in \cref{eq:history-state}, which directly corresponds to the length of the embeddable computation \cite{Bausch2016a}.

By adding a similar \emph{output} penalty $P_\mathrm{out}$ that raises the energy for \NO answers of an embedded computation, one can achieve that the ground state energy $\lmin(\Hw + P_\mathrm{out})$ either lies below a threshold $\alpha$ for a \YES answer---or above $\beta$ in case of a \NO outcome. We emphasise that this new ground state energy is now generally no longer zero, and the sum of both Hamiltonians is generally frustrated.

And finally, if we leave part of the input to the circuit $U$ encoded in $\Hw(V)$ unconstrained, but we add an output penalty, this enforces the ground state of $\Hw + P_\mathrm{out}$ to have a state initialised to an input that minimises overlap with the output penalty; this corresponds to the encoding of a nondeterministic computation, and the input is then e.g.\ a witness for a \QMA verifier.
In essence, this is how \cite{Kitaev2002} proved that the local Hamiltonian problem is \QMA-hard, i.e.\ the question of approximating the ground state energy of a local Hamiltonian, where the precision to which one needs to approximate said energy scales with the promise gap of the problem.
Generally, for a \BQP computation, the difference $|\beta-\alpha| = 1/\poly$ in the instance size, which translates to a ground state approximation precision of $1/\poly$ in the system size (cf.\ \cref{sec:hamiltonian-embedding}).

We will avoid rolling out the entire theory of history state Hamiltonians in this work and go a more abstract route.
More precisely, we will offload the question of proving $\EXP^{\parallel\QMA}$-hardness of TI-APX-SIM to a quantum circuit that simulates the oracle calls.
Starting from a rigorous definition of this type of computation we wish to encode, we will then require two mild assumptions on the type of circuit-to-Hamiltonian mapping used to translate the circuit to a ground state (the ability to access two outputs of the computation locally).
Given the mapping has these two properties, $\EXP^{\parallel\QMA}$-hardness will follow.

\subsection{Hamiltonians with a Universal Ground State}\label{sec:universal-gs-hams}
As mentioned in the previous section, the specifics of the Hamiltonian used to encode computation is, for our purposes, irrelevant: whether it takes the shape of a history state Hamiltonian (\cref{eq:history-state}), features a more complicated clock construction such as in \cite{Aharonov2009,Gottesman2009,Breuckmann_Terhal_2014,Bausch2016}, or is something completely different is not important.

More concretely, what we require of the circuit-to-Hamiltonian mapping is the ability to single out a low-energy subspace that encodes valid computations---meaning either a single ground state, multiple ground states, or multiple eigenstates with energy below some cutoff, all of which are to represent valid computations under the chosen circuit-to-Hamiltonian mapping.
This valid low-energy subspace needs to be separated from the rest of the spectrum, which may include states which do not encode a valid computation. Beyond this fundamental requirement, and in order to translate a circuit that simulates oracle queries into a ground state energy problem, the circuit-to-Hamiltonian mapping needs to have two high-level properties (formal requirements in \cref{def:local-mapping}):
\begin{enumerate}
    \item The possibility to inflict an energy penalty onto the output of a computation, which means that we can break the low-energy subspace of valid computations into a space $V_\YES$ and $V_\NO$ that correspond to computations that evolve correctly and incorrectly, respectively; and such that the largest eigenvalue of $V_\YES$ is below the smallest eigenvalue of $V_\NO$.
    \item Locality in the encoding, in the sense that neighbouring qubits in the circuit at some time-step need to map to neighbouring spins in the many-body system. 
\end{enumerate}
Both these assumptions are mild, and readers familiar with circuit-to-Hamiltonian mappings will immediately recognise that both of these are generally satisfied.

\begin{definition}[Conformity]\label{def:conforming}
    Let $H$ be a Hamiltonian with some well-defined structure $S$ (such as $k$-local interactions, all constraints drawn from a fixed finite family, with a fixed geometry such as 1D, translational invariance, etc). We say a Hermitian operator $P$ \emph{conforms} to $H$ if $H+P$ also has structure $S$.
\end{definition}

\noindent For example, if $H$ is a $1$D translationally invariant Hamiltonian on qubits, then $P$ conforms to $H$ if $H+P$ is also $1$D translationally invariant.

We now define an abstract notion of local circuit-to-Hamiltonian mappings; intuition given subsequently.

\begin{definition}[Local Circuit-to-Hamiltonian Mapping]\label{def:local-mapping} Let $\spa{X}=(\C^2)^{\otimes m}$ and $\spa{Y}=(\C^2)^{\otimes n}$. A map $\Hw:\unitary{\spa{X}}\mapsto\herm{\spa{Y}}$ is a \emph{local circuit-to-Hamiltonian mapping} if, for any $T>0$ and any sequence of $2$-qubit unitary gates $U=U_TU_{T-1}\cdots U_1$, the following hold:
\begin{enumerate}
    \item (Overall structure) $\Hw(U)\succeq0$ has a non-trivial null space, i.e.\ $\Null{\Hw(U)}\neq \emptyset$. This null space is spanned by (some appropriate notion of) ``correctly initialized computation history states'', i.e.\ with ancillae qubits set ``correctly'' and gates in $U$ ``applied'' sequentially.
    \item (Local penalization and measurement) Let $q_1$ and $q_2$ be the first two output wires of $U$ (each a single qubit), respectively. Let $\spre\subseteq \spa{X}$ and $\spost\subseteq\spa{Y}$ denote the sets of input states to $U$ satisfying the structure enforced by $\Hw(U)$ (e.g.\ ancillae initialized to zeroes), and null states of $\Hw(U)$, respectively. Then, there exist projectors $M_1$ and $P_T$, projector $M_2$ conforming to $\Hw(U)$, and a bijection $\f:\spre\mapsto\spost$, such that for all $i\in\set{1,2}$ and $\ket{\phi}\in\spre$, the state $\ket{\psi}=\f(\ket{\phi})$ satisfies
        \begin{equation}\label{eqn:sim}
            \Tr\big(\ketbra{0}{0}_i(U_TU_{T-1}\ldots U_1)\ketbra{\phi}{\phi}(U_TU_{T-1}\ldots U_1)^\dagger\big) = \Tr\big(\ketbra{\psi_T} M_i\big),
        \end{equation}
        where $\ket{\psi_T} = P_T \ket\psi / \enorm{P_T \ket\psi}$ is $\ket\psi$ postselected on measurement outcome $P_T$ (we require $P_T \ket\psi\neq 0$). Moreover, there exists a function $\g:\nats\times\nats\mapsto\reals$ such that
        \begin{align}
            \enorm{P_T\ket{\psi}}^2&=\g(m,T)\text{ for all }\ket{\psi}\in \Null{\Hw(U)},\label{eqn:sim3}\\
            M_i&= P_T M_i P_T\label{eqn:sim2}.
        \end{align}
\end{enumerate}
The map $\Hw$, and all operators/functions above ($M_1$,$M_2$,$P_T$,$\f$,$\g$) are computable given $U$.
\end{definition}
\noindent\emph{Intuition.} We stress the following about \cref{def:local-mapping}:
    \begin{enumerate}
        \item It places no restrictions on the efficiency of computing $\Hw$, $M_1$, $M_2$, $P_T$, $f$, $g$. Any such resource-restriction will later be application-dependent. Also, there is no restriction to how ``correctly initialized computation history states'' are encoded, other than the existence of the bijection $f$ between ``states initialized according to the rules of $U$'' and null states of $\Hw(U)$ satisfying the properties above.

        \item The term ``local'' in ``local circuit-to-Hamiltonian mapping'' is not referring to the locality of $\Hw(U)$ (e.g.~\cref{def:local-mapping} also allows $\Hw(U)$ to be a non-local but sparse Hamiltonian). Rather, it refers to the fact that local measurements on the output qubits of $U$ can be simulated via local measurements on the ground space of $\Hw(U)$ (up to postselection) via bijection $\f$ and \cref{eqn:sim}. Also, there is no restriction \emph{a priori} on $\g(m,T)$, other than $\g(m,T)\neq 0$ for all $m,T\geq0$.
        \item For our applications, we only require simulation of local measurements on output qubits $1$ and $2$ of $U$; hence the phrasing of Point 2 in \cref{def:local-mapping} (which, we note, makes this notion of local simulation milder than that used in universality results such as \cite{CMP,Kohler2020,Piddock2020}).
        Intuitively, the first qubit will encode whether $U$ accepts or rejects (i.e.\ outputs $1$ or $0$, respectively). In the setting of APX-SIM, $M_1$ will play the role of observable $A$ from \cref{def:apx}; as such, $M_1$ need not necessarily conform to $\Hw(U)$. In contrast, $M_2$ will be used to penalize a certain ``flag qubit'' in our construction of \cref{lem:augmented-circuit}, and will be part of the Hamiltonian $H$ from \cref{def:apx}; as such, we require $M_2$ to conform to $\Hw(U)$. Finally, there is nothing particular about the choice of $\ketbra{0}{0}_i$ in \cref{eqn:sim}; any fixed single-qubit projector would suffice.

        \item \Cref{eqn:sim3} says all null states of $\Hw(U)$ have the same weight on the final time step, $T$. This is used, for example, in \cref{lem:low-energy-as-many-yes-as-possible} when we wish to exchange one $\ket{\phi'}\in\spost$ with another state $\ket{\phi}\in\spost$, and say something meaningful about the computation encoded in $\ket{\phi'}$ versus $\ket{\phi}$.

        \item For the case of Kitaev's circuit-to-Hamiltonian construction~\cite{KSV02}, $P_T$ projects the clock register down to $\ket{T}$, and is $3$-local since there the clock is encoded in unary; translationally-invariant history state Hamiltonians such as from \cite{Gottesman2009,Bausch2016} are readily modified to allow $P_T$ to be one-local, by increasing the local dimension slightly and adding a final step to the Turing machine that drives the clock construction, making the halting configuration locally recognizable.
        In all history state Hamiltonians, $\g(m,T)=\enorm{P_T \ket\psi}^2=1/(T+1)$ for $\ket\psi$ a uniform history state.\footnote{Modified history state Hamiltonians with non-uniform superpositions over the time steps such as in \cite{Bausch2016a} scale accordingly with the weight on the last timestep.}
        Finally, \cref{eqn:sim2} captures the fact that $M_i$ has a clock register projecting onto $\ketbra{T}{T}$, and so is supported solely on the Hilbert space corresponding to time $T$ (i.e.\ projected onto by $P_T$).
    \end{enumerate}

\subsection{Oracle Queries as Subroutines}\label{sec:oracle-subroutines}

Throughout this and subsequent subsections, we will often use $\PQMApar$ as a concrete guiding example; however, our constructions (such as \cref{lem:augmented-circuit}) are designed to apply to a much more general set of classes of form $\DQ$ (see \cref{def:dTM,def:qV}).  With this in mind, recall that $\PQMApar$ corresponds to a polynomial-time classical deterministic TM with polynomially many parallel queries to a $\QMA$ oracle. As a first step, in this section we map an arbitrary $\PQMApar$ computation to a single ``QMA verification circuit''. We stress, however, that this is just part of our overall reduction; indeed, the construction of \cref{lem:augmented-circuit} by itself is \emph{not sound} (i.e.\ NO instances of $\PQMApar$ are not mapped to NO ``instances of QMA''). In~\cref{sec:hamiltonian-embedding}, we take the second step of carefully embedding the circuit produced by \cref{lem:augmented-circuit} into an appropriate circuit-to-Hamiltonian construction to obtain $\PQMApar$-hardness results. \Cref{sscn:other} then uses the generality of our constructions here to show how similar results follow for a wide variety of settings, such as the translationally invariant and small promise gap settings.

\paragraph{Reducing $\PQMApar$  to a single quantum verification circuit.} We now give a generic construction for embedding an arbitrary $\PQMApar$ circuit into a \emph{single} quantum verification circuit (i.e.\ an ``un-sound QMA circuit''). In doing so, as explained in \cite{GY16_2}, we must allow for the possibility that QMA oracle queries may be \emph{invalid}, in that they might violate the QMA promise gap. This entails accounting for two potential obstacles: first, from an ``oracle query'' perspective, a QMA oracle fed an invalid QMA query may output $0$ or $1$ arbitrarily. Second, from a ``QMA verification circuit'' perspective, we cannot assume anything about the acceptance probability of the verifier when fed the optimal proof for an invalid instance, even after standard error reduction. Specifically, for a valid YES (resp., valid NO) QMA instance, standard error reduction implies the optimal proof is accepted with probability exponentially close to $1$ (resp., $0$); for an invalid instance, this optimal probability may still be $1/2$. The first of these obstacles requires one to define a valid $\PQMApar$ machine's final output bit to be independent of how invalid queries are answered~\cite{Goldreich_2006} (otherwise, the output of the $\PQMApar$ machine is not necessarily well-defined on a given input).
Note that when the QMA oracle is replaced with a class of decision languages such as NP (rather than promise problems), such issues do not arise.

We now give the construction in \cref{lem:augmented-circuit}. Since the aim of this paper is \emph{generic} reductions for lifting hardness results for one class of problems to another, we state the following lemma rather abstractly. For this, we first require two definitions, the first of which is standard.
\begin{definition}(Deterministic class)\label{def:dTM}
    A set $C$ of languages is a \emph{deterministic class} if, for any language $L\in C$, there exists a deterministic Turing machine $M$ which can decide $L$ under the resource constraints specified by $C$. Formally, given any input $x\in\set{0,1}^n$, $M$ halts after using $R(n)$ resources (where $R$ may specify bounds on time or space), and accepts if $x\in L$ or rejects if $x\not\in L$.
\end{definition}
\noindent Standard examples of deterministic classes include \PP, \PSPACE, and \EXP.

\begin{definition}(Existentially quantified quantum verification class (\qvc))\label{def:qV}
    A set $C$ of promise problems is an \emph{existentially quantified quantum verification class} if any promise problem $A=(\ayes,\ano,\ainv)$ in $C$ satisfies the following. There exist computable functions $f,g,h:\nats\mapsto\nats$, as well as a deterministic Turing machine $M$ which, for any input $x\in\set{0,1}^n$, uses $R(n)$ resources to produce a quantum verification circuit $V$ (consisting of $1$- and $2$-qubit gates) and thresholds $c,s\in\reals^+$ such that $c-s>1/h(n)$. Here, $R(x)$ refers to resources such as time, space, etc, as required by $C$. The circuit $V$ takes in a quantum proof $\ket{\psi}$ on $f(n)$ qubits, $g(n)$ ancilla qubits initialized to all zeroes, and has a designated output qubit, such that:
    \begin{itemize}
        \item (YES case) If $x\in \ayes$, there exists a quantum proof $\ket{\psi}$ on $f(n)$ qubits such that measuring the output qubit of
        $
            V\ket{\psi}\ket{0\cdots 0}
        $
        in the standard basis yields $1$ with probability at least $c$.
        \item (NO case) If $x\in \ano$, for all quantum proofs $\ket{\psi}$ on $f(n)$ qubits, measuring the output qubit of
        $
            V\ket{\psi}\ket{0\cdots 0}
        $
        in the standard basis yields $1$ with probability at most $s$.
    \end{itemize}
    Without loss of generality, we assume the output qubit of $V$ is the first wire exiting $V$.
\end{definition}
\noindent For example, for \QMA, $R(n)$ denotes a polynomial bound (with respect to $n$) on the number of time and space steps taken by $M$, while $f$ and $g$ are fixed polynomials, $c=2/3$ and $s=1/3$. \NP is similar, except $c=1$, $s=0$, $V$ is diagonal in the standard basis and measures its proof in the standard basis (via the principle of deferred measurement~\cite{Nielsen_and_Chuang}).
In this way, classes such as \NP, \NEXP, \QCMA, \QMA, and so forth are examples of a \qvc.

We are now ready to state the following abstract lemma. As a concrete guiding example, consider $\class{D}=\PP$ and $\class{Q}=\QMA$ below.
\begin{lemma}\label{lem:augmented-circuit}
Let $x\in\set{0,1}^n$ be an instance of a problem in $\DQ$, where $\class{D}$ is a deterministic class (such as \PP) and $\class{Q}$ is a \qvc\
(such as \QMA). Let $U$ be a $\DQ$ machine (we will typically think of $U$ as a circuit with access to an oracle for \class{Q}) making $m$ parallel queries to a \class{Q}-oracle to decide $x$. Then, there exists a quantum circuit $V$ with the following properties:
\begin{enumerate}
    \item Given $x$ and $U$, $V$ can be computed in time polynomial in the sizes of $U$ and the verifier for $\class{Q}$ (both of which may be viewed as quantum circuits consisting of $1$- and $2$-qubit gates).
    \item $V$ takes as input $m+2$ registers: n input register $A$ containing $x\in\set{0,1}^n$, $m$ proof registers $B_i$ containing a joint quantum proof $\ket{\wm}$ (where ideally $\ket{\wm}=\ket{w_1}\ket{w_2}\cdots\ket{w_m}$), with register $B_i$ to be verified by a \class{Q}-circuit $V_i$ (see \cref{fig:augmented-circuit}), and an ancilla register $C$ which is assumed to be initialized to the all-zeroes state. Without loss of generality, we assume each verifier $V_i$ has the same completeness and soundness parameters $c$ and $s$, respectively.

    \item $V$ has two designated output wires: $\qout$ is supposed to encode the output of $U$, and $\qflag$ the number of $\class{Q}$ queries made by $U$ which were YES instances. (Without loss of generality, these are the first and second wires exiting $V$, respectively.) Formally, suppose $V$ is fed the joint proof $\ket{\wm}$, and let $\ket{\psi}$ denote the output state of $V$.


        Let sets $S_0$ and $S_1$ partition $\set{0,1}^m$ such that the \class{D} machine underlying $U$ rejects (resp. accepts) when given a string of query responses $y\in S_0$ (resp. $y\in S_1$). Define
            \[
                \pyw\coloneqq \Pr\left(\bigwedge_{i=1}^m V_i \text{ outputs } y_i \bigg\vert \ket{\wm}\right).
            \]
            Note that in the ideal case $\ket{\wm}=\ket{w_1}\cdots\ket{w_m}$, $\pyw$ simplifies to
            \[
                \pyw=\prod_{i=1}^m\Pr(V_i \text{ outputs } y_i\mid \ket{w_i}).
            \]
            In both cases,
            \begin{eqnarray}
                \trace\left(\ketbra{\psi}{\psi}\cdot\ketbra{1}{1}_{\qout}\right)&=&\sum_{y\in S_1}\pyw\label{eqn:out}\\
                \trace\left(\ketbra{\psi}{\psi}\cdot\ketbra{1}{1}_{\qflag}\right)&=&\sum_{y\in \set{0,1}^m}\pyw\cdot \sin^2\left(\angf\cdot\HW(y)\right)\label{eqn:flag}.
            \end{eqnarray}
where $HW(y)$ is the Hamming weight of $y$.
\end{enumerate}
\end{lemma}
\begin{proof}
\begin{figure}[t]
    \centering
    \begin{tikzpicture}
    \node[right] at (0, 0) {%
    \begin{quantikz}
	      \lstick{$\ket{q_1}$}       & \qw & \gate[wires=2]{\text{verifier}\ V_1\ \,}\gateoutput{$\mathrm{out}_1$} & \ctrl{6}           & \qw\ \ldots\ & \qw                & \gate[6,nwires={2,3,5}]{\quad U'\quad} &                                              \\
	      \lstick{$\ket{w_1}$}       & \qw &                                                                       &                    &              &                    &                                                &                                              \\
	                                 &                   &                                                                       &                    &              &                    &                                                &                                              \\
	      \lstick{$\ket{q_m}$}       & \qw & \gate[wires=2]{\text{verifier}\ V_m}\gateoutput{$\mathrm{out}_m$}     & \qw                & \qw\ \ldots\ & \ctrl{3}           &                                                &                                              \\
	      \lstick{$\ket{w_m}$}       & \qw &                                                                       &                    &              &                    &                                                &                                              \\
	\lstick{$\ket{x}$} & \qw\qwbundle{n} & \qw                                                                   & \qw                & \qw\ \ldots\ & \qw                &                                                & \qw\rstick{$\ket\qout$} \\
	       \lstick{$\ket{0}$}        & \qw   & \qw                                                                   & \gate{R(\ang)} & \qw\ \ldots\ & \gate{R(\ang)} & \qw                                            & \qw\rstick{$\ket\qflag$}
    \end{quantikz}
    };
    \draw[white, line width=1mm] (2, 1.1) -- (14, 1.1);
    \draw[white, line width=1mm] (2, 1.3) -- (14, 1.3);
    \draw[white, line width=1mm] (2, 0.9) -- (14, 0.9);
    \node at (3.4, 1.2) {$\vdots$};
    \node at (1.4, 1.2) {$\vdots$};
    \end{tikzpicture}
    \caption{The circuit $V$ constructed in \cref{lem:augmented-circuit}. The $V_i$ are the \class{Q}-verifiers, each taking input $\ket{q_i}$ and proof/witness $\ket{w_i}$. (In principle, states $\ket{w_i}$ can be entangled as one joint state $\ket{\wm}$; this is dealt with in the proof of \cref{lem:low-energy-as-many-yes-as-possible}.) $U'$ denotes the host postprocessing circuit in the original $\DQ$ circuit $U$, which takes the \class{Q}-query responses and outputs $U$'s final answer. The gates $R(\ang)$ denote a rotation in the standard basis of angle $\ang$. For simplicity, we have not depicted any preprocessing needed by $U$ to compute the inputs $\ket{q_i}$ to the \class{Q} verifiers $V_i$, nor have we depicted the ancilla register $C$. For clarity, as a black box, the circuit $V$ takes in the input to the $U$ circuit, the joint proof $\ket{\wm}$, and the ancilla register $C$.
	}
    \label{fig:augmented-circuit}
\end{figure}
As depicted in \cref{fig:augmented-circuit}, $V$ is constructed by translating the \class{D} machine underlying $U$ into a quantum circuit $U'$, and then ``simulating'' the $m$ (parallel) oracle calls $U$ makes as sub-routines, in the sense of executing their \class{Q}-verification circuits $V_i$ on the relevant subset of $\ket{\wm}$ given by  $\ketbra{w_i}=\Tr_{B_j\neq B_i}[\ketbra{\wm}]$.
Note that $U'$ is diagonal in the standard basis, i.e.\ is a classical circuit, and computes the inputs to the \class{Q}-verification circuits $V_i$ on-the-fly; the lanes $\ket{q_i}$ in \cref{fig:augmented-circuit} indicating the inputs to the respective verification subroutines are product states diagonal in the computational basis, i.e.\ $\ket{q_1}\ket{q_2}\cdots\ket{q_m}$.

The gate $R(\theta)$ in \cref{fig:augmented-circuit} denotes the rotation matrix
\[
    R(\theta)=\left(
                \begin{array}{cc}
                  \cos\theta &  -\sin\theta\\
                  \sin\theta & \cos\theta \\
                \end{array}
              \right).
\]
To formally state the action of the overall circuit $V$, let $X,Y,Z$ denote the input registers to $U'$ holding input $x\in\set{0,1}^m$, query response string $y=y_1\cdots y_m$, and ancilla (initialized to all zeroes), respectively. Since $U'$ is a classical circuit, without loss of generality it maps any
\[
    \ket{x}_X\ket{y}_Y\ket{0\cdots 0}_Z \mapsto \ket{x}_X\ket{y}_Y\ket{0\cdots 0f(y)}_Z,
\]
where $f(y)$ is the output of $U'$ (i.e.\ the \class{D} machine) given query response string $y$.
If we now let $F$ denote the flag qubit register, the output $\ket{\psi}$ of $V$ is given by
\begin{equation}
    \ket{\psi}=\sum_{y\in\set{0,1}^m}\alpha_y\ket{x}_X\ket{y}_Y\ket{0\cdots 0 f(y)}_Z\left(\cos\left(\angf\cdot \HW(y)\right)\ket{0}+\sin\left(\angf\cdot \HW(y)\right)\ket{1}\right)_F,
\end{equation}
where for succinctness we omit registers such as those containing proof $\ket{\wm}$, since $U'$ does not act on these registers.
Here, $\HW(y)$ is the Hamming weight of $y$, and where
\[
    \abs{\alpha_y}^2= \Pr\left(\bigwedge_{i=1}^m V_i \text{ outputs } y_i \bigg\vert \ket{\wm}\right).
\]
This immediately yields Equations~(\ref{eqn:out}) and~(\ref{eqn:flag}).
\end{proof}

\noindent{\emph{Remarks.}} We now make several remarks regarding \cref{lem:augmented-circuit}, including for the case when \class{Q} is a class of \emph{promise} problems. For concreteness, in our discussion here we set $\class{D}=\PP$ and $\class{Q}=\QMA$, in which case the construction of \cref{lem:augmented-circuit} runs in polynomial-time in $n$.
\begin{enumerate}
    \item The construction of \cref{lem:augmented-circuit} and \cref{fig:augmented-circuit} says nothing about valid versus invalid \QMA queries (i.e.\ when $q_i$ is an invalid \QMA query string in \cref{fig:augmented-circuit}). This will be dealt with in \cref{sec:hamiltonian-embedding}.
    \item When we later use \cref{lem:augmented-circuit}, we will penalize the flag qubit register $F$ carefully so as to force all \emph{valid} queries to be answered correctly. This is, in a nutshell, the purpose of the flag qubit.
    \item It is tempting to assume in \cref{fig:augmented-circuit} that, without loss of generality, the optimal strategy of a (potentially dishonest) prover is to send the optimal proof $\ket{w_i}$ for verifier $V_i$---call this assumption $\Gamma$. However, recall that the outputs $y_i$ of verifiers $V_i$ are post-processed by $U'$ to determine the final output of $V$. Thus, it may be in a dishonest prover's interest to intentionally send a rejecting proof $\ket{w_i}$ to $V_i$, even if $q_i$ is a YES instance, since setting $y_i=0$ may result in $U'$ outputting $1$ (whereas perhaps setting $y_i=1$ caused $U'$ to output $0$). Indeed, if $\Gamma$ \emph{did} hold, then  \cref{lem:augmented-circuit} itself maps an arbitrary $\PQMApar$ computation to a single QMA instance (given by verifier $V$), implying $\PQMApar=\QMA$.
    This is unlikely to hold, as it would imply $\QMA=\coQMA$, since $\QMA,\coQMA\subseteq\PQMApar$ (the result $\QMA=\coQMA$ is not believed to be true).
\end{enumerate}

\subsection{Generic Hardness Constructions via a Lifting Lemma}\label{sec:hamiltonian-embedding}

We now proceed by encoding a family of $\DQ$ instances into a local circuit-to-Hamiltonian construction and carefully penalizing the \emph{flag} qubit (not the output qubit!) to encourage the ground space of $\Hw(V)$ to encode correct query answers made by the $\class{D}$ machine to the $\class{Q}$ oracle. Again, we will sometimes consider the concrete case of $\class{D}=\class{P}$ and $\class{Q}=\QMA$, i.e.\ $\PQMApar$, in our discussions, though our lemmas will be stated more generally.

Specifically, we embed the circuits $V$ constructed in \cref{lem:augmented-circuit} into a local circuit-to-Hamiltonian construction (\cref{def:local-mapping}), resulting in a Hamiltonian $\Hw(V)$.
As outlined in \cref{sec:circuit-to-ham} (and also guaranteed more generally by \cref{def:local-mapping}), $\Hw(V)\succeq 0$ is an unfrustrated Hamiltonian with non-empty null space, and since we have only added input penalties to correctly initialize the computation $V$, the ground space of $\Hw(V)$ is spanned by states encoding a computation with the properties stated in \cref{lem:augmented-circuit}. We now comment on two outstanding issues following \cref{lem:augmented-circuit}.

\paragraph{How to fix soundness for \cref{lem:augmented-circuit}.} Recall now that the reduction of \cref{lem:augmented-circuit} was not \emph{sound}, meaning a NO instance of $\PQMApar$ was not necessarily mapped to a ``NO QMA circuit'' $V$. (Formally, it is not necessarily true in the NO case that for all proofs $\ket{w_1}\otimes\cdots\otimes \ket{w_m}$, $V$ rejects.) This is because, intuitively, each of the verifiers $V_i$ in \cref{fig:augmented-circuit} has a potentially different implicit quantifier for its proof register $\ket{w_i}$ (namely, $\exists$ for YES queries $\ket{q_i}$ and $\forall$ for NO queries $\ket{q_i}$).

To address this, our aim here is to force the ground state to encode as many \QMA-query answers as possible as \YES. Intuitively, this is because \NO queries are not a problem; \emph{any} witness $\ket{w_i}$ fed to a verifier $V_i$ in this case is a ``bad'' witness, i.e.\ accepted with ``small'' probability. For a \YES query, however, a cheating prover has the choice of inputting either a ``good'' or ``bad witness'', thus flipping the output bit of verifier $V_i$, and potentially flipping the overall output of $V$. By penalizing the flag qubit, any \YES query $\ket{q_i}$ which incorrectly has $V_i$ outputting $\ket{0}$ is assigned an additional ``unnecessary'' energy penalty, lifting any such history state above the true ground state energy. (Crucially, we have no knowledge of the actual ground state energy itself, and this value encodes the number of \YES and \NO queries.) In short, adding the flag penalty will have the effect of ensuring that the ground space will be spanned by states encoding computation for which as many ``good'' witnesses $\ket{w_i}$ as possible are fed into $V$.\footnote{While we generally call states of the form \cref{eq:history-state} history states, adding a penalty will result in a ground state as superposition of the same vectors, but with weights biased away from the location of the penalty in the time register. The effect of penalties on the ground state and ground state energy of history state Hamiltonians is extensively studied in \cite{Bausch2016a}, and bounds on the ground state energy are generally addressed with a variant of Kitaev's geometrical lemma \cite{Kitaev2002}.}

\paragraph{What about invalid queries?} When $\class{Q}$ is a promise class (such as \QMA), the possibility of a $\DQ$ machine making invalid queries $\ket{q_i}$ to a $\class{Q}$-oracle arises. In such cases, recall that the $\class{Q}$-oracle may answer $0$ or $1$ arbitrarily, and the corresponding $\class{Q}$-verifier $V_i$ may accept the optimal proof $\ket{w_i}$ with arbitrary probability. In terms of our Hamiltonian construction, this means there may be multiple valid history states (i.e.~encoding all valid query answers correctly) with energy exponentially close to the ground state energy. However, as discussed in \cite{GY16_2}, there will be an energy gap between the low energy space (which answers all valid queries \emph{correctly}, yielding precisely the UNSAT penalty of the history state Hamiltonian, cf.\ \cite{Bausch2016a}), and any history state which encodes at least one valid query's answer \emph{incorrectly}. (We remark that the comparison with~\cite{GY16_2} is not entirely accurate, as there Ambainis' query Hamiltonian~\cite{Ambainis2013} construction is used, which is not employed here. For this reason, the discussion of this paragraph is for intuition; in our proofs here, we do not leverage the results of \cite{GY16_2} regarding invalid queries, but work from first principles.)

\paragraph{The Lifting Lemma.} The main lemma of this section which encompasses all of the open points raised is as follows.
Let $\Delta(A)$ denote the spectral gap of Hermitian operator $A$ (i.e.~the gap between the two smallest distinct eigenvalues of $A$):
\begin{align}
\Delta(A) \coloneqq  \lambda_1(A)-\lmin(A).
\end{align}

\begin{lemma}[Lifting Lemma for APX-SIM]\label{lem:lift}
Let $x\in\set{0,1}^n$ be an instance of an arbitrary $\DQ$ problem, $U$ a $\DQ$ machine deciding $x$, and $V$ the verification circuit output by \cref{lem:augmented-circuit}.
Fix a local circuit-to-Hamiltonian mapping $\Hw$, and assume the notation in \cref{def:local-mapping}.
Suppose, there exists a computable function $\alpha:\nats\mapsto\nats$, such that, for any $\epsilon$ satisfying
\begin{equation}
    0\leq \epsilon \leq \frac{1}{\alpha}\left(\frac{1}{\alpha}+\frac{12\norm{M_2}^2}{\Delta(\Hw(V))}\right)\left(\frac{8m^2}{3\g(m,T)}\right),
\end{equation}
Then the Hamiltonian
$H\coloneqq \alpha(n)\Hw(V)+ M_2$
satisfies:
\begin{itemize}
    \item If $x$ is a YES instance, then for all $\ket{\psi}$ with $\bra{\psi}H\ket{\psi}\leq\lmin(H)+\frac{1}{\alpha^2}$,
        \[
            \bra{\psi}M_1\ket{\psi}\leq g(m,T)\cdot m\cdot\max(1-c+\epsilon,s)+\frac{12\norm{M_2}}{\alpha\Delta}.
        \]
    \item If $x$ is a NO instance, then for all $\ket{\psi}$ with $\bra{\psi}H\ket{\psi}\leq\lmin(H)+\frac{1}{\alpha^2}$,
        \[
            \bra{\psi}M_1\ket{\psi}\geq g(m,T)\left(1-m\cdot\max(1-c+\epsilon,s)\right)-\frac{12\norm{M_2}}{\alpha\Delta}.
        \]
\end{itemize}

\end{lemma}
\noindent \emph{Remarks.}
\begin{enumerate}
    \item The Lifting Lemma's sole degree of freedom is the function $\alpha$. All other quantities appearing, i.e.\ $\norm{M_2}$ (flag penalty, \cref{lem:augmented-circuit}), $\Delta$ (spectral gap of $\Hw(V)$), $g(m,T)$ (weight of final time step $T$ in history state, \cref{eqn:sim3}), $c$ (completeness) and $s$ (soundness), $m$ (number of Q-queries), are fixed functions stemming from the choice of circuit-to-Hamiltonian construction (\cref{def:local-mapping}) and classes D and Q.
    \item If we think of $M_1$ as playing the role of observable $A$ from \cref{def:forall-apx-sim} (where we follow the notation from \cref{def:local-mapping}), the Lifting Lemma brings us very close to obtaining $\DQ$-hardness of APX-SIM for the class of Hamiltonians produced by $\Hw$. (Recall here that, by the two conditions imposed on the circuit-to-Hamiltonian mapping, \cref{def:conforming,def:local-mapping}, $M_2$ conforms to $\Hw(U)$, and thus $H$ preserves the class/structure of Hamiltonians produced by $\Hw$.) To formally ``instantiate'' such a hardness result for a fixed class $\DQ$, it remains to choose an appropriate function $\alpha(n)$ ($n$ the input size) so that $\epsilon$ is sufficiently small so as to create an appropriate promise gap between the YES and NO thresholds in the Lifting Lemma.
    \item Note that for $m\in \omega(1)$, the standard completeness/soundness parameters of $2/3$ and $1/3$ will not suffice to create a promise gap in the Lifting Lemma. Thus, error reduction for the VQClass Q to completeness $1-1/\poly(n)$ versus soundness $1/\poly(n)$ appears necessary\footnote{Intuitively, this is because \cref{fig:augmented-circuit} runs all verifiers $V_i$ on proofs, and then $U'$ is guaranteed to output the correct answer only if \emph{all verifiers output their respective correct query answers}. If completeness and soundness are $2/3$ and $1/3$, respectively, the odds of this happening, even given optimal local proofs to each $V_i$, scales in the worst case as only $(2/3)^m$.} to apply our construction as a black box. This, in particular, means our construction cannot \emph{a priori} be applied with $\class{Q}=\class{StoqMA}$, since the latter is not known to have error reduction. (Note that recently, it was shown by Aharonov, Grilo, and Liu~\cite{AGL20} that if error reduction to $1-1/\omega(\poly(n))$ completeness versus constant soundness for StoqMA holds, then $\class{StoqMA}=\class{MA}$. However, for clarity, our construction only requires $1-1/\poly(n)$ completeness, so even if exponential error reduction is not possible in $\class{StoqMA}$, this does not rule out the regime our construction operates in.)
\end{enumerate}

\paragraph{Proof of the Lifing Lemma.} We prove the Lifting Lemma (\cref{lem:lift}) in the following series of lemmas, by showing i.\ that any low energy state of $H$ is close to a uniform history state in kernel $\Null{\Hw}$ (which is nonempty, as assumed in \cref{sec:circuit-to-ham}), ii.\ that uniform history state has all \emph{valid} queries answered correctly, and iii.\ that there is a large enough overlap with the overall output of the computation, such that the associated measurement problem of said site is hard, which is what we wish to show.\\

\noindent\emph{Step i. Low energy states of $H$ are close to uniform history states.} We first show that for a suitably-chosen function $\alpha$ (throughout, we assume the notation of \cref{lem:lift}), any low energy state with respect to $H$ has a ground state of $\Hw(V)$ which is nearby. Henceforth, we shall often drop the dependence on $n$ for simplicity, e.g.\ $\alpha(n)$ will be written $\alpha$, and $\Delta(A)$ denotes the spectral gap of Hermitian operator $A$. (For clarity, we define the spectral gap as difference between the two smallest \emph{distinct} eigenvalues of $A$.)
\begin{lemma}\label{Lemma:Trace_Distance}
    For brevity, define shorthand $\Delta$ for $\Delta(\Hw(V))$. Fix any function $\alpha:\nats\mapsto\nats$ such that
    \begin{equation}\label{eqn:6}
        \alpha > \max\left(\frac{4\norm{M_2}}{\Delta},\frac{\Delta}{3\norm{M_2}^2},1\right),
    \end{equation}
    and any $\delta\leq 1/\alpha^2$. Then, for any $\ket{\psi}$ such that $\bra{\psi}H\ket{\psi}\leq \lmin(H)+\delta$, there exists a uniform history state $\ket\phi \in \Null{\Hw(V)}$ such that
    \begin{equation}\label{eqn:trbound2}
        \trnorm{\ketbra{\psi} - \ketbra{\phi}} \leq\frac{12\norm{M_2}}{\alpha\Delta}
    \end{equation}
    and where $\ket{\phi}$ has energy
    \begin{equation}\label{eqn:5}
        \bra{\phi}H\ket{\phi}\leq \lmin(H)+\delta+\frac{12\norm{M_2}^2}{\alpha\Delta}.
    \end{equation}
\end{lemma}
\begin{proof}
     The ground state deviation bound of the Extended Projection Lemma (\cref{l:kkr}), and Equations~(\ref{eqn:enorm}) and (\ref{eqn:6}) imply the first claim via
     \[
        \trnorm{\ketbra{\psi} - \ketbra{\phi}} \leq 2\frac{\left(\norm{M_2}+\sqrt{\norm{M_2}^2+\frac{\Delta}{\alpha}}\right)}{\alpha\Delta-2\norm{M_2}}
        \leq\frac{6\norm{M_2}}{\alpha\Delta-2\norm{M_2}}
        \leq\frac{12\norm{M_2}}{\alpha\Delta}.
     \]
     Here we have used \cref{eqn:6} to go from the second last to last expression in this chain of inequalities.

     A similar calculation using the third claim of \cref{l:kkr} yields the second claim of this lemma.
\end{proof}


\noindent\emph{Step ii. All valid queries answered correctly.} So far we have shown that all low-energy states with respect to $H$ are close to unbiased history states, i.e.\ the states in the kernel of $\Hw(V)$. These history states guarantee a correct computation in superposition, as given in \cref{eq:history-state}. Yet $M_2$ penalizes the flag qubit, as formalized in \cref{lem:augmented-circuit}. What does this imply for the witness $\ket{\wm}$? We now show that given a sufficiently small $\delta$, the history state close to $\ket\psi$ encodes a series of queries such that all \emph{valid} queries are answered correctly---in other words, that as many of the valid queries as possible are answered as \YES.

For this, we first require the following lemma, which works in ``circuit-world'' (i.e.\ in \cref{fig:augmented-circuit}, before a circuit-to-Hamiltonian construction is applied).
\begin{lemma}\label{lem:lowhelp}
Assume the notation of \cref{lem:augmented-circuit}, which showed
    \[
        \trace\left(\ketbra{\psi}{\psi}\cdot\ketbra{0}{0}_{\qflag}\right)=\sum_{y\in \set{0,1}^m}\pyw\cdot \cos^2\left(\HW(y)\ang\right),
    \]
    where $\ket{\psi}$ denoted the output of $V$ given joint proof $\ket{\wm}$, and $\pyw=\Pr\left(\bigwedge_{i=1}^m V_i \text{ outputs } y_i\mid \ket{\wm}\right)$. Suppose there exists an $i\in\set{1,\ldots, m}$ and $\epsilon>0$ such that $\ket{\wm}$ is ``$\epsilon$-suboptimal on proof $i$'', meaning there exists a local proof $\ket{w'_i}$
     such that
    \begin{equation}\label{eqn:perturb}
        \Pr(V_i \text{ outputs } 1 \mid \ket{\wm})=\Pr(V_i \text{ outputs } 1 \mid \ket*{w_i'})-\epsilon.
    \end{equation}
    Then there exists a proof $\ket{\wm'}=\ket{w_1'}\otimes\cdots\otimes\ket{w_m'}$ which causes $V$ to output $\ket{\psi'}$ satisfying
    \[
        \trace\left(\ketbra{\psi}{\psi}\cdot\ketbra{0}{0}_{\qflag}\right)\geq \trace\left(\ketbra{\psi'}{\psi'}\cdot\ketbra{0}{0}_{\qflag}\right)+\frac{3}{8m^2}\epsilon.
    \]
\end{lemma}
\begin{proof}
    We proceed in two steps: first, we bound the rate at which $\cos^2(\sqrt{3}x/(2m))$ changes with integer increments in $x$. Second, we apply an exchange argument to swap out $\ket{\psi'}$ for $\ket{\psi}$.

    \paragraph{Bounding the cosine function.} Begin by writing
    \[
        \cos^2\left(\angf x\right)=\frac{1}{2}\left(1+\cos\left(\angft x\right)\right) \eqqcolon \frac{1}{2}(1+f(x)).
    \]
    Since $\sqrt{3}x/(2m)<\pi/2$ for $x\in[0,m]$, $f(x)$ is non-negative and concave down over $x\in[0,m]$. Thus, for $\epsilon>0$ satisfying $0\leq (x+\epsilon)\sqrt{3}/m\leq \pi/2$, we have $f(x+\epsilon)\leq f(x)+\epsilon f'(x)$, implying
    \begin{eqnarray*}
        f(x)&\geq& f(x+1)- f'(x)\\
        &=&f(x+1)+\angft\sin\left(\angft x\right)\\
        &\geq& f(x+1)+\angft\left(\angft x - \frac{1}{6}\left(\angft\right)^3x^3\right)\\
        &\geq& f(x+1)+\angftt x\\
        &\geq& f(x+1)+\angftt
    \end{eqnarray*}
    where the third statement follows from the Taylor series for sine, the fourth since
    \[
        \frac{1}{6}\left(\angft\right)^3x^3\leq \frac{1}{2}\angft x
    \]
    for $x\in[0,m]$, and the last when $x\geq 1$. As for the case of $x=0$ (which breaks the argument above since $f'(0)=0$), a direct Taylor series truncation for cosine yields $f(0)\geq f(1)+3/(4m^2)$. We conclude that for any $x,y\in\set{0,1}^{m}$ such that $\HW(y)>\HW(x)$,
    \begin{equation}\label{eqn:drop}
        \cos^2\left(\HW(x)\ang\right)-\cos^2\left(\HW(y)\ang\right)\geq \frac{3}{8m^2}.
    \end{equation}

    \paragraph{The exchange argument.} We now wish to run an exchange argument, meaning we wish to swap out $\ket{\wm}$ for a better proof $\ket*{\wm'}$ to get the desired statement. Since we cannot assume $\ket{\wm}$ is in tensor product across proofs, this argument needs to be done recursively (and in particular, via conditional probabilities). For this, assume without loss of generality that $i=1$ in the claim. For brevity, define $C(x)\coloneqq \cos^2(\sqrt{3}x/(2m))$, and we use shorthand such as \begin{align*}
    \Pr(y_1y_2=01) \quad&\text{to mean}\quad \Pr(V_1\text{ outputs }0\text{ and } V_2\text{ outputs }1\mid \ket{\wm}), \quad\text{and}\\
    \prs(y_1=1) \quad&\text{to mean}\quad \Pr(V_1\text{ outputs }1\mid \ket*{w_1'}).
    \end{align*}
    We recall that $\ket{w_1'}$ is defined as in \cref{eqn:perturb}. Our exchange argument then proceeds in two steps:
    \begin{itemize}
        \item Step 1 unrolls a recurrence ``from left to right'', ``partially'' achieving the desired exchange as it goes. At the end of this step, the expression we obtain contains a mixture of terms from the distribution induced by our desired target proof, $\ket*{\wm'}$, and ``offset'' terms $\epsilon_z$.
        \item Step 2 then unrolls a recurrence from ``right to left'', this time carefully pairing terms together and using \cref{eqn:drop} repeatedly to eliminate all but one $\epsilon_z$ term, which is the $\epsilon$ appearing in the statement of the final claim.
    \end{itemize}

    \noindent \emph{Step 1: Unrolling a recurrence from ``left'' to ``right''.}
    Since we are assuming $i=1$, we have the following recursive decomposition (for pedagogical reasons, we demonstrate the first two steps of the recurrence explicitly):
    starting from the first equation in \cref{lem:lowhelp}, $\trace\left(\ketbra{\psi}{\psi}\cdot\ketbra{0}{0}_{\qflag}\right)=:A$, we have
    \begingroup
    \allowdisplaybreaks
    \begin{eqnarray*}
        A&=&\sum_{y\in\set{0,1}^m}\pyw C(y)\\
        &=&\Pr(y_1=0)\left(\sum_{y'\in \set{0,1}^{m-1}}\Pr(y'\mid y_1=0) C(0y')\right)+
        \Pr(y_1= 1)\left(\sum_{y'\in \set{0,1}^{m-1}}\Pr(y'\mid y_1=1) C(1y')\right)\\
        &=&({\prs}(y_1=0)+\epsilon)\left(\sum_{y'\in \set{0,1}^{m-1}}\Pr(y'\mid y_1=0) C(0y')\right)+\\
        &&({\prs}(y_1=1)-\epsilon)\left(\sum_{y'\in \set{0,1}^{m-1}}\Pr(y'\mid y_1=1) C(1y')\right)\\
        &=&({\prs}(y_1=0)+\epsilon)\left[\Pr(y_2=0\mid y_1=0)\left(\sum_{y'\in \set{0,1}^{m-2}}\Pr(y'\mid y_1y_2=00) C(00y')\right)+\right.\\
        &&\hspace{29mm}\left.\Pr(y_2=1\mid y_1=0)\left(\sum_{y'\in \set{0,1}^{m-2}}\Pr(y'\mid y_1y_2=01) C(01y')\right)\right]+\\
        &&({\prs}(y_1=1)-\epsilon)\left[\Pr(y_2=0\mid y_1=1)\left(\sum_{y'\in \set{0,1}^{m-2}}\Pr(y'\mid y_1y_2=10) C(10y')\right)+\right.\\
        &&\hspace{29mm}\left.\Pr(y_2=1\mid y_1=1)\left(\sum_{y'\in \set{0,1}^{m-2}}\Pr(y'\mid y_1y_2=11) C(11y')\right)\right]\\
        &=&({\prs}(y_1=0)+\epsilon)\left[(\prs(y_2=0)+\epsilon_0)\left(\sum_{y'\in \set{0,1}^{m-2}}\Pr(y'\mid y_1y_2=00) C(00y')\right)+\right.\\
        &&\hspace{29mm}\left.(\prs(y_2=1)-\epsilon_0)\left(\sum_{y'\in \set{0,1}^{m-2}}\Pr(y'\mid y_1y_2=01) C(01y')\right)\right]+\\
        &&({\prs}(y_1=1)-\epsilon)\left[(\prs(y_2=0)+\epsilon_1)\left(\sum_{y'\in \set{0,1}^{m-2}}\Pr(y'\mid y_1y_2=10) C(10y')\right)+\right.\\
        &&\hspace{29mm}\left.(\prs(y_2=0)-\epsilon_1)\left(\sum_{y'\in \set{0,1}^{m-2}}\Pr(y'\mid y_1y_2=11) C(11y')\right)\right]\\
    \end{eqnarray*}
    \endgroup
    where $\epsilon_0,\epsilon_1\geq0$ since they denote the differences in acceptance probability for $V_2$ between the optimal proof $\ket{w'_2}$ and the current proof $\ket{\wm}$. We can continue unrolling the recurrence using terms $\epsilon_{z}$, so that in step $r$ of the recurrence (i.e.\ above, we have completed $2$ steps), $z\in\set{0,1}^{r-1}$. Here, $z$ denotes the string we have conditioned on thus far in the previous $r-1$ steps, meaning
    \[
        \epsilon_z = \Pr(y_r=0\mid y_1\cdots y_{r-1}=z)-\prs(y_r=0)=\prs(y_r=1)-\Pr(y_r=1\mid y_1\cdots y_{r-1}=z).
    \]
    Thus, unrolling the full recurrence yields:
    \begin{eqnarray}\label{eqn:doneright}
    \trace\left(\ketbra{\psi}{\psi}\cdot\ketbra{0}{0}_{\qflag}\right)&=&\sum_{y\in\set{0,1}^m}\left[\prod_{i=1}^m \left(\prs(y_i)+(-1)^{y_i}\epsilon_{y_1\cdots y_{i-1}}\right)\right]C(y).
    \end{eqnarray}
    \emph{Step 2: Unrolling a recurrence from ``right'' to ``left''.} We now iteratively unroll a recurrence in Equation~(\ref{eqn:doneright}), this time from right to left, each time removing a round of additive terms $\epsilon_z$ by carefully matching up certain expressions, until we are only left with the $\epsilon$ term. Namely, $\trace\left(\ketbra{\psi}{\psi}\cdot\ketbra{0}{0}_{\qflag}\right)$ equals:
    \begin{align*}
        &\sum_{y_1\cdots y_{m-1}}\left[\prod_{i=1}^{m-1} \left(\prs(y_i)+(-1)^{y_i}\epsilon_{y_1\cdots y_{i-1}}\right)\right]\cdot\\
        &\hspace{15mm}\left((\prs(y_m=0)+\epsilon_{y_1\cdots y_{m-1}})C(y_1\cdots y_{m-1}0)+(\prs(y_m=1)-\epsilon_{y_1\cdots y_{m-1}})C(y_1\cdots y_{m-1}1)\right)\\
        \geq&\sum_{y_1\cdots y_{m-1}}\left[\prod_{i=1}^{m-1} \left(\prs(y_i)+(-1)^{y_i}\epsilon_{y_1\cdots y_{i-1}}\right)\right]\cdot\\
        &\hspace{15mm}\left(\prs(y_m=0)C(y_1\cdots y_{m-1}0)+\prs(y_m=1)C(y_1\cdots y_{m-1}1)+\frac{3\epsilon_{y_1\cdots y_{m-1}}}{8m^2}\right)\\
        \geq&\sum_{y_1\cdots y_{m-1}}\left[\prod_{i=1}^{m-1} \left(\prs(y_i)+(-1)^{y_i}\epsilon_{y_1\cdots y_{i-1}}\right)\right]\cdot\\
        &\hspace{15mm}\left(\prs(y_m=0)C(y_1\cdots y_{m-1}0)+\prs(y_m=1)C(y_1\cdots y_{m-1}1)\right)
    \end{align*}
    where the second statement follows from \cref{eqn:drop}, and the third since $\epsilon_{y_1\cdots y_{i-1}}\geq 0$. (Note that all terms $\prs(y_i)+(-1)^{y_i}\epsilon_{y_1\cdots y_{i-1}}$ are also non-negative, since each such term was itself a probability.) Repeating this step $m-1$ times in total, we thus have
    \begin{eqnarray*}
        \trace\left(\ketbra{\psi}{\psi}\cdot\ketbra{0}{0}_{\qflag}\right)&\geq&\left(\prs(y_1=0)+\epsilon\right)\sum_{y_2\cdots y_m}\left(\prod_{i=2}^m\prs(y_i)\right)C(0y_2\cdots y_m)+\\
        &&\left(\prs(y_1=1)-\epsilon\right)\sum_{y_2\cdots y_m}\left(\prod_{i=2}^m\prs(y_i)\right)C(1y_2\cdots y_m)\\
        &=&\left(\sum_{y\in\set{0,1}^m}\left(\prod_{i=1}^m\prs(y_i)\right)C(y)\right)+\\
        &&\hspace{20mm}\epsilon\left(\sum_{y_2\cdots y_m}\left(\prod_{i=2}^m\prs(y_i)\right)\right)(C(0y_2\cdots y_m)-C(1y_2\cdots y_m))\\
        &\geq&\left(\sum_{y\in\set{0,1}^m}\left(\prod_{i=1}^m\prs(y_i)\right)C(y)\right)+
        \frac{3\epsilon}{m^2}\left(\sum_{y_2\cdots y_m}\left(\prod_{i=2}^m\prs(y_i)\right)\right)\\
        &=&\Tr\left(\ketbra{\psi'}\cdot\ketbra{0}_{\qflag}\right)+
        \frac{3\epsilon}{m^2},
    \end{eqnarray*}
where the third statement follows by \cref{eqn:drop} and the last since (1) $\ket{\psi'}$ is the state obtained by applying $V$ to the locally optimal (i.e.\ tensor product) proof $\ket{w_1'}\cdots\ket{w_m'}$, and (2) the latter has, by definition, the property that for all $i\in[m]$, $\ket{w_i'}$ is accepted by $V_i$ with probability $\prs(y_i=1)$.
%
%
%
%
%
\end{proof}

We now translate \cref{lem:lowhelp}, which held in ``circuit-world'', to the following main lemma of Step ii, which gives us the corresponding desired statement in ``Hamiltonian-world'' (i.e.\ after the circuit-to-Hamiltonian construction is applied).

\begin{lemma}\label{lem:low-energy-as-many-yes-as-possible}
Suppose history state $\ket{\phi} \in \Null{\Hw(V)}$ has preimage $\ket{\psiin}=\f^{-1}(\ket{\phi})$ (for bijection $\f$ from \cref{def:local-mapping}), where $\ket{\psiin}$ has proof $\ket{\wm}$. If there exists $\epsilon\geq 0$ and $i\in[m]$ such that $\ket{\wm}$ is $\epsilon$-suboptimal on proof $i$ (as defined in \cref{lem:lowhelp}), then
\begin{equation}\label{eqn:eta}
    \bra{\phi}H\ket{\phi}\geq \lmin(H)+\frac{3\g(m,T)}{8m^2}\epsilon,
\end{equation}
for $\g(m,T)$ defined in \cref{def:local-mapping}.
\end{lemma}
\begin{proof}
Let $\ket{\psiout}=V\ket{\psiin}$ and $\ket{\psiout'}$ denote the output of $V$ (\cref{fig:augmented-circuit}) when all proofs are proofs are optimal (i.e.\ set to $\ket{w_i'}$ in the terminology of \cref{lem:lowhelp}). Recall \cref{eqn:sim} in the \cref{def:local-mapping} of a local circuit-to-Hamiltonian mapping, which said $M_2$ simulates the projector $\ketbra{0}{0}_{\qflag}$ via
    \begin{equation}\label{eqn:here}
            \Tr\left(\ketbra{0}{0}_2(U_TU_{T-1}\ldots U_1)\ketbra{\psiin}(U_TU_{T-1}\ldots U_1)^\dagger\right) = \Tr\big(\ketbra{\phi_T} M_i\big),
    \end{equation}
(since we assumed in \cref{lem:augmented-circuit} that the second output qubit of $V$ is the flag qubit), where $\ket{\phi_T}$ is the history state $\ket{\phi}$ projected down onto time step $T$, which succeeds with probability $\g(m,T)$ (\cref{def:local-mapping}). We thus have
\begin{eqnarray}
    \bra{\phi}H\ket{\phi}&=&\bra{\phi}M_2\ket{\phi}\\
    &=& g(m,T)\bra{\phi_T}M_2\ket{\phi_T}\\
    &=& g(m,T)\trace\left(\ketbra{\psiout}\cdot\ketbra{0}{0}_{\qflag}\right)\\
    &\geq&g(m,T)\left(\trace\left(\ketbra{\psiout'}\cdot\ketbra{0}{0}_{\qflag}\right)+\frac{3}{8m^2}\epsilon\right)\\
    &=&\bra{\phi'}M_2\ket{\phi'}+g(m,T)\left(\frac{3}{8m^2}\epsilon\right)\\
    &\geq&\lmin(H)+\frac{3g(m,T)}{8m^2}\epsilon,
\end{eqnarray}
where the second statement follows from Equation \cref{eqn:sim2}, the third from \cref{eqn:here}, fourth from \cref{lem:lowhelp}, fifth from \cref{eqn:here} and defining $\ket{\phi'}=\f(V^\dagger\ket{\psiout'})$, and the last since $\ket{\phi'}\in\Null{\Hw(V)}$ by the definition of $\f$ in \cref{def:local-mapping} being a bijection.
\end{proof}

\noindent\emph{Step iii. Sufficiently high overlap with computation.} With the previous two steps in hand, we are in a position to argue that any low energy state of $H$ must indeed correctly encode the original $\DQ$ computation, and hence an appropriate measurement of a ground state will read off the $\DQ$ computation's answer.

\begin{lemma} \label{Lemma:A_Variable_Hardness}
    Consider any $\ket{\psi}$ satisfying $\bra{\psi}H\ket{\psi}\leq\lmin+\delta$. If $\delta\leq 1/\alpha^2$ and
    \begin{equation}
         \epsilon < \left(\delta+\frac{12\norm{M_2}^2}{\alpha\Delta}\right)\left(\frac{8m^2}{3\g(m,T)}\right),
    \end{equation}
    then
\begin{itemize}
    \item if $x$ is a YES instance for $\DQ$, then
    \begin{equation}
        \Tr(\ketbra{\psi}M_1)\leq g(m,T)\cdot m\cdot\max(1-c+\epsilon,s)+\frac{12\norm{M_2}}{\alpha\Delta}.
    \end{equation}
    \item if $x$ is a NO instance for $\DQ$, then
    \begin{equation}
        \Tr(\ketbra{\psi}M_1)\geq g(m,T)\left(1-m\cdot\max(1-c+\epsilon,s)\right)-\frac{12\norm{M_2}}{\alpha\Delta}.
    \end{equation}
\end{itemize}
\end{lemma}
\begin{proof}
    We first use \cref{Lemma:Trace_Distance} to map, assuming $\delta\leq 1/\alpha^2$, $\ket{\psi}$ to a history state $\ket{\phi}\in \Null{\Hw(V)}$ such that
    \begin{equation}\label{eqn:step1}
        \trnorm{\ketbra{\psi} - \ketbra{\phi}} \leq\frac{12\norm{M_2}}{\alpha\Delta}\quad\text{and}\quad
        \bra{\phi}H\ket{\phi}\leq \lmin(H)+\delta+\frac{12\norm{M_2}^2}{\alpha\Delta}.
    \end{equation}
    We next use \cref{lem:low-energy-as-many-yes-as-possible} to obtain that, for any $\epsilon$ satisfying
    \[
        \epsilon < \left(\delta+\frac{12\norm{M_2}^2}{\alpha\Delta}\right)\left(\frac{8m^2}{3\g(m,T)}\right),
    \]
    the preimage $\ket{\phiin}=f^{-1}(\ket{\phi})$ contains proof $\ket{\wm}$ which is not $\epsilon$-suboptimal on any proof $i\in[m]$. In words, for any $i\in[m]$, if $V_i$ has optimal acceptance probability $p_i^*$, then
    \[
        \Pr(V_i\text{ outputs 1}\mid \ket{\phiin})\geq p_i^*-\epsilon.
    \]
    Now, by \cref{lem:augmented-circuit}, if $q_i$ is a YES query, $p_i^*\geq c$, if $q_i$ is a NO query, $p_i^*\leq s$, and if $q_i$ is invalid, then $p_i^*$ can only be said to satisfy the trivial bounds $0\leq p_i^*\leq 1$. Letting $\syes,\sno,\sinv$ denote the partition of $[m]$ corresponding to YES, NO, and INVALID queries, it follows via the commutative quantum union bound (\cref{lem:unionbound}) that
    \[
        \Pr\left(\bigwedge_{i\in [m]} V_i \text{ outputs correct answer}\mid\ket{\phiin}\right)\geq 1-\left(\abs{\syes}(1-c+\epsilon)+\abs{\sno}s\right)=:\pgood
    \]
    where note $\abs{\sinv}$ does not appear since any answer to an invalid query is considered correct.

    Finally, since by \cref{def:dTM}, $U'$ in \cref{fig:augmented-circuit} is a deterministic computation, it follows that $V$ correctly accepts (respectively, correctly rejects) with probability at least $\pgood$, given $\ket{\phiin}$, when the $\DQ$ instance $x$ is a YES instance (respectively, NO instance). Recall from \cref{eqn:sim} that measuring $M_1$ on a history state $\ket{\phi}=f(\ket{\phiin})$ simulates a measurement on the output qubit of $V$ (\cref{fig:augmented-circuit}) via
    \begin{equation}
        \Tr\big(\ketbra{0}{0}_1V\ketbra{\phiin}V^\dagger\big) = \Tr\big(\ketbra{\phi_T} M_1\big),
    \end{equation}
    for $\ket{\phi_T} = P_T \ket\phi / \enorm{P_T \ket\phi}$. Thus, if $x$ is a YES instance,
    \begin{eqnarray*}
        \pgood&\leq&\trace\left(V\ketbra{\phiin}{\phiin}V^\dagger\cdot\ketbra{1}{1}_{\qout}\right)\\
        &=&1-\trace\left(V\ketbra{\phiin}{\phiin}V^\dagger\cdot\ketbra{0}{0}_{\qout}\right)\\
        &=&1-\Tr(\ketbra{\phi_T}M_1)\\
        &=&1-\frac{1}{g(m,T)}\Tr(\ketbra{\phi}M_1),
    \end{eqnarray*}
    where the third/fourth statements follow from \cref{def:local-mapping} and Equations~(\ref{eqn:sim3}) and (\ref{eqn:sim2}). Via an analogous argument for the NO case, we conclude:
\begin{itemize}
    \item If $x$ is a YES instance for $\DQ$, then
        \[
            \Tr(\ketbra{\phi}M_1)\leq g(m,T)(1-\pgood)\leq g(m,T)\cdot m\cdot\max(1-c+\epsilon,s).
        \]
    \item If $x$ is a NO instance for $\DQ$, then
    \[
        \Tr(\ketbra{\phi}M_1)\geq g(m,T)\left(1-m\cdot\max(1-c+\epsilon,s)\right).
    \]
\end{itemize}
This was for history state $\ket{\phi}$.
 The claim now follows for the original state $\ket{\psi}$ in the claim by combining \cref{eqn:step1} with H\"{o}lder's inequality.
\end{proof}

The Lifting Lemma now follows; we restate it below for convenience.
\begin{replemma}{lem:lift}[Lifting Lemma]
Let $x\in\set{0,1}^n$ be an instance of an arbitrary $\DQ$ problem, $U$ a $\DQ$ machine deciding $x$, and $V$ the verification circuit output by \cref{lem:augmented-circuit}. Fix a local circuit-to-Hamiltonian mapping $\Hw$, and assume the notation in \cref{def:local-mapping}. Then, there exists a computable function $\alpha:\nats\mapsto\nats$ such that, for any $\epsilon$ satisfying
\begin{equation}
    0\leq \epsilon \leq \frac{1}{\alpha}\left(\frac{1}{\alpha}+\frac{12\norm{M_2}^2}{\Delta(\Hw(V))}\right)\left(\frac{8m^2}{3\g(m,T)}\right),
\end{equation}
the Hamiltonian
$H\coloneqq \alpha(n)\Hw(V)+ M_2$
satisfies:
\begin{itemize}
    \item If $x$ is a YES instance, then for all $\ket{\psi}$ with $\bra{\psi}H\ket{\psi}\leq\lmin(H)+\frac{1}{\alpha^2}$,
        \[
            \bra{\psi}M_1\ket{\psi}\leq g(m,T)\cdot m\cdot\max(1-c+\epsilon,s)+\frac{12\norm{M_2}}{\alpha\Delta}.
        \]
    \item If $x$ is a NO instance, then for all $\ket{\psi}$ with $\bra{\psi}H\ket{\psi}\leq\lmin(H)+\frac{1}{\alpha^2}$,
        \[
            \bra{\psi}M_1\ket{\psi}\geq g(m,T)\left(1-m\cdot\max(1-c+\epsilon,s)\right)-\frac{12\norm{M_2}}{\alpha\Delta}.
        \]
\end{itemize}
\end{replemma}
\begin{proof}
    Follows immediately from \cref{Lemma:A_Variable_Hardness} by setting $\delta=1/\alpha^2$.
\end{proof}
As mentioned earlier, note that to apply \cref{lem:lift} to particular classes $\DQ$, it remains to select $\alpha$ appropriately so that \cref{lem:lift} creates the desired promise gap between YES and NO cases.

\subsection{Applying the Lifting Lemma}\label{sscn:other}

We now employ the Lifting Lemma (\cref{lem:lift}) to obtain hardness results for various complexity classes and types of circuit-to-Hamiltonian constructions. For clarity, the lemma applies for any local circuit-to-Hamiltonian construction (\cref{def:local-mapping}) and class $\DQ$ for deterministic class D (\cref{def:dTM}) and QVClass Q (\cref{def:qV}). What is required to apply it is to set $\alpha$ (relative to the other fixed quantities, which are $\norm{M_2}$ (flag penalty, \cref{lem:augmented-circuit}), $\Delta$ (spectral gap of $\Hw(V)$), $g(m,T)$ (weight of final time step $T$ in history state, \cref{eqn:sim3}), $c$ (completeness) and $s$ (soundness), $m$ (number of Q-queries)) so that the desired problem gap is obtained in \cref{lem:lift}.
An important point that for a given \qvc we must be able to choose the completeness and soundness bounds must satisfy
\begin{align}
m\cdot\max(1-c+\epsilon,s) = O\left(\frac{1}{\poly(n)}\right).
\end{align}
To be able to do this we need to be able to amplify the \qvc by the appropriate amount.
This is not always possible.

By construction, \cref{lem:lift} then yields a hardness result in which the constructed Hamiltonian satisfies all structural constraints (e.g.\ locality, geometry, etc) of the local circuit-to-Hamiltonian construction used as a starting point.


\paragraph{Application 1: $\PQMApar$ and 1D Hamiltonian.} In~\cite{GY16_2,GPY20}, it was shown that $\forall$-APX-SIM (and hence APX-SIM; as mentioned in \cref{sec:apx-sim}, $\forall$-APX-SIM reduces to APX-SIM trivially) is $\PQMApar$-complete for $5$-local Hamiltonians $H$ (\cite{GY16_2}, on arbitrary geometry, via Kitaev's circuit-to-Hamiltonian construction~\cite{KSV02}) and 1D 8-dimensional systems (\cite{GPY20}, via Hallgren, Nagaj, and Narayanaswami's 1D construction~\cite{HNN13}). Here, we re-derive both results directly via the Lifting Lemma.
\begin{corollary}\label{Corollary:1D}
    $\forall$-APX-SIM, and hence APX-SIM, are $\PQMApar$-hard for $5$-local Hamiltonians on qubits or 1D $8$-dimensional Hamiltonians, where in both cases the observable is $1$-local and $\delta\leq 1/\poly(n)$, and $b-a\geq 1/\poly(n)$, for $n$ the input size.
\end{corollary}
\begin{proof}[Proof sketch]
    Due to the generality of \cref{lem:lift}, we may treat both cases simultaneously. Namely, we start by taking the local circuit-to-Hamiltonian mapping in either \cite{KSV02} ($5$-local case) or \cite{HNN13} ($1$D $8$-dimensional), and make one modification---we drop the term $\hout$. (Normally, $\hout$ penalizes rejecting computations. In our setting, however, this task is delegated to the observable $A$ in APX-SIM.) For example, for concreteness, for~\cite{KSV02} we set $\Hw(U)=\hin+\hprop+\hstab$ (for $\hin,\hprop,\hstab$ defined in~\cite{KSV02}; their precise definitions are not needed here). In both cases, $M_1$ and $M_2$ are (properly encoded) $1$-local rank-$1$ projector ontos $\ketbra{0}{0}$ at time step $T$ on the output and flag qubits, respectively; thus, $\norm{M_2}=1$. As mentioned in \cref{sec:circuit-to-ham}, $\Delta$ scales as $1/T^2$, and $g(m,T)=1/(1+T)$ in both cases.
    Since $\class{D}=\PP$, the number of \QMA queries is $m\in\poly(n)$ ($n$ the input size); thus, setting completeness and soundness for QMA of $c\geq 1-1/\poly(n)$ and $s\leq 1/\poly(n)$, and $\alpha$ to be a sufficiently large fixed polynomial, we obtain an inverse polynomial promise gap in \cref{lem:lift}. Finally, all functions considered ($g(m,T)$, $\Delta$) and the reduction itself run in time polynomial in the input size, $n$. Thus, \cref{lem:lift} yields the claimed hardness results.
\end{proof}

\begin{table}[t]
    \hspace*{-0.4cm}
    \begin{tabular}{lllll}
        \toprule
        \tbhead{circuit-to-Ham.}{construction} & interaction topology & \tbhead{measurement}{precision} & APX-SIM variant & hardness \\
        \midrule
        \cite{Oliveira2008} & 2D planar, n-n, local dim 2 & \multirow{2}{*}{$\displaystyle \frac{1}{\poly}$} & \multirow{2}{*}{APX-SIM} & \multirow{2}{*}{$\displaystyle \PP^{\parallel \QMA}$} \\
        \cite{Aharonov2009} & 1D line, n-n, local dim 12 & & & \\
        \midrule
        \cite{Fefferman2016} & 3-local, local dim 2 & $1/\exp$ & APX-SIM & \PSPACE \\
        \midrule
        \cite{Bausch2016} & t-i, 1D line, n-n, local dim 44$^\dagger$ & \multirow{2}{*}{$\displaystyle \frac{1}{\poly}$} & \multirow{2}{*}{TI-APX-SIM} & $\displaystyle \EXP^{\parallel \QMA}$ \\
        \cite{Bausch2017} & t-i, 3D fcc lattice, 4-local, local dim 4 & & & $\EXP^{\parallel \QMA}$ \\
        \midrule
        \cite{Kohler2020} & t-i, 1D line, n-n, local dim 42 & $1/\exp$ & TI-APX-SIM & \PSPACE \\
         \bottomrule
    \end{tabular}
    \caption{Hardness of APX-SIM variants for various families of many-body systems.
    Measure precision is in terms of the system (Hamiltonian) size, not the input size.
    n-n stands for 2-local nearest-neighbour interactions, and t-i abbreviates translationally-invariant systems.
    Since $\PreciseQMA=\PSPACE$, and the latter is low for itself, $\PP^{\parallel\PreciseQMA} = \PSPACE$.\\
    $^\dagger$) By \cref{Corollary:1D_TI_NN_Hardness}, two extra symbols are necessary as compared to the raw construction in \cref{tab:uwHs}, increasing the local dimension from 42 to 44.}
    \label{tab:apx-sim-hardness}
\end{table}

\paragraph{Application 2: Translationally-invariant systems.} As alluded to in (e.g.) \cref{def:TIAPX}, in this setting we are only allowed $\BigO(\polylog(N))$ bits to describe a Hamiltonian on $N$ qubits, which is achieved (e.g.) in 1D by repeating a single fixed $2$-local term on a 1D chain consisting of $N$ $d$-dimensional systems ($d\gg 2$ in existing constructions~\cite{Gottesman2009,Bausch2016}). For such systems, we show the following.
\begin{corollary}\label{Corollary:1D_TI_NN_Hardness}
    $\forall$-TI-APX-SIM (and hence APX-SIM) is $\EXP^{\parallel \QMA}$-hard for a nearest neighbour, translationally invariant Hamiltonian on qudits of local dimension $44$, for a 1-local observable $A$, and $\delta=1/\exp(n)$, $b - a =\Omega(1/\exp(n))$.
\end{corollary}
\begin{proof}[Proof sketch]
    Let $x\in\set{0,1}^n$ be the input given to the EXP machine. The construction is similar to \cref{Corollary:1D}, but with some additional important observations.
    \begin{itemize}
        \item First, it is crucial to note that the circuit $V$ constructed in \cref{lem:augmented-circuit} takes in as input $x$, and takes in as proof $\ket{\wm}$. Since the EXP machine can produce QMA queries $\ket{q_i}$ of length \emph{exponential} in $n$ (since the EXP machine itself runs in $\exp(n)$ time), this means that although $x$ is $n$ bits, the proof $\ket{\wm}$ is on $\exp(n)$ qubits and the verifiers $V_i$ are exponential-time. This exactly matches what is expected for a ``\QMAEXP''-setup. (Recall, however, that while $V$ looks like a \QMAEXP verifier, taken on its own, it violates the soundness requirements of the reduction.)
        \item With $V$ in hand, we apply \cite{Gottesman2009} as our local circuit-to Hamiltonian construction (where the local terms for this construction are describable in $\BigO(1)$ bits).
        However, there is a catch---if we were to write out the full circuit of \cref{lem:augmented-circuit} during the reduction, it would take $\exp(n)$ time (since recall $V$ has $\exp(n)$ size).
        Yet this is not an issue: \cite{Gottesman2009} is able to implicitly implement the action of $V$ from \cref{lem:augmented-circuit}, not as an explicit circuit, but via a \emph{Quantum Turing Machines (QTM)} (i.e.\ each $V_i$ is given via a QTM, and these are called as subroutines by a global QTM taking in $x$ and $\ket{\wm}$). As in \cite{Gottesman2009}, these QTMs are by definition described under the hood by finite state machines whose descriptions are $\BigO(1)$-size\footnote{While this may seem odd at first glance, it holds for a simple reason---a Turing machine deciding a language L is specified independently of any particular input $x\in\set{0,1}^*$. Thus, its description cannot depend on the size of any particular input it is fed. This is intuitively why in 1D translationally invariant constructions such as \cite{Gottesman2009}, the local dimension can be $\BigO(1)$, and the only parameter the length of the chain (which encodes the input $x$ in unary).}. In this way, in contrast to what is stated in the claim of \cref{lem:augmented-circuit}, in this case we can implement \cref{lem:augmented-circuit} in $\poly(n)$ time (as opposed to polynomial in $U$ and $V$).
    \end{itemize}
    As per \cite{Gottesman2009}, $A=M_1$ and $M_2$ can now be chosen as constant-norm $2$-local observables. Since, however, the length of the computations in EXP and verifiers $V_i$ are exponential, we now require $\Delta,g(m,T)$ scaling as $1/\exp(n)$ (since $m$ and $T$ are in $\BigO(\exp(n))$). Similarly, since $m\in\exp(n)$, we use standard error reduction on the verifiers $V_i$ to ensure $c\geq 1-1/\poly(N)$ (i.e.\ $1-1/\exp(n)$) and $s\leq 1/\poly(N)$ (i.e.\ $s\leq 1/\exp(n)$). Setting $\alpha$ to be a sufficiently large fixed polynomial in $N$ (i.e.\ a fixed exponential in $n$), we obtain a $1/\poly(N)$ (i.e.\ $1/\exp(n)$) promise gap in \cref{lem:lift}. Finally, due to our use of QTMs rather than explicit circuits, the reduction runs in time polynomial in $n$.

    The reduced local dimension can be obtained by using \cite{Bausch2016}'s translationally-invariant construction instead of \cite{Gottesman2009}; the rest of the argument is the same.
    We obtain a $1$-local observable $A$ by incrementing the local dimension slightly (from $42$ to $44$; two extra letters for the quantum Thue system suffice to signal that the computation halted, as can be seen easily).
\end{proof}
Results for higher-dimensional lattices can be trivially obtained by repeating the Hamiltonian in a translationally-invariant fashion; a qualitatively different result for a 3D crystal lattice is the following:
\begin{corollary}\label{Corollary:3D_TI_NN_Hardness}
	$\forall$-TI-APX-SIM (and hence APX-SIM) is $\EXP^{\parallel \QMA}$-hard for a $4$-local, translationally invariant Hamiltonian on qudits of local dimension $4$ on a 3D fcc lattice, for a 1-local observable $A$, $\delta=1/\exp(n)$, $b - a =\Omega(1/\exp(n))$.
\end{corollary}
\begin{proof}[Proof idea.]
	Follows as in \cref{Corollary:1D_TI_NN_Hardness}, but using \cite{Bausch2017}.
\end{proof}

\paragraph{Application 3: Exponential Precision.}
We now consider the case studied in \cite{Fefferman2016} and \cite{Kohler2020} where the Hamiltonian from the circuit-to-Hamiltonian mapping has local terms which are described in $\BigO(\poly(N))$ bits (i.e.\ the input size for each instance is $n=\poly N$).
Here $N$ is the number of spins in the system; and we emphasise that the Hamiltonian can be translationally invariant nonetheless by using the quantum phase estimation extractor from \cite{Kohler2020} in the following.
Crucially, we ask for the local observable within an energy subspace within $\delta =\BigO(1/\exp(n))$ of the ground state.\footnote{Previously, in the translationally-invariant case, we had $n=\polylog N$; here $n=\poly N$; as such the gap in this case is exponentially small in the system size, whereas previously it was polynomially small.}

\begin{corollary}\label{Corollary:PreciseQMA_Hardness}
	$\forall$-TI-APX-SIM (and hence APX-SIM) is $\PSPACE$-hard for a nearest-neighbour translationally invariant Hamiltonian on qudits of local dimension $44$, for a 1-local observable $A$, $\delta=1/\exp(n)$, $b - a =\Omega(1/\exp(n))$.
\end{corollary}
\begin{proof}[Proof idea.]
    Use \cite{Kohler2020} (translationally invariant) as the circuit-to-Hamiltonian construction.
    As $\PreciseQMA=\PSPACE$ by \cite{Fefferman2016}, we know that $\PP^{\parallel\PreciseQMA} = \PP^{\parallel\PSPACE} = \PSPACE = \PreciseQMA$. We encode the circuit from \cref{lem:augmented-circuit}, but modify it so that only a single query to a \PreciseQMA oracle is made and have it such that the flag qubit and the output qubit are the same.
    We note that while we cannot generally amplify a \PreciseQMA completeness-soundness gap polynomially close to $1$ and $0$ in polynomial time (otherwise $\PreciseQMA=\QMA$), respectively, we \emph{can} attain such amplification in \emph{exponential} time in ``circuit-world''. Namely, we can apply the Marriott-Watrous strong error reduction technique~\cite{MW05}, which will blow up the runtime of the circuit to exponential, but crucially keep the space usage polynomial. We can thus augment \cref{lem:augmented-circuit} such that it performs such an amplification manually, so we can assume that the single oracle subquery has an output which, for valid queries, gives a probability polynomially close to $1$ and $0$.

    We add a projector $P$ which penalises the output/flag qubit.
    Finally, let the observable we wish to measure in APX-SIM also be the observable $P$.
    Then equivalents of \cref{Lemma:Trace_Distance}, \cref{Lemma:Trace_Distance}, and \cref{lem:low-energy-as-many-yes-as-possible} can be proven where the $1/\poly$ factors are replaced by $1/\exp$, and in \cref{lem:low-energy-as-many-yes-as-possible} is applies to just a single query.
    It is then possible to prove \cref{Lemma:A_Variable_Hardness} where the measurement $M_1=P$.
    Since $\PreciseQMA$ has an exponentially small promise gap, and is a polynomial time computation, then measuring $P$ with promise gap $b-a=\Omega(1/\exp(n))$ is \PreciseQMA-hard.
\end{proof}

\subsection{Containment Results}


In this section, we speak of the $1/\poly$-gap and $1/\exp$-gap regimes to mean the settings in which all promise gaps in LH and APX-SIM are at least inverse polynomial and at least inverse exponential, respectively, in the input size, $n$. For example, the ``standard'' Local Hamiltonian problem of~\cite{KSV02} is in the $1\poly$-gap regime. In contrast, the translationally invariant 1D setting (e.g.~\cite{Gottesman2009}) is in the succinct $1/\exp$-regime. Note that to make the notions of $1/\poly$- and $1/\exp$-gap meaningful, without loss of generality we may assume all Hamiltonians below are renormalized to be spectral-norm $1$. (Thus, a Hamiltonian with exponential weights on its local terms is now a norm-$1$ Hamiltonian with a $1/\exp$ promise gap.) Finally, we define $\textup{LH}(g,N)$ as the Local Hamiltonian problem~\cite{KSV02} with promise gap $g$.


\begin{lemma}\label{Lemma:Containment_Proof}
        Let $\mathcal{F}$ be a family Hamiltonians and Q a QVClass (\cref{def:qV}) such that $\textup{LH}(g)$ for $\mathcal{F}$ is in Q. In the $1/\poly$-gap regime, $\textup{APX-SIM}\in \PP^{\class{Q}[\log]}$, and in the $1/\exp$-gap regime, $\textup{APX-SIM}\in \PP^{\class{Q}}$.
\end{lemma}
\begin{proof}
    This is a straightforward application of the algorithm given in Section A.2 of \cite{Ambainis2013}. The algorithm first uses the Q-oracle to conduct a binary search in order to obtain an additive error $\eta$ estimate of the ground state energy of $H$. (This is where the containment of $\textup{LH}(g,N)$ in Q is used, for example.) With $\eta$ in hand, we run one final Q-query to check if there is a low-energy state (relative to $\eta$) of $H$ whose expected value for the observable $A$ satisfies the YES-criteria for APX-SIM.

    The only question is the level precision required, i.e.\ how should $\eta$ scale? In the $1/\poly$-gap regime, it suffices to use logarithmically many adaptive queries to the Q-oracle to obtain $\eta\in \BigO(1/\poly (n))$, as required to distinguish inverse polynomial promise gaps. In the $1/\exp$-gap regime, we need polynomially many adaptive queries to obtain $\eta\in \BigO(1/\exp(n))$. This yields the claim.
\end{proof}

A similar, but not quite identical result is
\begin{corollary}\label{Lemma:Containment_Proof_PSPACE}
    APX-SIM where $Q=\PreciseQMA$ in \cref{Corollary:PreciseQMA_Hardness} is contained in \PSPACE.
\end{corollary}
\begin{proof}
	The proof follows as for \cref{Lemma:Containment_Proof} in the $1/\exp$-gap regime; we need $1/\poly N$ many queries to obtain $\eta\in \BigO(1/\exp(N))$.
	As $\PP^\PreciseQMA=\PSPACE$ the claim follows.
\end{proof}

Finally, \cref{Corollary:1D_TI_NN_Hardness} showed $\forall$-TI-APX-SIM (and hence APX-SIM) is $\EXP^{\parallel \QMA}$-hard (in the $1/\exp$-gap regime). This is a somewhat odd-looking complexity class; the following theorem shows it is actually equal to something ``closer to $\PQMA$ in appearance''.

\begin{theorem}\label{thm:equality}
	$\EXP^{\parallel\QMA}=\PQMAexp$.
\end{theorem}
\begin{proof}
    For $\EXP^{\parallel\QMA}\subseteq\PQMAexp$,  \cref{Corollary:1D_TI_NN_Hardness} showed $\forall$-TI-APX-SIM (and hence APX-SIM) is $\EXP^{\parallel \QMA}$-hard in the $1/\exp$-gap regime. But since the family of Hamiltonians in TI-APX-SIM can be verified in the $1/\exp$-gap regime in \QMAEXP, \cref{Lemma:Containment_Proof} says $\textup{APX-SIM}\in \PQMAexp$.
	
	The reverse containment, $\PQMAexp\subseteq \EXP^{\parallel\QMA}$ is similar to Beigel's original proof~\cite{Beigel91} that $\PNPlog\subseteq\PNPpar$: the \EXP machine makes all possible queries within the decision tree of the \PP machine up front and in parallel, then simulates the P machine, along the way selectively using the results of whichever queries the adaptive nature of the P machine would have needed. The only catch is that normally, Hamiltonians verified in \QMAEXP act on $\exp(n)$ qubits, whereas those in \QMA act on $\poly(n)$ qubits. However, since an $\EXP$ Turing Machine can write exponentially long queries to the oracle, it can simply take the TI Hamiltonian $H$ the P machine would have fed to the \QMAEXP oracle (here we are using that the TI-local Hamiltonian problem is \QMAEXP-complete~\cite{Gottesman2009}), and explicitly write it out fully (i.e.\ write out all $\exp(n)$ local terms $H_i$), and feed it to the QMA oracle (which, being fed an $\exp(n)$-size input, is allowed an $\exp(n)$-size proof by definition).
\end{proof}

\begin{corollary}\label{Corollary:Completeness}
	$\forall$-TI-APX-SIM and TI-APX-SIM are both $\PQMAexp$-complete for families of Hamiltonians for which the Local Hamiltonian problem is \QMAEXP-complete.
\end{corollary}
\begin{proof}
	Follows from \cref{Corollary:1D_TI_NN_Hardness}, \cref{Lemma:Containment_Proof}, and \cref{thm:equality}.
\end{proof}

\begin{corollary}\label{Corollary:1D_TI_NN_Completeness}
    TI-APX-SIM is $\PP^{\QMAEXP}$-complete for a 1D, nearest neighbour, translationally invariant Hamiltonian for a 2-local observable $A$, and $\delta=1/\exp(n)$, $b - a =\Omega(1/\exp(n))$.
\end{corollary}
\begin{proof}
    From \cref{Corollary:Completeness}.
\end{proof}

\section{QCMA\textsubscript{EXP}-Completeness of TI-GSCON}\label{sec:TI_GSCON}

The main theorem of this section is \cref{thm:TIGSCON} (\cref{sscn:introGSCON} of Introduction), which recall says \TIGSCON is \QCMAEXP-complete for 1D, nearest neighbour, translationally invariant Hamiltonians on N qudits. We begin with relevant definitions and tools in \cref{sscn:defGSCON}. \Cref{sscn:qcmaexpcomplete} gives a framework for lifting general 1D translationally invariant circuit-to-Hamiltonian constructions to \QCMAEXP-hardness results. \cref{sscn:GI} then instantiates this framework using the Gottesman-Irani 1D construction~\cite{Gottesman2009}, yielding Theorem~\ref{thm:TIGSCON}.

\subsection{Definitions and Tools}\label{sscn:defGSCON}

\begin{definition}[$\p$-orthogonal states and subspaces\cite{GS14}]\label{def:korth}
For $\p \geq 1$, a pair of states ${\ket{v},\ket{w}\in(\C^d)^{\otimes N}}$ is \emph{$\p$-orthogonal} if for all $\p$-qudit unitaries $U$, we have $\bra{w}U\ket{v}=0$. We call subspaces $S,T \subseteq(\C^d)^{\otimes N}$ \emph{$\p$-orthogonal} if any pair of vectors $\ket{v}\in S$ and $\ket{w}\in T$ are $\p$-orthogonal.
\end{definition}

\noindent For example, $\ket{000},\ket{100}\in(\C^2)^{\otimes 3}$ are orthogonal but not $\p$-orthogonal for any $\p\geq 1$. Analogously, for $S$ and $T$ the $+1$ eigenspaces of $I \otimes \ketbra{000}{000}$ and $I \otimes \ketbra{111}{111}$, respectively, $S$ and $T$ are $2$-orthogonal subspaces.

We now state the Traversal Lemma~\cite{GS14,GMV17}. The version below is the special case of Lemma 9 of~\cite{GMV17} in which one sets $V=0$.

\begin{lemma}[Traversal Lemma~\cite{GS14,GMV17}]\label{l:traversal}
    Let $S,T \subseteq(\C^d)^{\otimes N}$ be $\p$-orthogonal subspaces. Fix arbitrary states $\ket{v}\in S$ and $\ket{w}\in T$, and consider a sequence of $\p$-qudit unitaries $(U_i)_{i=1}^m$ such that
    \[
        \norm{\ket{w}- U_m\cdots U_1 \ket{v}}_2 \leq \epsilon
    \]
for some $0\leq \epsilon < 1/2$. Define $\ket{v_i}\coloneqq U_i \cdots U_1 \ket{v}$ and $P\coloneqq I-\Pi_S-\Pi_T$. Then, there exists an $i\in[m]$ such that
    \[
        \bra{v_i} P \ket{v_i}\geq \left(\frac{1- \epsilon}{m}\right)^2.
    \]
\end{lemma}
\noindent Intuitively, the Traversal Lemma says that for any $\p$-orthogonal subspaces $S$ and $T$ with $\ket{v} \in S$ and $\ket{w} \in T$, and any sequence of $m$ $\p$-qudit unitaries mapping $\ket{v}$ to $\ket{w}$, there exists a time step $i\in[m]$ when our intermediate state has non-trivial overlap with the orthogonal complement of $S$ and $T$.

\subsection{Generic hardness constructions via a Lifting Lemma}\label{sscn:qcmaexpcomplete}

We now give a black-box mapping for ``lifting'' 1D translationally invariant circuit-to-Hamiltonian constructions to \QCMAEXP-hardness results for GSCON. While the end-goal is similar to the Lifting Lemma for APX-SIM (\cref{lem:lift}), here there are two further restrictions: first, we do not consider general classes (such as $\DQ$), but \QCMAEXP---this is roughly because here we need the ability to prepare history states efficiently, for which a classical proof appears necessary. Second, we restrict attention to general 1D translationally invariant circuit-to-Hamiltonian mappings (``TI-standard'' mappings; see \cref{def:TIstd} below), as opposed to the broader definition of ``local circuit-to-Hamiltonian mapping'' of \cref{def:local-mapping}. This is because for our GSCON hardness construction, we couple polynomially many ``switch qudits'' to whichever Hamiltonian construction $H$ we start with; this can immediately distort the structural properties of $H$ itself. As a result, proving generic statements capturing any possible setup beyond 1D can become cumbersome. However, to be clear, our framework \emph{can} in principle be generalized beyond 1D, since it allows ``switching on'' Hamiltonian terms in an arbitrary order (\cref{sscn:o3}).

To begin, we define the general class of 1D TI circuit-to-Hamiltonian mappings which we consider here. Below, a ``verification circuit'' should be viewed as a QCMA or \QCMAEXP verifier. (We omit QMA or \QMAEXP here, since the ``efficiently preparable proof'' requirement for the YES case below implies we can assume a classical proof.)

\begin{definition}[TI-standard]\label{def:TIstd}
    A circuit-to-Hamiltonian mapping from verification circuits $V$ to 1D, nearest neighbour, translationally invariant Hamiltonians $H=\sum_{i=1}^{N-1}H_{i,i+1}$ is \emph{TI-standard} if it satisfies the following conditions. Below, $N$ denotes the number of qudits $H$ acts on, and $\alpha$ and $\beta$ the completeness/soundness (a.k.a. ``low energy'' and ``high energy'') parameters for $H$, respectively:
\begin{enumerate}
    \item $H\succeq 0$, although the local terms may satisfy $H_{i,i+1}\not\succeq 0$.

    \item In the YES case, if the optimal proof to verifier $V$ is a classical string $y\in\set{0,1}^{\poly(N)}$, then there exists a (potentially non-uniformly generated) quantum circuit of size $L\in\poly(N)$ preparing a low energy state $\ket{\psilow}$ for $H$, i.e.\ $\bra{\psilow}H\ket{\psilow}\leq \alpha$.
    \item In the YES case, the subset of indices $F\subseteq[N-1]$ for which $H_{i,i+1}$ contributes negative energy to the low-energy state, i.e.\ all $i$ for which $H_{i,i+1}$ satisfies
        $
            \bra{\psilow}H_{i,i+1}\ket{\psilow}<0,
        $
        is computable in $\poly(N)$ time\footnote{One can replace the $<0$ condition here with $<-1/p(N)$ for some sufficiently small polynomial $p$; we omit this for simplicity.}.

    \item In the NO case, $\lmin(H)\geq \beta$. Here, we require $\abs{\alpha-\beta}\geq 1/\poly(N)$ (which is standard in the literature) and $\beta\geq 16(2L+7N)\alpha\geq 0$ (which is specific to our construction).
\end{enumerate}
\end{definition}
\noindent All of these assumptions are rather mild, as we now clarify.\\
\vspace{-1mm}

\noindent \emph{Remarks regarding \cref{def:TIstd}:}
 \begin{itemize}
    \item Assumptions 1 and 4  must be taken together (otherwise, $H\succeq 0$ can always be achieved by adding a multiple of the identity).

    \item The setting of $H\succeq 0$ but $H_i\not\succeq 0$ arises when applying our construction to the Gottesman-Irani 1D TI mapping (henceforth GI)~\cite{Gottesman2009} in \cref{sscn:GI}. Specifically, GI is not TI-standard in its original form, since it violates the final requirement $\beta\geq 16(2L+7N)\alpha$, which is crucial to our use of the Traversal Lemma. As we show in \cref{sscn:GI}, GI can be made TI-standard with minor modifications. In doing so, however, although the local terms of GI satisfy $H_i\succeq 0$, after our modification to TI-standard form, we no longer have $H_i\succeq 0$. Also, we remark the choice of $16(2L+7N)$ arises from our framework below; in principle, most Hamiltonian constructions such as GI can easily be made to satisfy $\beta\geq p(N)\alpha$ for any large enough polynomial $p$.

    \item Assumption 2 holds for most circuit-to-Hamiltonian constructions in the literature (such as~\cite{KSV02} and~\cite{Gottesman2009}), since the optimal (non-uniform) classical proof $y$ can be plugged into the corresponding history state construction.

    \item Assumption 3 is vacuously true when all $H_i\succeq 0$. When $H_i\not\succeq 0$ for some $i$, however, this is also generally a mild assumption, since it only cares about energies against $\ket{\psilow}$, which is typically a history state of some form (and thus the structure of $\ket{\psilow}$ is well-understood).
 \end{itemize}

\paragraph{Main lemma.} We now state the main lemma of this section, It is proved in Sections \ref{sscn:o1}, \ref{sscn:o2}, \ref{sscn:o3}, and \ref{sscn:finalGSCON}.
\begin{lemma}[Lifting Lemma for GSCON]\label{lem:liftGSCON}
    Let $V$ be a verifier for a \QCMAEXP promise problem. Fix any TI-standard circuit-to-Hamiltonian mapping which produces 1D TI Hamiltonians on qudits of local dimension $d$, and any $b\in\set{2,\ldots, N-1}$. Then, there exists a $\poly(\log N)$-time many-one reduction mapping any instance $x$ for $V$ to a TI-GSCON instance $H$ on $N$ qudits of local dimension $7d$ with $\p$-local unitaries $U_i$, such that $m\in\poly(N)$ and $\delta\in\Theta(1/\poly(N))$.
\end{lemma}
\noindent Above, $\p$, $m$ and $\delta$ are parameters from the definition of GSCON (\cref{def:GSCON}). Note the runtime of the reduction is $\poly(\log N)$, as opposed to $\poly(N)$, since as with all 1D TI constructions, it only outputs a single local interaction term $H_{i,i+1}$ (acting in this case on $\C^{7d}\otimes\C^{7d}$) and the length of the chain $N$ in binary. Throughout, we assume $d\in O(1)$, as is standard.

\paragraph{Proving \cref{lem:liftGSCON} in steps.} Given an instance $x$ of a \QCMAEXP promise problem with verification circuit $V$, our goal is to map it to an instance $(H,\ket{\psi},\ket{\phi},m)$ of TI-GSCON. In order to build intuition for the construction, we first briefly review the \QCMA-completeness proof of \GSCON of~\cite{GS14}. Then, we show how to overcome various obstacles in adapting this to the translationally invariant setting.\\

\noindent\emph{Review of~\cite{GS14}.} Let $H'$ be a local Hamiltonian encoding a \QCMA computation (via, e.g., the Kitaev circuit-to-Hamiltonian construction~\cite{KSV02}). For this review, assume $\p=2$. The high-level premise of \cite{GS14} is to construct a ``switch gadget'', so that ``traversing the low energy space'' of some Hamiltonian $H$ corresponds to ``switching'' $H'$ on and off. Specifically, consider
\begin{equation}\label{eqn:gs}
    H=H'_{\sA}\otimes (I-\ketbra{000}{000}-\ketbra{111}{111}_{\sB}),
\end{equation}
where we have appended $3$ ``switch qubits'' in a new ``switch register'' $\sB$. Set the start and final states to $\ket{\psi}=\ket{0\cdots 0}_{\sA}\ket{000}_{\sB}$ and $\ket{\phi}=\ket{0\cdots 0}_{\sA}\ket{111}_{\sB}$, respectively. Intuitively, in the YES case, to map $\ket{\psi}$ to $\ket{\phi}$, an honest prover first prepares the history state of $H'$ in $\sA$ (which can be done with a (non-uniform) polynomial size circuit, since $\QCMA$ has a \emph{classical} proof), and then flips each of the bits in $\sB$ sequentially from $0$ to $1$. Whereas $\ket{\psi}$ is in the null space of $H$, as soon as the first two switch qubits are turned on, the current state is in the support of $H$, and so Hamiltonian $H'$ is ``activated''. Since we are in a YES case, however, the history state has low energy against $H'$, and so we satisfy the YES conditions of \GSCON. In the NO case, one uses the Traversal Lemma (Lemma~\ref{l:traversal}) to show that activating $H'$ in this manner is unavoidable; any sequence of $2$-local unitaries mapping $\ket{\psi}$ to $\ket{\phi}$ must have non-trivial overlap at some time step on $(I-\ketbra{000}{000}-\ketbra{111}{111}_{\sB})$, thus activating $H'$, which administers a large energy penalty. Thus, we satisfy the NO conditions of \GSCON.

\subsubsection{Obstacle 1: Geometric locality} \label{sscn:o1}
Although the construction of~\cite{GS14} is  oblivious to the structure of $H'$, the switch gadget of Equation~(\ref{eqn:gs}) itself is highly geometrically non-local---each qubit of switch register $\sB$ couples to \emph{all} qubits of $H'$. This allows one to simultaneously switch on or off {all} local terms of $H'$. However, in 1D  we do not have this non-local luxury. Instead, here our high-level approach is to give each computation qudit in $\sA$ its \emph{own}\footnote{We remark that Breuckmann and Terhal's~\cite{BT14} \emph{space-time circuit-to-Hamiltonian construction} uses a similar idea of endowing each computation qubit with its own \emph{clock} register. In both~\cite{BT14} and here, the aim is to adhere to desired geometric constraints. For clarity, however, the contexts of ``local clock''~\cite{BT14} versus `local switch'' (the current work) qubits, along with their analyses, are fundamentally different.} switch qudit, and sequentially flip each of these clock qudits on one at a time, thus activating each local term of $H'$ sequentially. Let us now make this concrete.

We begin with any TI-standard circuit-to-Hamiltonian construction, $H'=\sum_{i=1}^{N-1}H'_{i,i+1}$ acting on $(\C^d)^{\otimes N}$. For pedagogical purposes, we also assume for now that $H'_{i,i+1}\succeq 0$ for all $i$ (later we will drop this assumption). For clarity, we henceforth refer to the $i$th qudit space as $\sA_i$, so that $H'\in\herm{\bigotimes_{i=1}^N\sA_i}$. We augment each qudit $i\in\set{1,\ldots, N}$ with a $3$-dimensional ``switch'' qutrit space $\sB_i$, so that the new Hamiltonian $H$ we obtain acts on
$
   (\C^d\otimes\C^3)^{\otimes N}\simeq \bigotimes_{i=1}^{N}\sA_i\otimes\sB_i$,
i.e.\ a chain of $N$ $3D$-dimensional qudits. Switch qudit $i\in\set{0,\ldots,N-1}$ will be used to switch on local interaction term $H'_{i,i+1}$. (Note there are only $N-1$ local terms, but $N$ switch qudits. The $N$th switch qudit, which is not coupled to any term $H'_{i,i+1}$, simply makes all local Hilbert spaces along the chain identical.) For each qudit, we update each $H'_{i,i+1}$ in two steps:
\begin{enumerate}
    \item (Add local switches) For all $i\in\set{1,\ldots, N-1}$, set $H'_{i,{i+1}}\coloneqq H'_{\sA_i,\sA_{i+1}}\otimes\ketbra{1}{1}_{\sB_i}$.
    \item (Add string constraints) Define $E\coloneqq \ketbra{01}{01}+\ketbra{02}{02}+\ketbra{20}{20}+\ketbra{21}{21}$. Then, for all $i\in\set{1,\ldots, N-1}$, for $\Delta$ to be chosen as needed, set
    $
            H_{i,{i+1}}\coloneqq H'_{i,{i+1}}+\Delta E_{\sB_{i},\sB_{i+1}}$.

\end{enumerate}
Our Hamiltonian is thus $H=\sum_{i=1}^{N-1}H_{i,i+1}$. The start and final states are $\ket{\psi}=\bigotimes_{i=1}^{N}\ket{0}_{\sA_i} \ket{0}_{\sB_i}$ and $\ket{\phi}=\bigotimes_{i=1}^{N}\ket{0}_{\sA_i} \ket{2}_{\sB_i}$, respectively.

\begin{figure}[t]
\[
\begin{array}{l|l}
  \text{Setting}& \text{Switch register contents} \\
    \hline
  \text{Honest} &000 \mapsto  100  \mapsto 110 \mapsto 111\mapsto 112\mapsto 122 \mapsto 222\\
  \text{Cheating} & 000 \mapsto  100  \mapsto 122 \mapsto 222 \\
\end{array}
\]
\caption{(First row) An honest prover essentially counts in unary from left to right, then right to left. (Second row) A cheating prover ``shortcuts'' this process by using a $b$-local unitary (in the example above, $b=2$) to instantly flip the last $b$ qubits from $0^b$ to $2^b$.}
\label{fig:switchpre}
\end{figure}

The intuition of this construction is as follows. Since we made the additional assumption that $H'_{i,i+1}\succeq 0$ for all $i$, $\ket{\psi}$ and $\ket{\phi}$ are in the null space (and hence ground space) of $H$, as desired. To then map $\ket{\psi}$ to $\ket{\phi}$ in the YES case, an honest prover first prepares the low energy state $\ket{\psilow}$ of $H'$ in $\sA$. 
By the rules of the string constraints $E$, the only way for an honest prover to map all zeroes in $\sB$ ($\ket{\psi}$) to all twos in $\sB$ ($\ket{\phi}$) is to now build a ``string'' of ones from left to right (i.e.\ flip $\sB_1$ from $\ket{0}$ to $\ket{1}$, followed by $\sB_2$, and so forth, until $\sB$ is all ones), followed by a string of twos from right to left (i.e.\ flip $\sB_N$ from $\ket{1}$ to $\ket{2}$, followed by $\sB_{N-1}$, and so forth, until $\sB$ is all twos)---see \cref{fig:switchpre} for an illustration. Thus, there exists a time step at which the first $N-1$ switch qubits are set to $\ket{1}$. We can hence apply the following observation, which we informally state for reference:
 \begin{observation}\label{obs:over}
    When the first $N-1$ switch qubits are set to $\ket{1}$, the entire Hamiltonian $H'$ has been activated on $\sA$, which in turn ``checks'' the energy of the state in $\sA$.
 \end{observation}
\noindent In a YES case, Observation~\ref{obs:over} poses no problem, since the low energy state $\ket{\psilow}$ prepared in $\sA$ satisfies $\bra{\psilow}H'\ket{\psilow}\leq\alpha$. As for the intuition for the NO case, this will uncover the first of the various holes in this first  construction attempt; as such, let us discuss these problems as a group next.\\

\noindent \emph{Holes in this construction.} The following are problems with the current construction:
\begin{enumerate}
\item \emph{Obstacle 2: Partial activation of the chain.} Suppose we are in the NO case. Since we cannot assume a cheating prover applies any particular sequence of local unitaries $U_i$, we apply the Traversal Lemma (Lemma~\ref{l:traversal}), hoping to attain the following claim: since $\ket{\psi}$ and $\ket{\phi}$ are $\p$-orthogonal, there must exist a time step at which the intermediate state $\ket{\psi_i}=U_i\cdots U_1\ket{\psi}$ has non-negligible overlap with $\ket{1\cdots 1}_{\sB_{1,\ldots, N-1}}$, and thus we can apply Observation~\ref{obs:over}. Except, this claim is demonstrably false---a cheating prover can sequentially (from left to right) flip the first $N-\p$ zeroes in $\sB$ to ones to get $\ket{1}^{N-\p}\ket{0}^{\p}_\sB$, then use a {single} $\p$-local unitary on $\sB_{N-\p+1,\ldots,N}$ to ``shortcut'' to state $\ket{1}^{N-\p}\ket{2}^{\p}_{\sB}$, and subsequently sequentially flip all ones to twos (from right to left)---see \cref{fig:switchpre}. In other words, the chain corresponding to $H'$ is only ``partially activated'', rendering it unclear whether soundness holds.

\item \emph{Obstacle 3: Temporary surges in energy penalty.} For the YES case, if we drop the additional assumption that $H'_{i,i+1}\succeq 0$ for all $i$, then completeness can also fail. To see this, recall that in the YES case, we set register $\sA$ to $\ket{\psilow}$, and subsequently ``switched on'' each of the local terms $H'_{i,i+1}$ one-by-one from left to right. To satisfy the YES conditions for GSCON, we must hence have that \emph{for all} $i\in\set{1,\ldots, N-1}$,
     \[
           \bra{\psilow}H'_{1,2}+\cdots+ H'_{i,i+1}\ket{\psilow}\leq \alpha,
     \]
     i.e.\ we stay in the low-energy subspace of $H'$ throughout this entire iterative ``switching'' process. However, if (say) $H'_{N-1,N}\not\succeq 0$, then it may happen that
     \[
           \bra{\psilow}H'_{1,2}+\cdots+ H'_{N-2,N-1}\ket{\psilow}\geq \beta\quad\text{but}\quad\bra{\psilow}H'_{1,2}+\cdots+ H'_{N-1,N}\ket{\psilow}\leq \alpha.
     \]
    This is because switching on $H'_{N-2,N-1}$ may yield a temporary energy surge above $\beta$, which is then counter-acted by a large negative energy contribution from $H'_{N-1,N}$, bringing us back down below $\alpha$. Indeed, this is precisely the obstacle we will see when instantiating our final framework in \cref{sscn:GI} with the GI construction~\cite{Gottesman2009}.
\end{enumerate}

\noindent We now outline how to deal with these obstacles in \cref{sscn:o2} and \cref{sscn:o3}.

\subsubsection{Obstacle 2: Partial activation of the chain} \label{sscn:o2}

Recall that Obstacle 2 says a cheating prover can avoid switching on all local terms of $H'$, i.e.\ may ``prematurely cut'' the 1D chain at any of the last $b$ positions. Here, we give two TI fixes for this issue, of which we adopt the latter.\\\vspace{-2mm}

\noindent\emph{Fix 1: White-box.} The first fix is to make white-box use of the particular details of the chosen 1D TI construction. GI, for example, has the remarkable property of being (what we call) ``self-organizing''. Exploiting this, we can make GI ``robust'' against partial activation, meaning roughly that with minor modifications, we can make GI encode precisely the computation we want, even if we have no control of where exactly in the last $\p$ positions the chain is cut. This is outlined in \cref{app:altfix2} for the interested reader.\\

\vspace{-2mm}
\noindent\emph{Fix 2: Black-box.} Here, however, we are aiming for generic lifting theorems, and as such are not able to make white-box use of any particular 1D TI construction. Instead, we take a black-box approach, which appears to require forcing a cheating prover to switch on \emph{all} local terms. \emph{A priori} this may seem impossible, since on the one hand, the prover can always use a single $\p$-local unitary to ``shortcut'' any switching logic we are looking to enforce, whereas on the other hand, any constraints we add must be applied in a TI fashion to each local term. However, it turns out that the solution to Obstacle 3 (\cref{sscn:o3}) will also allow us to resolve Obstacle 2---namely, by moving to a higher dimensional switch register, we can ``trap'' any cheating prover in a low-dimensional switch subspace, which forces the prover to indeed switch on all local terms.

\subsubsection{Obstacle 3: Temporary surges in energy penalty}\label{sscn:o3}

Since we cannot assume $H'_i\succeq 0$ for all $i\in\set{1,\ldots, N-1}$, Obstacle 3 says it is possible (e.g.) that
\[
     \bra{\psilow}H'_{1,2}+\cdots+ H'_{N-2,N-1}\ket{\psilow}\geq \beta\quad\text{but}\quad\bra{\psilow}H'_{1,2}+\cdots+ H'_{N-1,N}\ket{\psilow}\leq \alpha,
\]
Thus, in the YES case, we do not remain in the low-energy space of $H$, contrary to what is required by \GSCON.

The issue here is the \emph{order} in which the local terms $H'_{i,i+1}$ are switched on (i.e.\ we are proceeding from left to right on the chain). Indeed, letting $e_i\coloneqq \bra{\psilow}H'_{i,i+1}\ket{\psilow}$, if we were to first switch on only terms $H'_{i,i+1}$ for which $e_i<0$, followed by only terms with $e_i\geq 0$, then the intermediate energies obtained will always be at most $\alpha$ in the YES case (since by assumption, $\sum_i e_i \leq \alpha$). Unfortunately, such an alternate switching order may be arbitrary (i.e.\ dependent on the particular TI-standard construction used), and thus not allowed under our current string constraints $E$. To handle this, we hence move to a higher dimensional switch space, affording us added degrees of freedom to switch on Hamiltonian terms in an \emph{arbitrary} order. In doing so, as mentioned in \cref{sscn:o2}, we will also be able to resolve the partial activation obstacle.\\

\noindent\emph{How to switch on Hamiltonian terms arbitrarily.} Replace each $3$-dimensional switch qudit $\sB_i$ with a $7$-dimensional qudit with $2$-local nearest-neighbour relations given by \cref{fig:neighbour}, and update the starting and final states to $\ket{\psi}=\bigotimes_{i=1}^{N}\ket{0}_{\sA_i} \ket{0}_{\sB_i}$ and $\ket{\phi}=\bigotimes_{i=1}^{N}\ket{0}_{\sA_i} \ket{6}_{\sB_i}$ respectively.

\begin{figure}[t]
\[
  \begin{array}{c|ccccccc}
     & 0 & 1 & 2 & 3 & 4 & 5 & 6\\
     \hline
    0 & \cmark & \cmark & \xmark & \xmark & \xmark & \xmark & \xmark\\
    1 & \cmark & \cmark & \cmark & \xmark & \xmark & \xmark & \xmark\\
    2 & \xmark & \xmark & \cmark & \xmark & \xmark & \xmark & \xmark\\
    3 & \xmark & \xmark & \cmark & \cmark & \cmark & \xmark & \xmark\\
    4 & \xmark & \xmark & \xmark & \xmark & \cmark & \xmark & \xmark\\
    5 & \xmark & \xmark & \xmark & \xmark & \cmark & \cmark & \cmark\\
    6 & \xmark & \xmark & \xmark & \xmark & \xmark & \cmark & \cmark\\
  \end{array}
\]
\caption{The expanded neighbour relations $E$ allowed to overcome Obstacle 3. A \cmark\ (\xmark) at position $(i,j)$ means symbol $i$ can (cannot) appear immediately to the left of $j$ in $\sB$. For example, $(1,3)$ contains $\xmark$, so $E$ forbids substring $13$. Note that $3$ can only have $3$ to its left (column $3$), and $2$ (resp. $4$) can only have $2$ (resp. $4$) to its right (rows $2$ and $4$, respectively).}
\label{fig:neighbour}
\end{figure}
The intuition is as follows. Instead of $\ket{1\cdots 1}_{\sB}$, the ``bottleneck'' configuration which the Traversal Lemma now leverages to ensure all Hamiltonian terms $H'$ are switched on is $\ket{3\cdots 3}_{\sB}$. However, to overcome Obstacle 3, we also allow $\ket{k}_{\sB_i}$ for $k\in\set{1,2,3,4,5}$ on switch qudit $\sB_i$ to activate term $H'_{i,i+1}$; thus, there is a ``warm-up'' phase in which we first switch on ``preferred'' terms, before eventually moving to the ``full blast'' configuration of $\ket{1\cdots 1}_{\sB}$ (\cref{fig:switch}). As for Obstacle 2, in the dishonest case, the use of the additional dimensions $\set{2,4}$ allows us to ``trap'' a cheating prover in a low-dimensional switch subspace (think of dimensions $2$ and $4$ as the lower and upper pieces of bread for a sandwich, which trap a dishonest prover in the middle ``filling'' layer corresponding to dimension $3$) which forces them to activate \emph{all} of the $H'$ Hamiltonian terms (see \cref{eqn:4} of \cref{sscn:finalGSCON}). Further intuition is provided by the honest prover's actions in the completeness analysis of \cref{sscn:finalGSCON}.

\subsubsection{Putting the pieces together} \label{sscn:finalGSCON}

We are now ready to formally show completeness and soundness. We collect all the relevant definitions in this subsection, and set all remaining parameters.

\begin{definition}[Lifted Hamiltonian] \label{Def:Lifted_Hamiltonian}
	Let $x$ be an instance of a \QCMAEXP promise problem, with verification circuit $V$ such that for any YES instance $x$, $V$ accepts some classical proof with probability at least $1-\epsilon$, and for all NO instances $x$, $V$ accepts all proofs with probability at most $\epsilon$.
	By standard error reduction via parallel repetition, we may assume $\epsilon$ is any desired fixed inverse exponential in the length of $x$.
	Let $H'=\sum_{i=1}^{N-1} H'_{i,i+1}$, be the 1D, translationally invariant, nearest neighbour Hamiltonian generated by applying a TI-standard construction (as defined in \cref{def:TIstd}) to $V$. 
	Further let the completeness/soundness parameters of $H'$ be $\alpha$ and $\beta$, respectively.
	Define $E$ as the $2$-local projector onto the set of forbidden $2$-local nearest-neighbour substrings in \cref{fig:neighbour}.
	Then define the lifted Hamiltonian as $H=\sum_{i=1}^{N-1}H_{i,i+1}$, where:
	\begin{eqnarray} \label{Eq:Local_Terms}
	H_{i,i+1}&\coloneqq & H'_{i,i+1}\otimes(\ketbra{1}{1}+\ketbra{2}{2}+\ketbra{3}{3}+\ketbra{4}{4}+\ketbra{5}{5})_{\sB_i}+\Delta E_{\sB_i,\sB_{i+1}}\label{eqn:finalterm}
	\end{eqnarray}
	for $\Delta$ to be chosen as needed.
\end{definition}
\noindent See \cref{sscn:GI} for our instantiation via GI~\cite{Gottesman2009}

For our proof of \QCMAEXP-hardness of TI-GSCON, the start and final states are
\begin{align}\label{Eq:Start_State}
\ket{\psi}&=\bigotimes_{i=1}^{N}\ket{0}_{\sA_i} \ket{0}_{\sB_i}
\end{align}
and
\begin{align}\label{Eq:Final_State}
\ket{\phi}&=\bigotimes_{i=1}^{N}\ket{0}_{\sA_i} \ket{6}_{\sB_i}
\end{align}
respectively.
Set $\aaa=\alpha$, $\bbb=\beta/(8m^2)$, $\ccc=0$, $\ddd=1/2$, and $m=2L+7N$, where $L\in\poly(N)$ is the size of the circuit preparing the ground state of $H'$ in the YES case. Set $\delta=(\aaa+\bbb)/2$, which by definition of TI-standard is at least inverse polynomial in $N$ (since $\beta$ is at least inverse polynomial in $N$). Finally, we may set $\p$ to any value from $\set{2,\ldots, N-1}$ ($\p$ the locality of each $U_i$); $\p=2$ suffices to show completeness, and soundness holds for all $\p\in\set{2,\ldots, N-1}$.

\paragraph{Completeness.}
Suppose $x$ is a YES instance. The following lemma shows there is a short path through the low energy subspace between states $\ket{\psi}$ and $\ket{\phi}$, as desired.
\begin{lemma}\label{Lemma:Completeness_Analysis}
Using the notation of \cref{Def:Lifted_Hamiltonian}, let $\lmin(H')\leq \alpha$. There exists a circuit $U=U_m\dots U_2U_1$ of $2$-local gates $U_i$ such that $U\ket{\psi}=\ket{\phi}$, and all intermediate states $\ket{\psi_i} = U_i\dots U_2U_1\ket{\psi}$ satisfy $\bra{\psi_i}H\ket{\psi_i}\leq \aaa$.
$\ket{\psi}$ and $\ket{\phi}$ are defined in equations \ref{Eq:Start_State} and \ref{Eq:Final_State} respectively.
\end{lemma}
\begin{proof}
 By definition of TI-standard, since the optimal \QCMAEXP proof is a {classical} string of size $\poly(N)$, there exists a $\poly(N)$-length sequence $U'=U_L\cdots U_1$ of $L$ $1$-and $2$-qudit unitaries (acting on $\bigotimes_{i=1}^N \sA_i$) which prepares a low energy state $\lowstate$ of $H'$. The circuit $U$ of the claim now acts as follows (see Figure~\ref{fig:switch} for an explicit example when $N=4$).
\begin{enumerate}
    \item \emph{Prepare low energy state.} Compute $U'_{\sA}\otimes I_{\sB}\ket{\psi}_{\sA,\sB}$, i.e.\ perform the mapping
        \[
            \ket{\psi}_{\sA,\sB}\mapsto \lowstate_{\sA}\ket{0\cdots 0}_{\sB}.
        \]
    \item \emph{``Warm up''.} Let $F\subseteq [N-1]$ denote the set of indices $i$ for which $\lowstatebra H'_{i,i+1}\lowstate<0$, which is efficiently computable by definition of TI-standard. One at a time, map $\ket{0}_{\sB_i}\mapsto\ket{1}_{\sB_i}$ for each $i\in F$, in any order.
    \item \emph{``Full blast''.} One at a time, map $\ket{0}_{\sB_i}\mapsto\ket{1}_{\sB_i}$ for all $i\in [N]\setminus F$, in any order.
    \item \emph{``Left deke''.} Map $\ket{1}_{\sB_i}\mapsto\ket{2}_{\sB_i}$ for all $i$ in sequence $(N,\ldots,1)$ (i.e.\ right to left).
    \item \emph{``Right deke''.} Map $\ket{2}_{\sB_i}\mapsto\ket{3}_{\sB_i}$ for all $i$ in sequence $(1,\ldots,N)$ (i.e.\ left to right).
    \item \emph{``Left deke''.} Map $\ket{3}_{\sB_i}\mapsto\ket{4}_{\sB_i}$ for all $i$ in sequence $(N,\ldots,1)$ (i.e.\ right to left).
    \item \emph{``Right deke''.} Map $\ket{4}_{\sB_i}\mapsto\ket{5}_{\sB_i}$ for all $i$ in sequence $(1,\ldots,N)$ (i.e.\ left to right).
    \item \emph{``Cool down''.} One at a time, map $\ket{5}_{\sB_i}\mapsto\ket{6}_{\sB_i}$ for all $i\in [N]\setminus F$, in any order.
    \item \emph{``Complete shut down''.} One at a time, map $\ket{5}_{\sB_i}\mapsto\ket{6}_{\sB_i}$ for each $i\in F$, in any order.
    \item \emph{Uncompute low energy state.} Apply $(U')^\dagger_{\sA}\otimes I_{\sB}$ to our state.
\end{enumerate}

\begin{figure}[t]
\[
\begin{array}{l|l}
  \text{Phase}& \text{Switch register contents} \\
    \hline
  \text{Warm up} &0000 \mapsto  1000 \mapsto 1010\\
  \text{Full blast} & 1010 \mapsto  1110 \mapsto 1111  \\
  \text{Left deke} & 1111 \mapsto 1112 \mapsto 1122\mapsto 1222\mapsto 2222\\
  \text{Right deke} & 2222 \mapsto 3222 \mapsto 3322 \mapsto 3332 \mapsto 3333\\
  \text{Left deke} & 3333 \mapsto 3334 \mapsto 3344 \mapsto 3444 \mapsto 4444\\
  \text{Right deke} & 4444 \mapsto  5444\mapsto 5544\mapsto 5554 \mapsto 5555  \\
  \text{Cool down} & 5555 \mapsto  5556  \mapsto  5656 \\
  \text{Complete shutdown} & 5656  \mapsto  5666 \mapsto 6666\\
\end{array}
\]
\caption{An honest prover's sequence of states in $\sB$, organized by phase, and for $N=4$. In this example, for concreteness we assume $F=\set{1,3}$.}
\label{fig:switch}
\end{figure}

\noindent To see why this works, we analyze each step above, respectively:
\begin{enumerate}
    \item \emph{Prepare low energy state.} For $t\in\set{0,\ldots, L}$, the $t$th intermediate state is
        \[
            \ket{\psi_t}=(U_t\cdots U_1)_{\sA}\otimes I_{\sB}\ket{\psi}=(U_t\cdots U_1\ket{0\cdots 0})_{\sA} \ket{0\cdots 0}_{\sB}.
        \]
        For all $t\in \set{0,\ldots, L}$, since $\sB$ is set to all zeroes, we have
        $
            \bra{\psi_t}H\ket{\psi_t}=0\leq\aaa.
        $
    \item \emph{``Warm up''.} This step addresses Obstacle 3 by switching on the terms of $H'$ indexed by $F$ first. Formally, after this step we have
        \[
            \ket{\psi_{L+\abs{F}}}=\lowstate_{\sA}\ket{1\cdots 1}_{\sB_F}\ket{0\cdots 0}_{\sB_{[N]\setminus F}}.
        \]
        Moreover, since all local terms indexed by $F$ have negative energy against $\lowstate$, throughout this phase (i.e.\ for all $i\in\set{1,\ldots, \abs{F}}$) we have  $\bra{\psi_{L+i}}H\ket{\psi_{L+i}}\leq 0\leq \aaa$.

    \item \emph{``Full blast''.} This step switches on all remaining terms of $H'$. After this step, all of $H'$ is activated, since we have state
        \[
            \ket{\psi_{L+\abs{F}}}=\lowstate_{\sA}\ket{1\cdots 1}_{\sB}.
        \]
        Crucially, since in the warm-up phase we already switched on all terms indexed by $F$, we know the sequence of energies seen in the ``full blast'' phase is monotonically increasing (since all terms not indexed by $F$ yield non-negative energy against $\lowstate$), i.e.
        \[
            \bra{\psi_{L+\abs{F}+1}}H\ket{\psi_{L+\abs{F}+1}}\leq 
            \cdots\leq \bra{\psi_{L+N}}H\ket{\psi_{L+N}}=\lowstatebra H'\lowstate\leq\aaa.
        \]
    \item \emph{``Left deke\footnote{According to \url{https://en.wikipedia.org/wiki/Deke_(ice_hockey)}: ``A deke feint or fake is an ice hockey technique whereby a player draws an opposing player out of position or is used to skate by an opponent while maintaining possession and control of the puck. The term is a Canadianism formed by abbreviating the word decoy.''}'', ``right deke''}. Throughout these phases, all of $H'$ remains continuously activated (i.e.\ energies obtained do not change during these phases). Thus, they do not affect the completeness analysis. (These phases are used in the soundness analysis to ensure that switch qubits are ``trapped'' being set to digits from $\set{2,3,4}$.)

    \item \emph{``Cool down''.} This phase begins the shut down process, returning us to an analogous point to the end of the warm-up phase. By the same argument as the full-blast phase, all energies seen are upper bounded by $\aaa$, and the state at the end of this phase is
        \[
            \ket{\psi_{L+6N-\abs{F}}}=\lowstate_{\sA}\ket{5\cdots 5}_{\sB_F}\ket{6\cdots 6}_{\sB_{[N]\setminus F}}.
        \]
    \item \emph{``Complete shutdown''.} This phase shuts off all remaining Hamiltonian terms, returning us to:
        \[
            \ket{\psi_{L+7N}}=\lowstate_{\sA}\ket{6\cdots 6}_{\sB}.
        \]
        By the same argument as the warm-up phase, all energies seen are at most $\eta_1$.
    \item \emph{Uncompute low energy state.} We are now back to the null space of $H$, and can uncompute $\lowstate$, so that
        \[
            \ket{\psi_{2L+7N}}=\ket{0\cdots 0}_{\sA} \ket{6\cdots 6}_{\sB}=\ket{\phi},
        \]
        i.e.\ we have reached our desired target state. Since $m=2L+7N$ and $\bra{\psi_i}H\ket{\psi_i}\leq \alpha =:\aaa$ for all $i\in [m]$, this completes the completeness analysis and proves \cref{Lemma:Completeness_Analysis}.
\end{enumerate}
\end{proof}

\paragraph{Soundness.}
Suppose $x$ is a NO instance. The following lemma shows that any short path from $\ket{\psi}$ to $\ket{\phi}$ must leave the low-energy subspace, as desired.
\begin{lemma}\label{lem:GSCONsound}
	Using the notation of \cref{Def:Lifted_Hamiltonian}, let $\lmin(H')\geq \beta$, and fix any $\p\in\set{2,\ldots, N-1}$. Consider any sequence $U=U_m\cdots U_1$ of $\p$-local unitary operators acting on $\bigotimes_{i=1}^{N}\sA_i\otimes\sB_i$. Then, either there exists $i\in [{m}]$ such that intermediate state ${\ket{\psi_i}\coloneqq U_i\cdots U_2U_1\ket{\psi}}$ satisfies $\bra{\psi_i}{H}\ket{\psi_i}\geq \frac{1}{2}\left(\frac{\beta}{4m^{2}}\right) = {\bbb}$, or $\enorm{ \ket{\psi_m}-\ket{\phi}} \geq 1/2 = {\ddd}$.
	$\ket{\psi}$ and $\ket{\phi}$ are defined in equations \ref{Eq:Start_State} and \ref{Eq:Final_State} respectively.
\end{lemma}
\begin{proof}
Assume, for sake of contradiction, that $\enorm{\ket{\psi_m}-\ket{\phi}}<1/2$, and that $\bra{\psi_i}H\ket{\psi_i}< \bbb$ for all $i\in[m]$. Define $\p$-orthogonal subspaces
\begin{eqnarray*}
    S_{012} &=& I_{\sA,\sB_{1,\ldots,N-\p-1}}\otimes\Span\left(\set{\ket{s}_{\sB_{N-\p,N}}\mid s\in\set{0,1,2}^{\p+1}}\right)\\
    S_{456} &=& I_{\sA,\sB_{1,\ldots,N-\p-1}}\otimes\Span\left(\set{\ket{s}_{\sB_{N-\p,N}}\mid s\in\set{4,5,6}^{\p+1}}\right),
\end{eqnarray*}
and recall that we set
\begin{equation}\label{eqn:temp}
    H_{i,i+1}\coloneqq H'_{i,i+1}\otimes(\ketbra{1}{1}+\ketbra{2}{2}+\ketbra{3}{3}+\ketbra{4}{4}+\ketbra{5}{5})_{\sB_i}+\Delta E_{\sB_i,\sB_{i+1}}
\end{equation}
for $E$ the projector forbidding the $2$-local substrings depicted in Figure~\ref{fig:neighbour}. We first show that, for sufficiently large $\Delta$, $\ket{\psi_i}$ has almost all its amplitude on a state in the null space of $H_E\coloneqq \Delta\sum_{i=1}^N E_{\sB_i,\sB_{i+1}}$.
\begin{lemma}\label{Lemma:psi_i_Form}
    Assume $\bra{\psi_i}H\ket{\psi_i}< \bbb$ for all $i\in[m]$. Then, there exists $i\in[m]$ such that $\ket{\psi_i}$ can be written as $\ket{\psi_i}=\gamma_1\ket{\gamma_1}+\gamma_2\ket{\gamma_2}$, with $\braket{\gamma_1}{\gamma_2}=0$, $\ket{\gamma_1}\in\Null{H_E}$, and
\begin{eqnarray}
    \trnorm{\ketbra{\psi_i}{\psi_i}-\ketbra{\gamma_1}{\gamma_1}}&<&2\sqrt{\frac{\eta_2}{\Delta}},\label{eqn:trbound}\\
    \bra{\gamma_1}I-\Pi_{S_{012}}-\Pi_{S_{456}}\ket{\gamma_1}&>& \frac{1}{4m^2}-2\sqrt{\frac{\eta_2}{\Delta}}\label{eqn:LB2}.
\end{eqnarray}
\end{lemma}
\noindent The proof of \cref{Lemma:psi_i_Form} is deferred to the end of this section below. From \cref{Lemma:psi_i_Form}, we draw the following conclusions:
    \begin{enumerate}
        \item Recalling that $S_{012}$ ($S_{456}$) projects onto $\set{0,1,2}^*$ ($\set{4,5,6}^*$) in $\sB_{N-\p,\ldots, N}$, respectively, Equation~(\ref{eqn:LB2}) implies $\ket{\gamma_1}=\chi_1\ket{\chi_1}+\chi_2\ket{\chi_2}$ for orthonormal vectors $\set{\ket{\chi_1},\ket{\chi_2}}$,
    $
        \abs{\chi_1}^2> (4m^{-2})-2\sqrt{\eta_2/\Delta},
    $
    such that $\ket{\chi_1}$ has registers $\sB_{N-\p,\ldots,N}$ supported solely on the intersection of two sets:
    \begin{itemize}
        \item All strings in the null space of $H_E$ (since $\ket{\gamma_1}\in\Null{H_E}$ implies $\ket{\chi_1}\in\Null{H_E}$), and
        \item all strings in the null space of $I-\Pi_{S_{012}}-\Pi_{S_{456}}$ (by Equation~(\ref{eqn:LB2})).
    \end{itemize}
    But the intersection of these two sets has precisely the regular expression
            \begin{equation}\label{eqn:4}
        33^*(2^*\cup 4^*),
            \end{equation}
    where the first $3$ is located in $\mathcal{B}_{N-b}$. This follows since by \cref{fig:neighbour}, a $3$ can only have a $2$, $3$, or $4$ to its right, and once we put down a $2$ to the right (resp. $4$), we can only put down more $2$'s (resp. $4$'s).\\

    \vspace{-2mm}
    Similarly, $\ket{\chi_2}$ is supported in registers $\sB_{N-\p,\ldots,N}$ solely on the span of strings from set $\set{0,1,2}^{b+1}\cup \set{4,5,6}^{\p+1}$ (note \cref{fig:neighbour} disallows a digit from set $\set{0,1,2}$ to be neighbours with a digit from $\set{4,5,6}$). Thus, $\ket{\chi_1}$ and $\ket{\chi_2}$ are orthogonal on the last $\p+1$ switch qudits (since the former must have a $\ket{3}$ on these qudits, but the latter cannot).

        \item Again since $\ket{\chi_1}\in\Null{H_E}$, combining Equation~(\ref{eqn:4}) with Figure~\ref{fig:neighbour} now implies, in fact, that \emph{all} switch qudits ``to the left'' of $\sB_{N-\p}$ are also set to $\ket{3}$, i.e.
    \[
        \ket{\chi_1}=\ket{3\cdots 3}_{\sB_{1,\ldots, N-\p}}\otimes \ket{\chi_1'}_{\sA,\sB_{N-\p+1,\ldots, N}},
    \]
   for some unit vector $\ket{\chi_1'}$. Together with \cref{eqn:4}, this implies the entire register $\sB$ of $\ket{\chi_1}$ is supported only on symbols from $\set{2,3,4}$.

   \item  Since all of $\sB$ is now supported on symbols from $\set{2,3,4}$, it follows from Equation~(\ref{eqn:temp}) that \emph{all} terms of $H'$ are switched on (thus resolving Obstacle 2). Hence,
       \begin{eqnarray}
            \bra{\gamma_1}H\ket{\gamma_1}&=&\bra{\gamma_1}\sum_{i=1}^{N-1} H'_{i,i+1}\otimes(\ketbra{1}{1}+\ketbra{2}{2}+\ketbra{3}{3}+\ketbra{4}{4}+\ketbra{5}{5})_{\sB_i}\ket{\gamma_1}\nonumber\\
            &>&\left(\frac{1}{4m^{2}}-2\sqrt{\frac{\eta_2}{\Delta}}\right)\bra{\chi_1}H\ket{\chi_1}\nonumber\\
            &=&\left(\frac{1}{4m^{2}}-2\sqrt{\frac{\eta_2}{\Delta}}\right)\trace(H'\trace_{\sB}(\ketbra{\chi_1}{\chi_1}))\nonumber\\
            &\geq&\left(\frac{1}{4m^{2}}-2\sqrt{\frac{\eta_2}{\Delta}}\right)\beta\label{eqn:betabound},
       \end{eqnarray}
       where the first statement follows since $\ket{\gamma_1}\in\Null{H_E}$, the second since (1) $H'\succeq 0$ and (2) since
       \[
        \bra{\chi_1}H'_{i,i+1}\otimes(\ketbra{1}{1}+\ketbra{2}{2}+\ketbra{3}{3}+\ketbra{4}{4}+\ketbra{5}{5})_{\sB_i}\ket{\chi_2}=0,
       \]
       since $\ket{\chi_1}$ and $\ket{\chi_2}$ are orthogonal on the last $\p+1$ switch qudits (even when projected down onto $\Span(\ket{1},\ket{2},\ket{3},\ket{4},\ket{5})$), the third statement since all of register $\spa{B}$ is supported on symbols $\set{2,3,4}$, and the last statement since $\lmin(H')\geq \beta$ by assumption.
   \end{enumerate}
We conclude that $\ket{\gamma_1}$ is high energy against $H$. We now show a similar result for $\ket{\psi_i}$, giving the desired contradiction. To do so, we follow the proof of the Projection Lemma of ~\cite{Kempe2006}. For brevity, define
$
    H_1\coloneqq \sum_{i=1}^{N-1}H'_{i,i+1}\otimes(\ketbra{1}{1}+\ketbra{2}{2}+\ketbra{3}{3}+\ketbra{4}{4}+\ketbra{5}{5})_{\sB_i}
$
so that $H=H_1+H_E$. Then, for $\Delta>2\snorm{H'}=2\snorm{H_1}$, recalling that $H_E\ket{\gamma_1}=0$,
	\begin{eqnarray*}				\bra{\psi_i}H\ket{\psi_i}&\geq&\left[(1-\abs{\gamma_2}^2)\bra{\gamma_1}H_1\ket{\gamma_1}+2\textup{Re}(\gamma_1\gamma_2\bra{\gamma_1}H_1\ket{\gamma_2})+\abs{\gamma_2}^2\bra{\gamma_2}H_1\ket{\gamma_2}\right]+\Delta\abs{\gamma_2}^2\\ &\geq&\bra{\gamma_1}H_1\ket{\gamma_1}+(\Delta-2\snorm{H_1})\abs{\gamma_2}^2-2\snorm{H_1}\abs{\gamma_2}\\
&>&\bra{\gamma_1}H\ket{\gamma_1}-2\snorm{H_1}\sqrt{\frac{\eta_2}{\Delta}}\\
&>&\frac{1}{4m^{2}}\beta-2\sqrt{\frac{\eta_2}{\Delta}}\left(\beta+\snorm{H'}\right),
	\end{eqnarray*}
    where the first statement follows since $\abs{\gamma_1}^2+\abs{\gamma_2}^2=1$ and $H_E\ket{\gamma_1}=0$, the second since $\abs{\gamma_1}\leq 1$, the third when $\Delta>2\snorm{H_1}$ and since $\abs{\gamma_2}^2<\eta_2/\Delta$, and the last by Equation~(\ref{eqn:betabound}) and since $\snorm{H_1}=\snorm{H'}$. Crucially, note that $H'$ is independent of $\Delta$ (recall $H'$ is the TI-standard Hamiltonian we have plugged in as a black-box). Thus, we may set $\Delta$ to a sufficiently large fixed polynomial in $N$ so that $ \bra{\psi_i}H\ket{\psi_i}> \frac{1}{2}\left(\frac{\beta}{4m^{2}}\right)=\bbb$. This yields the desired contradiction.
\end{proof}

To finish the proof of \cref{lem:GSCONsound}, we give the proof of \cref{Lemma:psi_i_Form} below.
\begin{proof}[Proof of \cref{Lemma:psi_i_Form}]
    Since $\ket{\psi}\in S_{012}$ and $\ket{\phi}\in S_{456}$, the Traversal Lemma (Lemma~\ref{l:traversal}) says there exists $i\in[m]$ such that
    \begin{equation}\label{eqn:LB}
        \bra{\psi_i}I-\Pi_{S_{012}}-\Pi_{S_{456}}\ket{\psi_i}\geq \frac{1}{4m^2}.
    \end{equation}
	Now, since the $E_{\sB_i,\sB_{i+1}}$ are pairwise commuting projectors, $H_E\succeq \Delta\Pi_E$ for $\Pi_E$ the projector onto the orthogonal complement of $\Null{H_E}$.  Then,
	\[
	\bbb> \bra{\psi_i}H\ket{\psi_i}\geq\Delta\bra{\psi_i}H_E\ket{\psi_i}\geq\Delta\bra{\psi_i}\Pi_E\ket{\psi_i},
	\]
	where the second inequality follows since $H'\succeq 0$ (specifically, in the NO case, $\lmin(H')\geq \beta\geq 0$ by definition of TI-standard), and the last inequality since $H_E\succeq \Pi_E$. We conclude $\ket{\psi_i}$ can be written as $\ket{\psi_i}=\gamma_1\ket{\gamma_1}+\gamma_2\ket{\gamma_2}$, for $\ket{\gamma_1}\in\Null{H_E}$ (note $\ket{\gamma_1}$ lives in the full $N$-qudit space, not just the switch space $\spa{B}$), $\set{\ket{\gamma_1},\ket{\gamma_2}}$ orthonormal vectors, and $\abs{\gamma_1}^2> 1-(\bbb/\Delta)$. Equation~(\ref{eqn:trbound}) of the claim now follows by Equation~(\ref{eqn:enorm}). Combining Equation~(\ref{eqn:trbound}) with Equation~(\ref{eqn:LB}) and the H\"{o}lder inequality yields Equation~(\ref{eqn:LB2}) of the claim.
\end{proof}

\subsection{Proof of \QCMAEXP-completeness}\label{sscn:GI}

With the Lifting Lemma for GSCON (\cref{lem:liftGSCON}) in hand, we are ready to obtain our \QCMAEXP-completeness result by plugging in (with a minor modification) the Gottesman-Irani (GI) 1D TI construction~\cite{Gottesman2009}. For convenience, we reproduce the statement of \cref{thm:TIGSCON} below.

\begin{reptheorem}{thm:TIGSCON}
    \TIGSCON is \QCMAEXP-complete for 1D, nearest neighbour, translationally invariant Hamiltonians
on N qudits, for $m\in\poly(N)$, $\delta \in \Theta(1/\poly(N))$, and any $\p\in\set{2,\ldots, N-1}$.
\end{reptheorem}
\begin{proof}
Containment of \TIGSCON in \QCMAEXP for $\delta\in\Omega(1/\poly(N))$ follows immediately since $\GSCON\in\QCMA$ for any interaction graph~\cite{GS14}. Thus, it remains to show \QCMAEXP-hardness of \TIGSCON, which we aim to obtain by plugging GI into \cref{lem:liftGSCON}.

Unfortunately, in its original form GI is not TI-standard, as it does not satisfy the requirement $\beta\geq 16(2L+7N)\alpha\geq 0$. This is because GI's promise thresholds are as follows: there exist polynomials $p,q:\N\mapsto\N$, such that in the YES (NO) case, the GI Hamiltonian has roughly energy at most $\alpha=p(N)$ (at least $\beta=p(N)+1/q(N)$). Thus, although $\abs{\alpha-\beta}\geq 1/\poly(N)$, the large additive offset $p(N)$ violates the bound $\beta\geq 16(2L+7N)\alpha\geq 0$. In our setting, this is a problem, since the Traversal Lemma now only allows us to conclude energy at most $p(N)$ in the YES case, and at least $(p(N)+1/q(N))/(m^2)$ in the NO case. This is not sufficient to obtain a non-trivial promise gap.

This is not difficult to address, but requires us to explain why the large $p(N)$ offset arises in GI. Roughly, to ensure computations in its ground space are syntactically well-formed, GI wishes to force all history states to start and end with bracket symbols, i.e.\ $\leftend$ and $\rightend$, respectively. However, a TI Hamiltonian cannot distinguish between the start of the chain versus (say) the middle. Thus, to force the desired bracket behavior, on {every qudit} in the chain GI places constraint $I-\ketbra{\leftend}{\leftend} - \ketbra{\rightend}{\rightend}$. (There is also a second penalty term which does not play a role here; see \cref{app:altfix2}.) As a history state of GI has precisely one occurrence of each of $\leftend$ and $\rightend$, it is penalized by $N-2$; this is the offset $p(N)$ above.

The natural solution is to replace $H_{i,i+1}$ from GI with
\begin{equation}~\label{eqn:shift}
    H_{i,i+1}-\left(1-\frac{2}{N}\right)I.
\end{equation}
This additively shifts the spectrum of $\sum_{i=1}^{N-1} H_{i,i+1}$, so that the resulting Hamiltonian has ground state energy between $0$ and $\epsilon/N^2$ in the YES case, and at least  $(1-\epsilon)/N^2$ in the NO case; see Sections 5.8.1 and 5.9 of~\cite{Gottesman2009}. (Here, $\epsilon$ is the probability that the verifier outputs the wrong answer on a given proof.) By using standard error reduction on the \QCMAEXP verifier we plug into GI, we can make $\epsilon$ a sufficiently small fixed inverse polynomial in $N$, so that the $\beta\geq 16(2L+7N)\alpha\geq 0$ condition is satisfied.

Before concluding, we have one last check to make---due to the additive shift of Equation~(\ref{eqn:shift}), even though $H\succeq 0$, we may now have $H_{i,i+1}\not\succeq 0$ (whereas before the shift, \cite{Gottesman2009} guaranteed $H_{i,i+1}\succeq 0$ for all $i$). Thus, we must ensure that all ``negative-energy-contributing'' terms indexed by $F$ (\cref{def:TIstd}) can be identified in time $\poly(N)$. This is indeed the case, since the history state $\lowstate$ of ~\cite{Gottesman2009} satisfies:
\begin{eqnarray}
      \lowstatebra H_{1,2}\lowstate &=&-1+\frac{2}{N}\label{eqn:energies1}\\
      \lowstatebra H_{i,i+1}\lowstate&=&\frac{2}{N}\qquad\qquad\text{ for }i\in\set{2,\ldots, N-2}\label{eqn:energies2}\\
      \lowstatebra H_{N-1,N}\lowstate&\leq&\left(-1+\frac{2}{N}\right)+\frac{\epsilon}{N^2}.\label{eqn:energies3}
\end{eqnarray}
Thus, we may set $F=\set{1,N-1}$.
\end{proof}

\section{Outlook and Discussion}\label{sec:discussion}

In this paper, we have developed generic frameworks for lifting circuit-to-Hamiltonian mappings to hardness results for APX-SIM and GSCON.

For APX-SIM, this not only simplifies the existing \PQMAlog-hardness proofs of APX-SIM \cite{Ambainis2013,GY16_2,GPY20}, but also allows us to obtain new hardness results, most notably for 1D translationally invariant systems: \PQMAexp-completeness and \PSPACE-completeness for inverse polynomial and inverse exponential precision (with respect to the length of the chain), respectively. The genericness of the lifting approach allows us to apply comparatively low-dimensional constructions of translationally invariant constructions, such as \cite{Bausch2016}, and thus prove hardness of these simpler Hamiltonians. We envision these lifting results being useful for further work on the topic in a similar manner, i.e.\ mapping hardness of the Local Hamiltonian problem to hardness of APX-SIM. Furthermore, we hope that the techniques behind our lifting approach, such as ``compressing'' multiple queries into a single qubit (Lemmas \ref{lem:lowhelp} and \ref{lem:augmented-circuit}), may find uses in further work in dealing with quantum oracle classes. Notably in this regard, \cref{lem:lowhelp} provides a tool to convert local deviations in proof optimality to rigorous analytical lower bounds on energy penalties.

For GSCON, we have given a similar generic lifting framework, albeit not as broad as that for APX-SIM (i.e.\ for GSCON, we restricted to general 1D translationally-invariant circuit-to-Hamiltonian mappings). Here, we applied our lifting results to show \QCMAEXP-completeness of GSCON on translationally invariant 1D Hamiltonians. Again, it would be interesting to know whether the techniques here, such as the ``space-time''-like local clock construction and ``logically protected'' target switch subspace design, find uses elsewhere. Taken together, these ideas allowed us to achieve a remarkable level of robustness for our construction, which remains sound even if an adversary acts on all but one qudit per time step.

\paragraph{Further Work.} A natural open question is to generalise our APX-SIM results to yet more physical Hamiltonians. This could be done using Hamiltonian simulation techniques, as utilised in \cite{GPY20}, however, current Hamiltonian simulation techniques for translationally invariant systems unfortunately do not (as of time of writing) yield more natural classes of Hamiltonians \cite{Kohler2020}.

With regards to TI-GSCON, an immediate question is whether one can extend our lifting results to the more general case of ``local circuit-to-Hamiltonian constructions'' (\cref{def:local-mapping}), as was done for APX-SIM. Currently, our TI-GSCON lifting framework restricts to 1D TI circuit-to-Hamiltonian constructions. An additional open question is if there exist non-trivial classes of Hamiltonians for which GSCON is easy. Given the limitations on Hamiltonian structure that the translational invariance places on the Hamiltonian, it may be easier to find such a class for translationally invariant Hamiltonians.

A higher-level goal is to ``unify'' the APX-SIM and GSCON complexities with the hardness of the Local Hamiltonian problem, such that the completeness of the latter for a family of Hamiltonians (relative to, e.g., QMA, QCMA etc) immediately implies completeness of APX-SIM, GSCON, and potentially other low energy properties for the appropriate related complexity classes.
Furthermore, the relation between the complexity of determining low energy properties and universality properties for classes Hamiltonians should be investigated.

Finally, a minor point to be addressed is how to make the lifting construction generalise to complexity classes for which we are unable to run error reduction: e.g.\ StoqMA (see \cref{sec:hamiltonian-embedding} for a discussion). The complexity of APX-SIM for Hamiltonian for which LH is StoqMA-complete is known to be $\PP^{\class{StoqMA}[\log]}$-complete, but we cannot simply include this at present as a consequence of our lifting construction.

\section*{Acknowledgements}
SG thanks Dorian Rudolph for helpful discussions, and acknowledges funding from DFG grant 432788384.
J.\,D.\,W.\ is supported by the EPSRC Centre for Doctoral Training in Delivering Quantum Technologies [EP/L015242/1].
J.\,B.\ is grateful for support from the Draper's Research Fellowship at Pembroke College.

\printbibliography

\appendix

\section{Alternate fix: Resolving partial activation via the ``self-organizing'' property of GI}\label{app:altfix2}

While we use a more general, black-box method in \cref{sscn:o3} to resolve the partial activation problem in our GSCON construction, one could also take a white-box approach by exploiting a remarkable property of the Gottesman-Irani (GI) construction~\cite{Gottesman2009}---the fact that GI is (what we call) ``self-organizing''. Since this latter property is interesting in its own right, we flesh out this white-box approach here for the interested reader.

  As discussed in \cref{sscn:GI}, in GI all strings in the ground space must be bracketed, i.e.\ start with $\ket{\leftend}$ and end with $\ket{\rightend}$. To enforce this in a TI fashion, GI combines two penalty terms (placed on each pair of neigboring qudits): (1) $I-\ketbra{\leftend}{\leftend} - \ketbra{\rightend}{\rightend}$ (ensures at least one copy of $\ket{\leftend}$ and $\ket{\rightend}$ exist), and (2) $\ketbra{s}{s}\otimes\ketbra{\leftend}{\leftend} + \ketbra{\rightend}{\rightend}\otimes \ketbra{s}{s}$ for all symbols $s$ (ensures $\ket{\leftend}$ and $\ket{\rightend}$ only appear at the left and right ends of the chain, respectively). Thus, the ends of the chain are not ``explicitly'' set; rather, one can prematurely cut the chain anywhere desired, and the ground space will ``self-organize'' to be ``properly-formed'' even for the truncated chain (c.f. Section 5.8.2 of~\cite{Gottesman2009}, where a similar effect is exploited to handle the 1D case with \emph{closed} boundary conditions). This means that although the cheating strategy for Obstacle 2 only switches on the first $N-\p$ terms $H'_{i,i+1}$, the resulting effective Hamiltonian $\sum_{i=1}^{N-\p}H'_{i,i+1}$ still has a ``properly-formed'' ground space encoding \emph{some} quantum computation.

To understand \emph{which} computation a truncated chain of length $N-\p$ encodes, use the fact that the ground space of GI encodes two Turing machines (TMs): a classical ``counter'' TM $M_1$, and a quantum TM (QTM) $M_2$ running the \QMAEXP verification. Suppose one wishes to consider input $x\in\set{0,1}^*$ to a \QMAEXP promise problem. GI first ``runs'' $M_1$, which computes a function $f:\N\mapsto\set{0,1}^*$ with the following properties:
 \begin{enumerate}
    \item After running $M_1$ for $T\in\N$ steps (for all $T\geq T_0$ for some fixed $T_0\in\N$), $M_1$ has precisely string $f(T)$ on its work tape (i.e.\ with the rest of the tape blank).
    \item For any $x\in\set{0,1}^*$, one can compute in time polynomial in the length of $x$ a preimage $T\in\N$ such that $f(T)=x$.
 \end{enumerate}
 Once $M_1$ halts, $M_2$ starts by reading input $x$ from $M_1$'s tape, and then runs the \QMAEXP verification.

 In~\cite{Gottesman2009}, a chain of length $N\in\N$ results in simulating $M_1$ and $M_2$ for $T=N-2$ steps each\footnote{The $-2$ in $T=N-2$ arises because the first and last qudits in a ground state are occupied by the opening and closing brackets, $\ket{\leftend}$ and $\ket{\rightend}$.}. This suffices (assuming without loss of generality $T\geq T_0$) for $M_1$ to correctly compute input $x$ to $M_2$, and for $M_2$ to subsequently complete its verification and halt. So, consider any desired \QMAEXP input $x\in\set{0,1}^*$ with preimage $T$ satisfying $f(T)=x$. Obstacle 2 can also be resolved via two changes to our setup:
 \begin{enumerate}
    \item Set the length of the chain in our construction of $H$ to $N=(T+2)+\p$.\\

    \vspace{-4mm}
        \emph{Intuition:}  This ensures that if a cheating adversary switches on the first $t\geq N-\p$ Hamiltonian terms, the effective chain length is {at least} $N-\p=T+2$. If, in addition, $t=N-\p$, this suffices to overcome Obstacle 2, since then $M_1$ is run for precisely the right number of steps to produce input $x$, and $M_2$ has sufficient time to verify $x$ and halt.

    \item ``Pad'' the action of TM $M_1$, so that after each of its steps (including its last one), we insert $\p$ ``idling'' steps, meaning steps in which $M_1$ does not change its tape contents. (Thus, if $M_1$ used to take $T$ steps to compute a string $x$, with the ``padding'' it will take $T\p+T$ steps to compute $x$, which is a constant factor increase if $\p\in \BigO(1)$.) Call the new TM $M_1'$ with corresponding function $f'$ ($f'$ is to $M_1'$ as $f$ is to $M_1$). Redefine $N\coloneqq T(\p+1)+2$.\\

    \vspace{-4mm}
        \emph{Intuition:} Continuing with the intuition from the previous bullet point, if $t>N-\p$, we are not guaranteed that $f(t)=x$ for our desired input $x$, since $M_1$ is running for \emph{too many} steps. Ideally, we would like to ``idle'' $M_1$ once $t=N-\p$, so that the correct output $x$ on $M_1$'s tape is no longer changed. However, recall $N$ is not explicitly given to $M_1$ as input, but rather encoded implicitly as the length of the chain the Hamiltonian acts on. And the TI Hamiltonian cannot detect when it is ``within'' $\p$ sites of the right end of the chain so that it can ``signal $M_1$ to begin idling''. Thus, we instead do the padding uniformly, i.e.\ after \emph{each} step of $M_1$. The resulting machine $M_1'$ thus satisfies the property $f(t)=x$ \emph{for all} $t\in\set{T(\p+1)-\p,\ldots, T(\p+1)}$, as desired.
 \end{enumerate}

\end{document}